\documentclass[12pt]{article}%
\usepackage{amsfonts}
\usepackage{amssymb}
\usepackage{graphicx}
\usepackage{setspace}
\usepackage[round]{natbib}
\usepackage{amsmath}%
\usepackage{amsthm}%
\usepackage{bbm}
\usepackage{palatino}
\usepackage{paralist}
\usepackage{multirow}
\usepackage{booktabs}
\usepackage{dsfont}
\usepackage{indentfirst}
\usepackage{enumerate}
\usepackage{enumitem}
\usepackage{longtable}
\usepackage{bm}
\usepackage{wrapfig}
\usepackage{caption}
\usepackage{subcaption}
\usepackage{mwe}
\usepackage{float}
\usepackage{booktabs}
\usepackage{lscape}
\usepackage[hyperindex,breaklinks]{hyperref}
\hypersetup{colorlinks=true,       
	linkcolor=red,       
	citecolor=blue,        
	filecolor=magenta,      
	urlcolor=cyan           
}

\newcommand{\be}{\begin{equation}}
\newcommand{\ee}{\end{equation}}
\newcommand{\x}{\mathbf{x}}
\newcommand{\y}{\mathbf{y}}
\newcommand{\myt}{\mathbf{v}}
\newcommand{\myu}{\mathbf{u}}
\newcommand{\myg}{\mathbf{g}}
\newcommand{\mylambda}{\boldsymbol\lambda}
\newcommand{\mybeta}{\boldsymbol{\beta}}
\newcommand{\myepsilon}{\boldsymbol{\epsilon}}
\newcommand{\act}{\sigma}

\newtheorem{theorem}{Theorem}
\newtheorem{remark}{Remark}
\newtheorem{Assumption}{Assumption}
\theoremstyle{definition}
\newtheorem{definition}{Definition}[section]

\DeclareMathOperator*{\maxf}{maximize \quad}
\newcommand{\maxdisp}[3]{\begin{align*}\maxf_{#1} & #2\\\st  & #3\end{align*}}

\title{Deep Partial Least Squares for Empirical Asset Pricing\footnote{This paper was presented as a keynote talk at the R in Finance Conference, Chicago, 2022. The authors would like to thank Peter Carl (Balyasny Asset Management), Francois Cocquemas (Florida State University), and Kris Boudt (KU Leuven), for useful discussions and input on this paper.}~\footnote{The source code and data for this paper is available at \url{https://github.com/Kem975/Deep_Partial_Least_Squares}.}}

\author{
	\makebox[.4\linewidth]{Matthew Dixon\footnote{Email: matthew.dixon@iit.edu}}\\
	\textit{\small  Department of Applied Mathematics }\\
	\textit{\small  Illinois Institute of Technology}\\	
	\and
	\makebox[.4\linewidth]{Nicholas G. Polson}\\
	\textit{\small  Booth School of Business}\\
	\textit{\small  University of Chicago}\\
	\and
	\makebox[.4\linewidth]{Kemen Goicoechea}\\
	\textit{\small  Department of Applied Mathematics}\\
	\textit{\small   Illinois Institute of Technology}\\
}

\graphicspath{{../fig/}{../wine/}{fig/}}

\begin{document}
\maketitle

\begin{abstract}
\noindent 
We use deep partial least squares (DPLS) to estimate an asset pricing model for individual stock returns that exploits conditioning information in a flexible and dynamic way while attributing excess returns to a small set of statistical risk factors. The novel contribution is to resolve the non-linear factor structure, thus advancing the current paradigm of deep learning in empirical asset pricing which uses linear stochastic discount factors under an assumption of Gaussian asset returns and factors. This non-linear factor structure is extracted by using projected least squares to jointly project firm characteristics and asset returns on to a subspace of latent factors and using deep learning to learn the non-linear map from the factor loadings to the asset returns. The result of capturing this non-linear risk factor structure is to characterize anomalies in asset returns by both linear risk factor exposure and interaction effects. Thus the well known ability of deep learning to capture outliers, shed lights on the role of convexity and higher order terms in the latent factor structure on the factor risk premia. On the empirical side, we implement our  DPLS factor models and exhibit superior performance to LASSO and plain vanilla deep learning models. Furthermore, our network training times are significantly reduced due to the more parsimonious architecture of DPLS. Specifically, using 3290 assets in the Russell 1000 index over a period of December 1989 to January 2018, we assess our DPLS factor model and generate information ratios that are approximately 1.2x greater than deep learning. DPLS explains variation and pricing errors and identifies the most prominent latent factors and firm characteristics. 

\end{abstract}

\noindent {\bf Key Words:}  Deep Learning, Partial Least Squares, Shrinkage, Empirical Asset Pricing, Conditional Latent Factor Models, Non-Linear Risk.

\newpage
\topmargin=0in
\textwidth=6.6in
\section{Introduction}
In this paper, we develop a deep factor methodology known as deep partial least squares (DPLS) as a projection based dimensionality reduction technique for  empirical asset pricing regressions which, combined with deep learning, provides improved predictive performance on financial panel data while also providing information on the risk factor premia. Our DPLS approach is shown to be a natural, composable, and powerful approach to projecting both the predictors and the responses onto a relatively small set of orthonormal latent variables so as to maximize the covariance between the projected predictors and the projected responses.  In contrast to principal component dimension reduction, PLS \& DPLS projections are dependent on the response and thus leverage a wider dataset of firm charactersistcs for predicting returns. \cite{Polson2021} find evidence that DPLS shows superior predictive performance on non-financial cross-sectional data.
\cite{Dixon2020} show that deep learning identifies the most important interaction effects and factors.  
 DPLS builds on this by boosting the covariance between the response and predictor, while simultaneously using the same statistical loadings generated by PLS. 
Furthermore, because DPLS is based on deep learning, it is able to automatically capture interaction effects. 

Partial Least Squares (PLS) is a wide class of methods for modeling relations
between sets of observed variables by means of latent variables. It is comprised
of regression and classification tasks as well as dimension reduction techniques
and modeling tools. The underlying assumption of all PLS methods is that the
observed data is generated by a system or process which is driven by a smaller
number of statistical latent variables.

Across hedge funds, asset management, and proprietary trading firms, it is commonplace to use supervised learning for macroeconomic and financial forecasting by performing "feature engineering" --- expanding the original set of covariates to include terms such as interactions, dummy variables or nonlinear functionals of the original predictors \citep{Hoadley2000, Dobrev2013RobustFB,DBLP:journals/corr/DixonKB16}. It's also increasingly common to aggregate multiple data sources for financial forecasting, combining traditional financial data such as stock prices with novel signals such as news sentiment \citep{arxiv.2009.07947, ALGABA2021}. Our goal is to employ supervised dimension reduction techniques which, when combined with deep learning, not only leads to a more parsimonious architecture with improved out-of-sample performance, but enable factor models to be properly seated in an equilibrium asset pricing framework in which the statistical factors serve as latent risk factors and a non-linear risk factor structure is represented by the network architecture. Our approach corresponds to the use of a non-linear stochastic discount factor (see \cite{Cochrane2001}).

The rest of the paper is outlined as follows.  Section \ref{sect:relation} outlines the general connections with previous work and highlights the contribution of this paper. Section \ref{sect:factors} surveys the related works on factor modeling works in more detail culminating in the introduction of our non-linear factor model, best viewed as a non-linear generalization of the BARRA factor model.
Section \ref{sect:regression} describes our general modeling methodology using deep learning within partial least squares for multi-index regression.
Following \cite{Polson2021}, we provide a nonlinear extension of partial least squares by using deep layers to model the relationship between the scores generated in PLS. 
Then in Section \ref{sect:DPLS_factors}, we present the DPLS factor model and some of its variants. Section \ref{sect:results} demonstrates the application of our DPLS factor model to equity data consisting of monthly returns and firm characteristics. Finally, Section \ref{sect:summary} concludes with directions for future research.

\subsection{General Connection to Previous Work} \label{sect:relation}
Within the investment management industry, asset managers seek novel predictive firm characteristics to explain anomalies which are not captured by classical capital asset pricing and factor models. Recently a number of independent empirical studies, have shown the importance of using a higher number of, often highly correlated,  predictors related to firm characteristics and other common factors \citep{moritz2016, 10.1093/rfs/hhv059, Gu2018, Feng2018}.  In particular, \cite{Gu2018} analyze a dataset of more than 30,000  individual  stocks  over a 60 year period from 1957 to 2016, and determine over 900 correlated baseline signals.

Despite being widely used in the asset management industry, it is becoming increasingly well documented that ordinary linear regression is inadequate for high dimensional factor models.  The frequently encountered problem of nearly collinear regressors can be addressed using regularization to provide shrinkage estimation \citep{hahn_partial_2013, Dobrev2013RobustFB, 10.1093/restud/rdw005}. \cite{Dobrev2013RobustFB} develops a Regularized Reduced Rank Regression model, a special case of which generalizes principal component regression by applying reduced rank rather than linear regression to the principal components of the
regressors, thereby disentangling the forecasting factors driving the outcomes from the factor structure in the predictors.

\cite{moritz2016,Gu2018, pelgery2019} highlight the inadequacies of OLS regression in variable selection over high dimensional datasets --- in particular the inability to capture outliers. In contrast, deep neural network can explain more structure in stock returns because of their ability to fit flexible functional forms with many covariates, without apriori knowledge of the interaction effects. \cite{Feng2018} demonstrates the ability of a three-layer deep neural network to effectively predict asset returns from fundamental factors. \cite{GU2021429} apply auto-encoders to fit a latent factor conditional asset pricing model and demonstrate that the approach not only out-performs other factor models but is suitable for no-arbitrage pricing. However, techniques which project on statistical factors are difficult to economically interpret and have limited hedging utility.  

\cite{Dixon2020} demonstrate a three-layer deep neural network approach for fundamental factor modeling which not only outperforms OLS regression and LASSO but also provides interpretability --specifically, the importance of the predictors and interactions can be ranked and the uncertainty in predictor sensitivities can be estimated with bootstrap sampling techniques. \cite{Dixon2022} extend their methodology to a Bayesian estimation framework to provide uncertainty quantification and a predictive distribution of excess returns. 

There are a number of advantages of DPLS methodology over traditional deep learning. First, PLS uses linear methods to extract predictors and weight matrices. Due to a fundamental estimation result of \cite{brillinger_generalized_2012} we show that these estimates can provide a stacking block for our DPLS model. Moreover, traditional diagnostics (bi-plots, scree-plots) to diagnose the depth and nature of the activation functions allow us to find a parsimonious model architecture. \cite{frank_statistical_1993} provides a Bayesian shrinkage interpretation of PLS. This permits the discrimination between unsupervised and supervised learning method. In other words, based on the value of shrinkage factors, one can deduce whether a supervised learning method such as DPLS is necessary or whether an unsupervised dimension reduction method such as PCA is adequate.

One of the remarkable discoveries by \cite{Polson2021} in fitting DPLS regression models is that the PLS projection can be entirely separated from the training of the deep network, thus enabling standard deep neural networks algorithms, such as stochastic gradient descent, to be used without modification. Drawing on a seminal result by \cite{brillinger77} which paved the way for the landmark paper by \cite{brillinger_generalized_2012} and later \cite{Naik2000}, we formalize the critical observation that under some limiting assumptions,  the projection directions in PLS are invariant to the functional form of the regression between score matrices, up to a constant of proportionality. 


Our theoretical development parallels \cite{brillinger_generalized_2012}, namely that PLS regression is a strongly consistent estimator to the regression coefficient even when there exists a unknown non-linear relationship between the predictors and the responses. We invert this line of reasoning and instead show that the linear PLS estimate is strongly consistent with the score regression coefficients when the regression on scores is a (known) non-linear function. This unlocks PLS regression as a composable method -- non-linear functional approximations can be composed with the PLS projection matrices without violating the consistency result. More precisely, one can use still use the same projections and x-loadings estimated by PLS, under a linear relationship between the input and output scores, while subsequently using deep learning to improve the predictive performance of the y-scores.

In the context of empirical asset pricing, the DPLS method is shown to deliver a key point of departure from  \cite{Gu2018, pelgery2019, KPS2019} in enabling the modeling of a \emph{non-linear factor structure} which corresponds to an assumed \emph{non-linear stochastic discount factor}.  Furthermore, we develop interpretative techniques for DPLS which show both firm characteristic and latent factor sensitivities and interaction effects are provided by the DPLS factor model.

As a further contribution, we show how our DPLS factor models not only provide a non-linear generalization of BARRA fundamental factor models but also briefly discuss how they provide a non-linear, non-stationary extension of dynamic factor models \citep{16320}.

On the empirical side, we develop a no-arbitrage DPLS factor model with 50 factors for 3290 Russell 1000 indexed stocks over an approximate 30 year period and compare performance with LASSO and pure deep learning based fundamental factor models. Furthermore, we assess the risk-reward trade-off of our factor model and identify the implications of including a non-linear factor structure.

\section{Factor Modeling} \label{sect:factors}
Our goal is to explain the differences in the cross-section of returns for individual stocks using a  latent factor no-arbitrage asset pricing model,  whereby any expected gain in asset returns over the risk-free rate is solely attributed to factor risk exposure. But rather than this returns model being linear in factors, or equivalently, assuming linear dependence of a stochastic discount factor on the factors, we seek to introduce non-linearity. The rationale for non-linearity in factor modeling is well motivated in the literature and can be reasoned from Stein's lemma, that the non-linearity in the dependency of the returns on the factors arises precisely when the asset returns and factors are non-Gaussian.

While the practice of applying non-linear parametric functions, such as deep learning, to factor models is conceptually straightforward,  the challenge becomes in how to go beyond a mere mechanical exercise in learning a non-linear function of predictor variables that maximizes the out-
of-sample explanatory power for realized excess returns of stocks,  to characterizing the risk-return tradeoff of assets built from our factor model, under the equilibrium asset pricing principle that characteristic-based predictability are determined through risk exposures to factors alone \citep{KPS2019}.
Thus, in informal terms, in order for factor models to be correctly seated in the framework of equilibrium asset pricing, it must resolve the predictability attributed to the factor exposures and not conflate this with mis-pricing. Furthermore, they must be no-arbitrage factor models and hence provide no intercept when regressed against excess returns.

The need for both explanatory power from a set of predictor variables, taken to be lagged firm characteristics, and risk-return properties inevitably leads to model design trade-offs. For example some machine learning approaches to factor modeling yield strong predictive power but yield no insight into risk factor attribution and therefore have questionable economic interpretability and practical utility. On the other hand, other approaches deploy machine learning in such a way as to guarantee that the characteristic-based predictability is attained solely through factor risk exposures while accounting for non-linearity between the factor loadings and the firm characteristics. However, those approaches that do resolve the risk factors, still rely on a linear relationship between the returns and factor loadings, or equivalently the stochastic discount factor remains linear in the factors \citep{pelgery2019}.   

Then there are questions of how best to reduce the set of predictor variables to a smaller number of latent factors to manage the risk-reward tradeoff.  Such factors should ideally carry economic interpretability and be representable by factor mimicking portfolios.  We can broadly separate the main approaches to non-linear factor models into the following three categories:

\begin{enumerate}
\item \textbf{Forecasting models} which predict the excess returns directly as a non-linear function of the lagged firm characteristics. This essentially amounts to a pure exercise in prediction of excess returns given the lagged firm characteristics, without assessing the risk/return tradeoff in terms of the latent factors (see \cite{Gu2018, Dixon2020}).
\item \textbf{Linear Conditional Latent Factors} preserve a linear relationship between the excess returns and the factors -- amounting to the use of a linear stochastic discount factor in the factors -- but introducing the non-linearity between the conditional factor loadings and the lagged firm characteristics (see, for example, \cite{KPS2019, pelgery2019, GU2021429}).
These approaches allow for latent factors and time-varying loadings which ``\emph{can explain individual stock
returns and estimate the best performing factors without taking an a prior stand on what the factors are}'' \cite{KPS2019}. This category of approaches is thus suitable for asset pricing as it provides insight into the factor risk premia by conditioning on firm characteristics.
\item \textbf{Non-linear Conditional Latent Factors} model the excess returns directly as a non-linear function of the lagged firm characteristics, but allow for the estimation of the conditional factor loadings and the latent factors under a non-linear stochastic discount factor in the factors.  Hence, the approach is suitable for asset pricing and it operates under the less restrictive assumption that the stochastic discount factor is linear in the factors. Equivalently, as per Stein’s lemma, such an approach does not assume that the asset returns and factors are Gaussian.
\end{enumerate}


Our DPLS factor model bridges these first two categories to arrive at a model in the final category in its most general form. We stop short of claiming that the modeling approach recovers the second category of approaches, as there are nuanced issues in how the non-linearity is represented, the estimation of the latent factors, and the handling of the intercepts. 
Specifically, our approach is to use deep projected least squares regression to identify \emph{non-linear factor structure}, that is the asset returns are attributed to non-linear risk factors, with non-linear conditional factor loadings in the lagged firm characteristics, rather than linear risk factors under non-linear conditional factor loadings. Moreover, the approach is suitable for no-arbitrage pricing as the intercept captures the portion of the model returns not attributed to the risk factors. We thus build on several ideas which culminate from the extensive factor model literature, some of the most closely related works we now detail here.

\cite{Gu2018} present a static predictive model for the excess returns at time $t$, $\mathbf{r}^e_t:=[\mathbf{R}^e]_t$,  which learns a flexible parameterized function, $h_{\theta}(\cdot)$ of a set of predictors -- lagged firm characteristics, $\mathbf{z}_t$, observed at time $t$:

$$r^e_{i,t+1}=h_{\theta}(z_{i,t}) + \epsilon_{i,t+1}, ~i=1,\dots, N_t, ~ t =1,\dots, T,$$
where $\boldsymbol{\epsilon}_t:=[E]_t$ are the idiosyncratic zero-mean Gaussian errors.
\cite{Dixon2020} present a time dependent parameterization of this model which is trained period-by-period. The advantage of the time dependent model is that it needn't assume stationarity in the returns and firm characteristics. Rather, at each period, a new model is fitted, which only relies on data from the current period and not over a questionably relevant history. It does however result in fewer training observations to fit the model and hence shrinkage, or other forms of model reduction, are needed to prevent over-fitting of highly parameterized models such as deep neural networks. Equivalently, Bayesian inference under a choice of prior on the parameters can reduce the number of training samples needed (see \cite{Dixon2022} for further details of a Bayesian deep fundamental factor model). 

These predictive models, however, offer no insight into the latent factor structure and therefore offer no insight into whether the characteristics/expected return  relationship is driven by compensation for exposure to latent risk factors \citep{KPS2019}. Specifically, there appears to be mechanism for identifying how the characteristics proxy for loadings on common risk factors.

\cite{GU2021429} present a conditional latent factor model for excess asset returns, in which the conditional and dynamic $K-$factor loadings $\mybeta(z_{i,t})$ are a non-linear parametric function of $p$ lagged firm characteristics, $z_{i,t}$. Additionally the excess asset returns are linear in the $K$-latent factors, $\mathbf{f}_t$, which are in turn a linear combination of all excess returns in the universe. The model is fitted period-by-period, i.e. $t=1,\dots,T$ as follows:

\begin{eqnarray}
r^e_{i,t+1}&=&\mybeta(z_{i,t})^\top\mathbf{f}_{t+1} + \epsilon_{i,t+1}, ~i=1,\dots, N_t,\\
\mybeta(z_{i,t})&:=&g_{\theta_t}(z_{i,t}), \\
f_{k,t}&:=&\mathbf{w}_{k,t}^\top\mathbf{r}^e_t, ~ k=1\dots K. 
\end{eqnarray}

The advantage of their approach is three-fold: (i) asset return predictability derives entirely from the lagged firm characteristics via the exposures to the risk factors;  (ii)  the factors aren't pre-determined heuristically by a two-step procedure, as in the Fama-French model, but rather the factor structure is determined statistically in a single-step, taking a data driven approach to factor structure derivation and not relying on heuristics for characterizing the factor-structure; and (iii) the factors can be interpreted from how the dynamical loadings map to firm characteristics. For example, in any given period, their model can characterise which firm characteristics are the most prominent across all assets in any factor loading. 


A technical development in \cite{GU2021429} is the presentation of an auto-encoder like neural network architecture which is customized to represent the structure of the model. Specifically a feed-forward architecture for the firm characteristics yields a hidden layer output representing the factor loadings and this is subsequently multiplied by latent factors which are given by a linear feedforward layer, with $N_t$ excess asset returns as inputs. While this network exhibits the signature bottleneck of an autoencoder, it is technically not an auto-encoder. This is because, like PCA, autoencoders are unsupervised learning methods which rely solely on returns to identify the latent factor structure. They do not leverage conditioning variables to identify this factor structure.

However, in seeking to modify this architecture to include firm characteristic data, they have essentially arrived at what partial least squares regression was exactly designed for, namely supervised dimension reduction by accounting for the returns and leveraging the lagged firm characteristics to identify the latent factor structure. Hence, although \cite{GU2021429} serve as the inspiration for our work,  we argue that the deep partial least squares regression architecture is the natural choice for non-linear conditional latent factor models.  

Another important distinguishing aspect of factor models is their approach to variable selection and regularization.  The literature seems to be divided in statistical variable selection and regularization approaches versus the use of financial theory or ``model based'' variable selection and regularization.  A prime example of statistical approaches would be LASSO penalized linear models which use shrinkage and variable selection to manage high dimensionality by forcing
the coefficients on most regressors near or exactly to zero. This can, however, produce sub-optimal forecasts
when predictors are highly correlated. A simple example of this problem is a case in which all of the
predictors are equal to the forecast target plus an iid noise term. In this situation, choosing a subset
of predictors via LASSO penalty is inferior to taking a simple average of the predictors and using this
as the sole predictor in a univariate regression. However, other approaches to variable selection are commonly used in factor modeling. PCA, PCR, and PLS are all examples of projection methods which rely on a shrinkage mechanism which shrinks away from the origin for certain eigen-directions. The resulting shrinkage factors along each of the eigen-directions, referred to as principle components in PCA and PCR, or scale factors in PLS, are chosen to maximize the explained covariance either in the predictors or between the predictors and the responses respectively. These static approaches appear widely in the literature, most notably the asymptotic PCA (APCA) method \citep{RePEc:eee:jfinec:v:21:y:1988:i:2:p:255-289} which is suitable when $N>T$, and the use of PCA \citep{swanson2014learning} when $N<T$. 
As a panel regression, $\mathbf{R}^e= B F + E$, \cite{16320} and \cite{BaiandNg2002} apply PCA to the covariance of asset returns:
$$\bar{\mathbf{R}}^e= \hat{P}\Lambda \hat{Q}^\top + \hat{E},$$
where $\hat{P}$ and $\hat{Q}$ are the left and right singular vectors of the covariance matrix, representing the factor loadings, and $\Lambda$ is a diagonal matrix with eigenvalues, sorted in descending order, representing the factors (i.e. principle components). The first latent factor corresponds to the direction that captures the maximum variance in the asset returns, the second factor corresponds to the direction of maximum variance under a complementary projection in the remaining orthogonal subspace, and so forth. This parsimony of factors is appealing as it can explain asset return covariance by a small number of systematic factors across the market.

These approaches suffer from many shortcomings -- aside from violating the no-arbitrage principle unless the PCA methodology is modified to use uncentered second moments of excess returns-- the main shortfalls being that the latent factors poorly explain the outliers in asset returns. As put succinctly by  \cite{KozakNagelSantosh2018JF}, ``\emph{It is perfectly possible for a
[principle component] factor to be important in explaining return variance but to play no role in pricing}''. Another point of distinction is that the factors are purely statistical and not macro-financial, hence such risks are difficult to interpret in a macro-financial model.

A key development relevant to this paper is \cite{KELLY2015294} , who apply PLS as a three-pass regression filter to successfully filter out the impact of irrelevant risk factors, a problem that the method of principal components alone cannot be guaranteed to do. However, unlike their three-step regression, our approach is a period-by-period single step model and, like \cite{GU2021429}, our betas are dynamic. So while our latent risk factors may vary period-to-period, we avoid data stationarity assumptions which are needed by \cite{KELLY2015294}. 

Another line of reasoning is pursued by \cite{10.1093/rfs/hhaa020}, who propose a new method for estimating latent asset pricing factors that fit the time-series and
cross-section of expected returns. Their estimator generalizes PCA by including a penalty on the pricing error in expected return leading to statistical risk factors which depend only on small set of firm characteristics. Their work sheds light on the ``factor zoo'' problem, namely which risk
factors are really important and which factors are subsumed by others \citep{Cochrane2001}. While our model is period-by-period cross-sectional and does not mirror the experimental design of \cite{10.1093/rfs/hhaa020}, we do seek to identify how, in any period, out-of-sample predictions of asset returns are attributed to the statistical risk factors and which firm characteristics, combined with their interaction effects, are most prominent.

Another modernized approach, which generalizes PCA to allow both dynamic loading and factor structure derived from firm characteristics, is the instrumented PCA (IPCA) factor model \cite{KPS2019}, which for ease of exposition, we show in its restricted form:
\begin{eqnarray}
r^e_{i,t+1}&=&z_{i,t}'\Gamma f_{t+1} + \epsilon_{i,t+1}, ~i=1,\dots, N_t,~ t=1,\dots,T,
\end{eqnarray}
where $\Gamma \in \mathbb{R}^{p\times K}$ projects the lagged firm characteristics on to the $K< p$ latent factors and is static. Note that the loadings are now dynamic due to the reliance on the lagged firm characteristics. This model allows for a potentially large number of characteristics to be mapped on to a smaller set of common latent risk factors that explain anomalies in asset returns. The projected characteristics, or more precisely, a linear combination of characteristics, thus proxy for loadings on each of the common risk factors and admit the intuitive interpretation that higher loadings result in higher expected returns to compensate for higher exposures, i.e. increased factor risk premia. More importantly, the most important characteristics associated with each risk factor can be estimated from observed excess returns and firm characteristics without, apriori, the use of heuristic mappings between them.




Our experimental design broadly follows the cross-sectional modeling approach used by
\citep{10.2307/1831028} and BARRA factor models (see \cite{RePEc:ucb:calbrf:44, Carvalho9}) are appealing because of their simplicity and their economic interpretability, generating tradable portfolios. The BARRA model is able to capture linear effects of macroeconomic events on individual securities. The choice of factors are microeconomic characteristics - essentially common factors, such as industry membership, financial structure, or growth orientation \citep{Nielsen2010}.

Under the assumption of homoscedasticity, factor realizations can be estimated in the BARRA model by ordinary least squares regression\footnote{The BARRA factor model is often presented in the more general form with heteroscedastic error but we do not consider the non-linear extension here.}. OLS regression exhibits poor expressability and relies on the Gaussian errors being independent of the regressors. Generalizing to non-linearities and incorporating interaction effects is a harder task.

The BARRA fundamental factor model\footnote{Specifically, the predictive form of the BARRA model is conventionally used for risk modeling, however, we shall use adopt a non-linear version of this model for a stock selection strategy based on predicted returns.} is a single-index cross-sectional model which assumes that the factor loadings are proxied by the firm characteristics. The model expresses the linear relationship between $K$ fundamental factors and the asset's returns:
\be
\mathbf{r}_{t+1} =B_{t+1}\mathbf{f}_{t+1} + \mathbf{\epsilon}_{t+1},~ t=1,\dots,T,
\ee
where $B_t= [\mathbf{1}~|~{\boldsymbol \beta}_{1}(t)~|\cdots|~{\boldsymbol \beta}_{K}(t)]= [\mathbf{1}, Z_{t-1}]$ is the $N\times K+1$ matrix of known factor loadings (betas): $\beta_{i,k}(t):=\left(\boldsymbol{\beta}_k\right)_i(t)$ is the exposure of asset $i$ to factor $k$ at time $t$ and assumed to proxy the $k^{th}$ firm characteristic. The loadings therefore correspond to asset specific attributes such as market capitalization, industry classification, style classification. $\mathbf{f}_t=[\alpha_t, f_{1,t},\dots, f_{K,t}]$  is the $K+1$ vector of unobserved factor realizations at time $t$, including $\alpha_t$. $\mathbf{r}_{t}=\frac{P_{t}}{P_{t-1}}-1$ is the $N-$vector of single-period asset returns observed at time $t$. The errors are assumed independent of the factor realizations $\rho(f_{i,t}, \epsilon_{j,t})=0, \forall i,j, t$ with heteroschedastic Gaussian error, $\mathbb{E}[\epsilon^2_{j,t}]=\sigma^2_j$. Then, in principle, factor realizations can be estimated with OLS regression. Note that there are more general estimation frameworks for estimating factors under heteroschedastic error (see e.g. \cite{Korteweg2009CorporateCS, Tsay:2010}). The BARRA model can be viewed as a special case of IPCA when the projection is the identity matrix \citep{KPS2019}.

Following \cite{Dixon2020}, we can generalize the BARRA model by considering a period-by-period non-linear cross-sectional fundamental factor model of the form
\be \label{eq:non-linear}
\mathbf{r}_{t+1}=  F_{t+1}(Z_{t}) + \mathbf{\epsilon}_{t+1}, ~ t=1,\dots,T,
\ee
where $\mathbf{R}_{t+1}$ are single-period asset returns, $F_{t+1}:\mathbf{R}^{K}\rightarrow \mathbf{R}$ is a differentiable non-linear function that maps the $i^{th}$ row of $Z_{t}$ to the $i^{th}$ asset return at time $t+1$. The map is assumed to incorporate a bias term so that $F_{t+1}(\mathbf{0})=\bf{\alpha}_{t+1}$. In the special case when $F_{t+1}(Z_{t})$ is linear, the map is $F_{t+1}(Z_{t})=Z_{t}\mathbf{f}_{t+1}$. $\mathbf{\epsilon}_{t+1}$ are zero mean i.i.d. error variates.

The forecasted returns, $\hat{r}_{t+1}$ are given by the conditional expectation, $\hat{r}_{t+1}=\mathbb{E} [r_{t+1}~ |~ Z_t ]= F_{t+1}(Z_t)$. Hence, we seek to predict returns $\mathbf{r}_{t+1}$ over the horizon $[t,t+1]$ from the lagged firm charactertistcs $Z_{t}$, which are by construction measurable at time $t$. This model can hence be mapped on to a growing body of research on returns prediction using neural networks. One appealing aspect of this approach is that stationarity of the factor realizations is not required since we only predict one period ahead and, in contrast to time series models, then retrain in each subsequent period using the latest observations only. 

In Section \ref{sect:DPLS_factors}, we show how DPLS is applied to this model to project the firm characteristics onto a smaller set of latent risk factors with resulting non-linear factor structure in the learned model and characterization of the risk factor premia.

\subsection{Stochastic discount factors}

For completeness, we comment on our modeling approach from an entirely different line of reasoning rooted in arbitrage pricing theory. \cite{Cochrane2001} shows how the one-factor CAPM model of Sharpe can be obtained as a result of consumption utility maximization with the Markowitz quadratic utility function, and under the assumption that the market is at equilibrium and all investors are identical (and the same as the representative investor) and hold the market portfolio. Furthermore, \cite{Cochrane2001} also shows how different asset pricing models can be equivalently interpreted as particular models for the stochastic discount factor (SDF) (an `index of bad times') $ m({\bf f}_t,t) $ where $ {\bf f}_t $ is a set of observed or inferred dynamic variables such as market and sector returns, macroeconomic variables etc. The SDF $ m({\bf f}_t,t) $ does \emph{not} depend on characteristics of individual assets such as stocks, because with this theory the SDF is \emph{universal} for all assets including stocks, bonds, derivatives etc. 

\begin{figure}[h!]
\centering
\begin{tabular}{cc}
\includegraphics[width=0.45\textwidth]{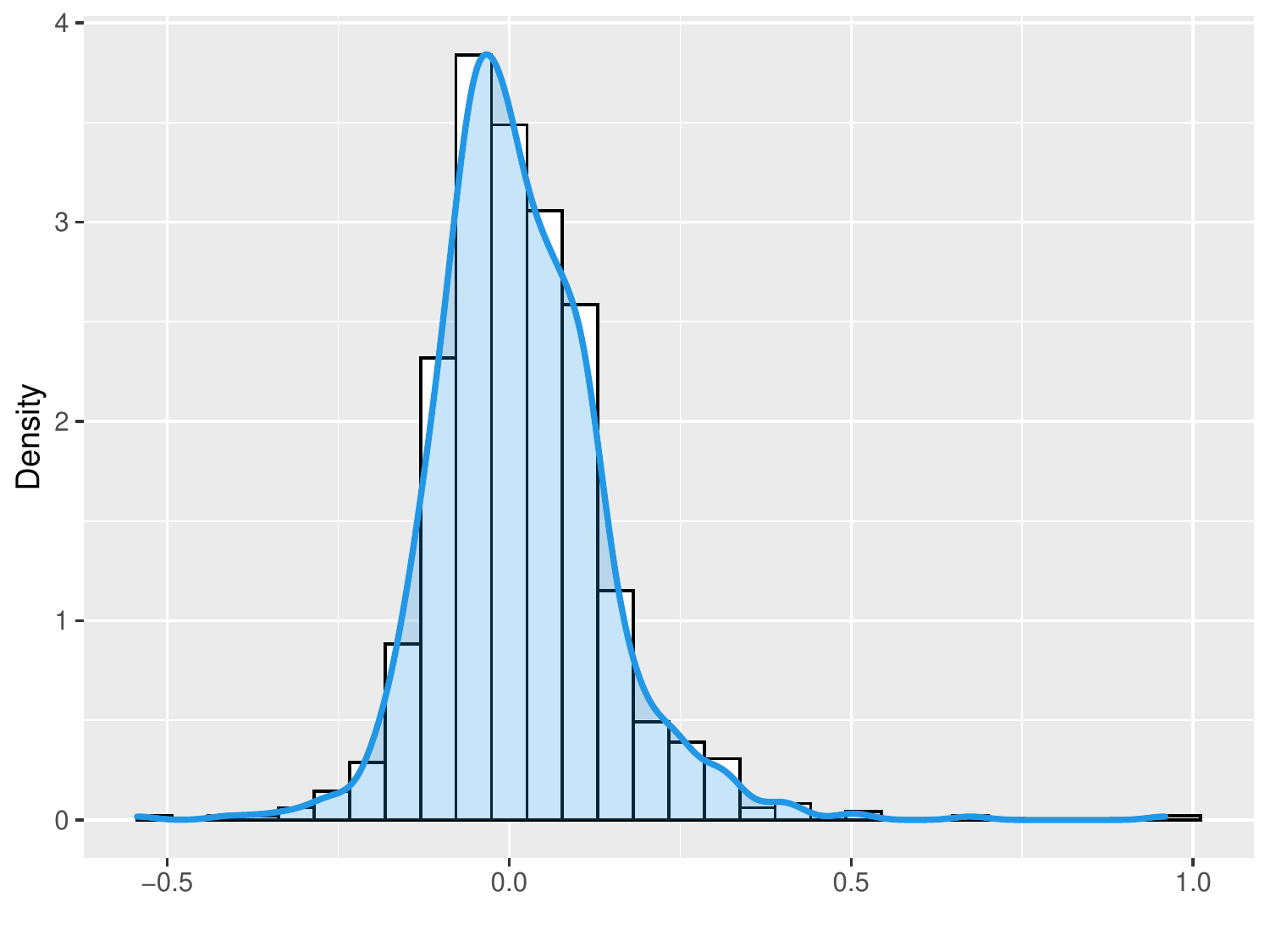}  &
\includegraphics[width=0.45\textwidth]{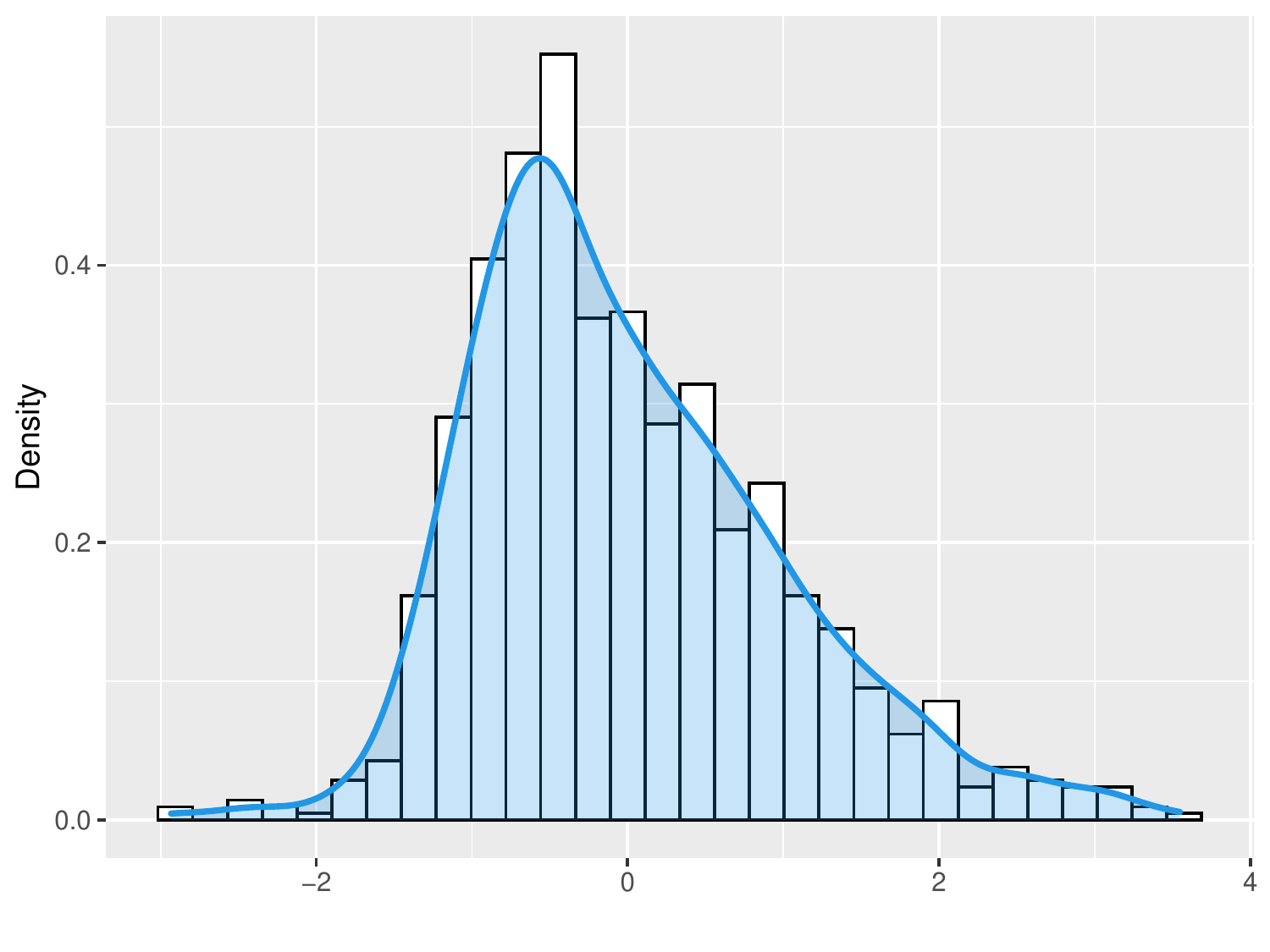} \\
(a) Excess returns & (b) Book-to-Price \\
\includegraphics[width=0.45\textwidth]{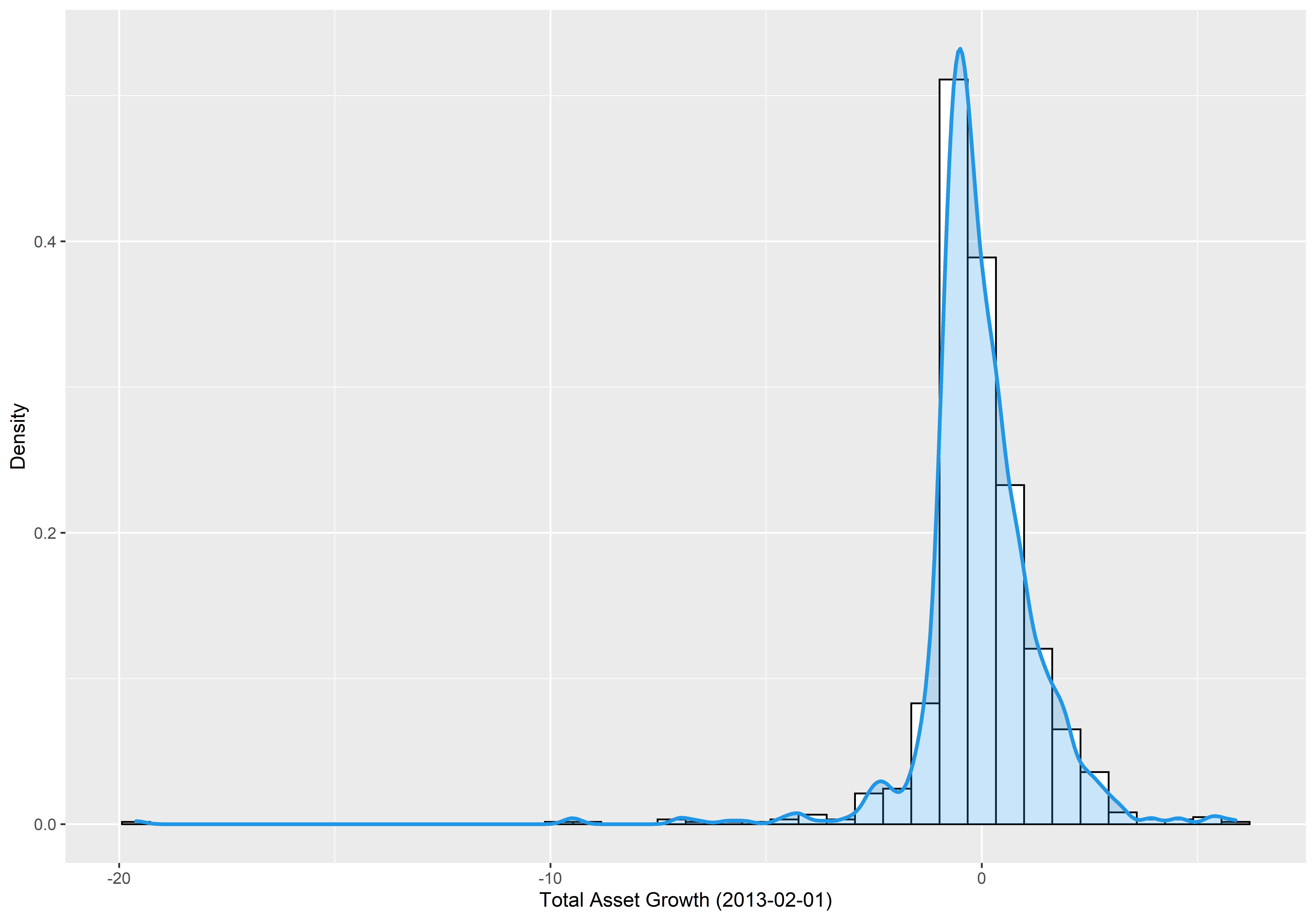}  &
\includegraphics[width=0.45\textwidth]{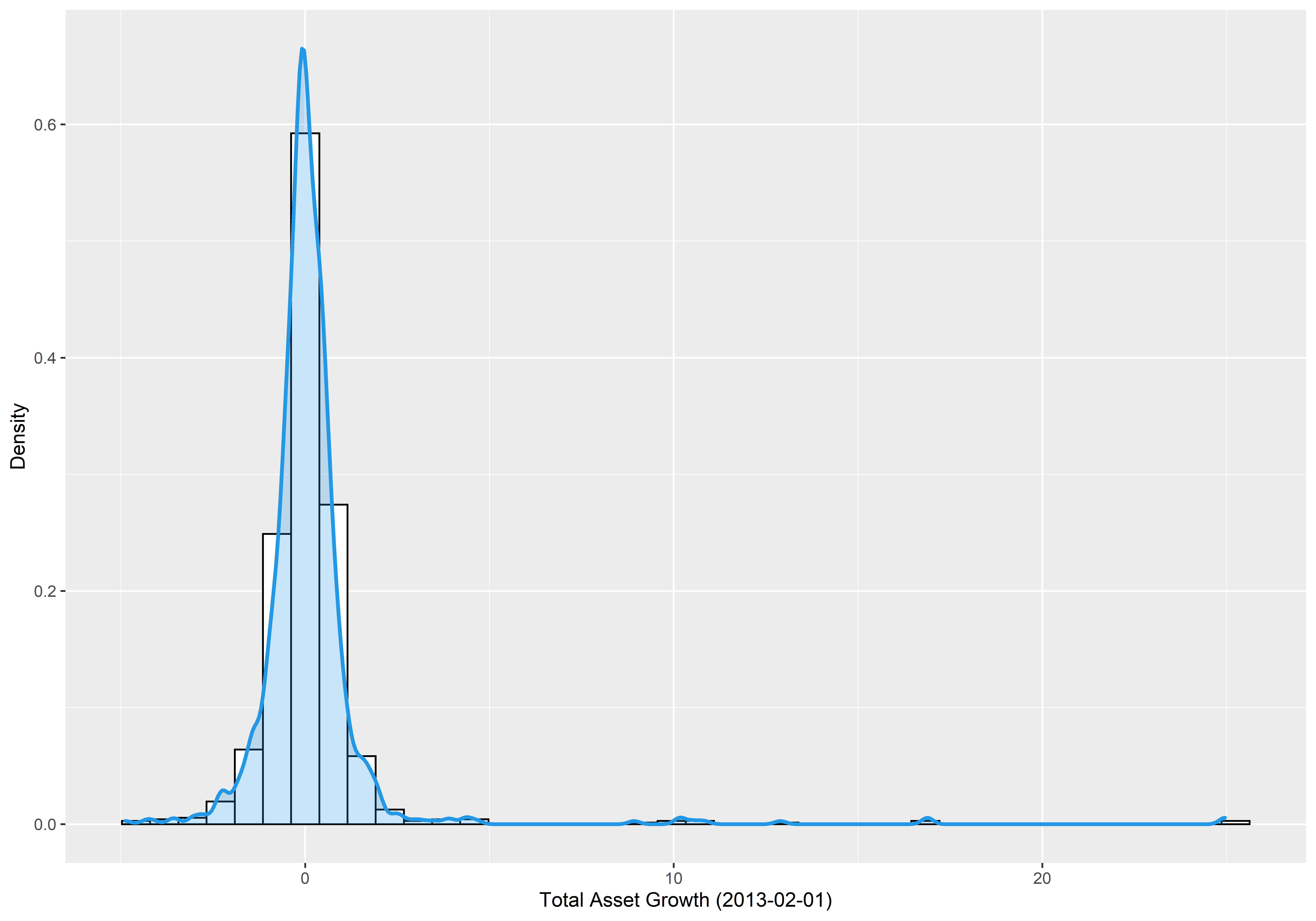} \\
(c) Total asset growth & (d) Cash flow volatility\\
\end{tabular}
\caption{\textit{Cross sectional excess returns and fundamental factors across stocks in the Russell 1000 index in February 2013. (a) The monthly excess asset returns (Skewness 0.96, Kurtosis 5.37); (b) Book-to-price ratios (Skewness 0.71, Kurtosis 0.58); (c) Total asset growths (Skewness -3.21,  Kurtosis 40.28) (d) Cash flow volatilities (skewness 7.76, Kurtosis 86.37).}}
\label{fig:histograms}
\end{figure}

While less visited than the above literature on estimating risk premia, there is yet another line of reasoning that focuses on estimation of risk prices, i.e., the extent to which the factor associated with a characteristic helps price assets by contributing to variation \citep{KozakNagelSantosh2018JF, pelgery2019}. The former authors arrive at a characteristics-sparse stochastic discount factor (SDF) representation that is linear in only a few such factors and conclude that there simply isn't enough redundancy across the factors to justify only a handful of factors, e.g. FF3, FF4, FF5 etc.  
\citep{pelgery2019}'s approach to enforcing no-arbitrage factor modeling is predicated on the notion that the mere existence of a SDF $ m({\bf f}_t,t) $ itself is critically dependent on the absence of arbitrage, once this assumption is accepted, any model building amounts to 
designing a parameterized model   $ m_{\theta}({\bf f}_t,t) $ where $ \theta $ is a set of trainable parameters, and training the model on available data. \cite{pelgery2019} introduced Generative Adversarial Networks (GANs) to 
incorporate the no-arbitrage condition as part of the neural network algorithm and estimate a linear stochastic discount factor that explains all stock returns from the conditional moment
constraints implied by no-arbitrage. For ease of exposition, here we drop the $t$ subscript on the SDF, returns, and factors.

It is well known, as a consequence of Stein's lemma, that any nonlinear model of the stochastic discount factor, $m = g(\mathbf{f})$, with a Lipschitz continuous $f$, can be turned into a linear model $m = a+\mathbf{b}^\top \mathbf{f}$ by assuming normal returns, $R$, and normal factors, $\mathbf{f}$ (see pg. 154 of \cite{Cochrane2001}). Such a linear stochastic discount factor corresponds to a linear factor beta pricing model $\mathbb{E}[R]=\alpha + \mybeta^\top\mylambda$, where $\mylambda=\mathbb{E}[\mathbf{f}]$.

However, as demonstrated by the cross-sections of all Russell index assets on February 2013 in Figure \ref{fig:histograms}, the assumption of normality on monthly excess returns and factors is too strong and we therefore seek a non-linear stochastic discount factor corresponding to a non-linear factor beta pricing model, $h(\mathbf{\mylambda})$. We show that expressing the stochastic discount factor $m$ in terms of a non-linear function, $g(\mathbf{f})$, in the factors $f$ is equivalent to the non-linear factor beta pricing model.

\begin{theorem}[Non-linear discount factors]

Given the non-linear stochastic discount factor 

$$m=g(\mathbf{f}), ~g: \mathbb{R}^K \rightarrow \mathbb{R} \in C^{\infty}(\mathbb{R}^p),~ 1=\mathbb{E}[m]$$

\noindent $\exists$ a function $h: \mathbb{R}^K \rightarrow \mathbb{R} \in C^{\infty}(\mathbb{R}^p)$ s.t. $\mathbb{E}[R]=h(\mathbf{\mylambda}).$
\end{theorem}
\noindent See Appendix \ref{sect:appendix:proof} for the proof.
\begin{remark}
Conversely, given a $h(\mathbf{\mylambda})$, $\exists$ a $g(\mathbf{f})$.
\end{remark}

\begin{remark}
Note for, avoidance of doubt, that we do not estimate the SDF in this paper as in \cite{KozakNagelSantosh2018JF, pelgery2019}. Rather we reason, on the basis of this theorem, that if we estimate risk premia via a non-linear function, $h(\mathbf{\mylambda})$, then there is a corresponding non-linear SDF. The economic implications of a corresponding utility function being non-Markowitz quadratic are beyond the scope of this paper. However, we comment in passing that many constrained utility functions used by practitioners to incorporate exposure limits, maximum drawdown limits, shorting selling limits, CVaR penalized expected returns etc,  depart from the mean-variance utility function and motivate a more flexible form of the SDF than linear in the discount factors.
\end{remark}
\section{DPLS Regression} \label{sect:regression}
Informally, our goal is to jointly project the excess asset returns and firm characteristics with projection matrices $Q_t$ and $P_t$  on to a pair of components $(U_t,V_t)$ respectively such that:

$$\mathbf{r}^e_t=U_tQ_t + \mathbf{\epsilon}_t^r, ~Z_{t-1}=V_tP_t + \mathbf{\epsilon}_t^z,$$
in a such a way that the prediction of excess asset returns  $$\hat{\mathbf{r}}^e_{t+1}=\mathbb{E}[ \mathbf{r}^e_{t+1}~|~ Z_t]=\hat{U}_tQ_t$$ is based on using the learned map $\hat{U}_t=g_{\theta}(V_t)$ in the projected space. The approach is thus using conditioning information from the lagged firm characteristics and the excess asset returns to arrive at a non-linear representation of the latent risk factors.

To formalize this method, we begin by reviewing linear PLS regression in the more general setup before introducing the DPLS method of \cite{Polson2021}. Then in Section \ref{sect:DPLS_factors} we introduce a deep PLS factor model, where the factor loadings, $V_t$, are interpreted as risk premia and the latent risk factors are given by the gradient of the fitted function, $g_{\theta}(V_t)$, in the projected space.

In general, PLS regression uses both $X\in \mathbb{R}^{N\times p}$ and $Y^{N\times q}$ to calculate the projection, further it simultaneously finds projections for both input $X$ and output $Y$, making it applicable to more general problems with multi-index output vectors as well as high dimensional input vectors. PLS searches for a set of latent vectors (i.e. components) that perform a simultaneous decomposition
of $X$ and $Y$ with the constraint that these components explain as much as
possible of the sampled covariance between $X$ and $Y$. 

PLS hence finds projection directions which maximize the sample covariance between $X$ and $Y$--- the resulting projections $U$ and $V$ for $Y$ and $X$, respectively are called  score matrices, and the projection matrices  $P$ and $Q$ are called loadings. \\

\begin{figure}
\centering
\includegraphics[width=0.6\textwidth]{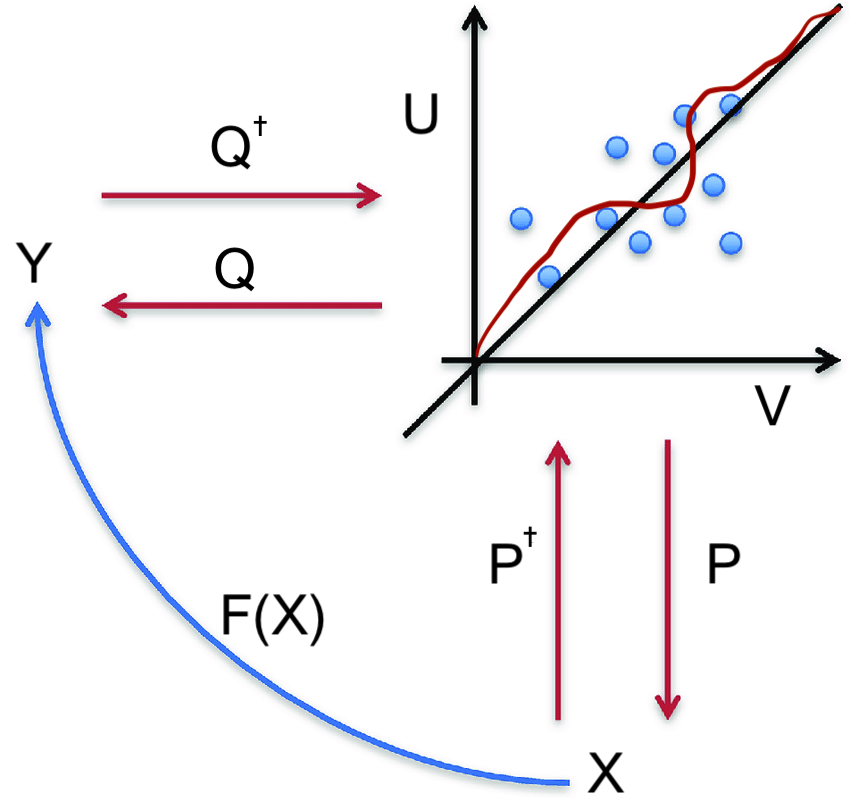} 

\caption{\textit{An overview of the PLS (DPLS) method. PLS uses a linear approximation between U and V, whereas DPLS uses a deep neural network approximation (red line in regression plot).}}
\label{fig:dpls}
\end{figure}

The input-score $V$, has $K$ columns, one for each ``feature" and $K$ is chosen via cross-validation. The key principle is that $V$ is a good predictor of $U$, the output-scores. These relations are illustrated in Figure \ref{fig:dpls} and summarized by the equations below for the general multivariate PLS model:

\begin{align*}
	Y&  =  U Q+ E^y, \\
	X& =  V P+ E^x, V^\top V=I, 
\end{align*}
Here $Q \in \mathbb{R}^{K\times q}$ and $P\in \mathbb{R}^{K\times p}$ are low-rank projection
(loading) matrices and $V, U \in \mathbb{R}^{N\times K}$ are projections of $X$ and $Y$ respectively. The columns of $V$ are orthonormal and referred to as the ``latent vectors''.

DPLS uses a feedforward neural network to regress the output-score on the input-scores:
\begin{equation*}
U =  G_{W,b}(V) + E,
\end{equation*}
where $G_{W,b}$ is a deep neural network with $L$ layers, that is, a super-position of univariate semi-affine functions, $\act^{(\ell)}_{W^{(\ell)},b^{(\ell)}}$, to give 
\begin{equation*}
U:=G_{W,b}(V)=(\act^{(L)}_{W^{(L)},b^{(L)}}\circ\dots\circ \act^{(1)}_{W^{(1)},b^{(1)}})(V),
\end{equation*}
and the unknown parameters are a set of weight matrices $W=(W^{(1)},\dots, W^{(L)})$ and a set of bias vectors $b=(b^{(1)},\dots, b^{(L)})$.  Any weight matrix $W^{(\ell)}\in \mathbb{R}^{m\times n}$, can be expressed as $n$ column m-vectors $W^{(\ell)}=[\mathbf{w}^{(\ell)}_{,1},\dots, \mathbf{w}^{(\ell)}_{,n}]$. We denote each weight as $w^{(\ell)}_{ij}:=\left[W^{(\ell)}\right]_{ij}$.

The $\ell^{th}$ semi-affine function is itself defined as the composition of the activation function, $\act^{(\ell)}(\cdot)$, and an affine map:
\be
\act^{(\ell)}_{W^{(\ell)},b^{(\ell)}}(Z^{(\ell-1)}):=\act^{(\ell)}\left(W^{(\ell)}Z^{(\ell-1)} + b^{(\ell)}\right),
\ee
where $Z^{(\ell-1)}$ is the output from the previous layer, $\ell-1$. The activation functions, $\act^{(\ell)}(\cdot)$, e.g. $\act^{(\ell)}(\cdot)=\max(\cdot,0)$, are critical to non-linear behavior of the model. Without them, $G_{W,b}$ would be a linear map and, as such, would be incapable of capturing interaction effects between the inputs. This is true even if the network has many layers.

\paragraph{Deep learning theory}
 The choice of feed-forward architecture is motivated on theoretical and empirical grounds. First, feed-forward networks are furnished with a universal approximation theorem \citep{Hornik1989}. It has further been shown that deep networks can achieve superior performance versus linear additive models, such as linear regression, while avoiding the curse of dimensionality \citep{poggio_deep_2016}.

\cite{Martin2018} show that deep networks are implicitly self-regularizing and \cite{TishbyZ15} characterizes the layers as 'statistically decoupling' the input variables.

There are additionally many recent theoretical developments which characterize the approximation behavior as a function of network depth, width and sparsity level \citep{2018arXiv180309138P}.  Recently \cite{DBLP:journals/corr/HarveyLM17} prove upper and lower bounds on the expressability of deep feedforward neural network classifiers with the piecewise linear activation function, such as ReLU activation functions. They show that the relationship between expressability and depth is determined by the degree of the activation function. 
 
 These bounds are tight for almost the entire range of parameters. Letting $n$ denote the total number of weights, they prove that the VC-dimension is $\mathcal{O}(nLlog(n))$. 
 
 They further showed the effect of network depth on VC-dimension with different non-linearities: there is no dependence for piecewise-constant, linear dependence for piecewise-linear, and no more than quadratic dependence for general piecewise-polynomial. Thus the relationship between expressability and depth is determined by the degree of the activation function. 
 
 There is further ample theoretical evidence to suggest that shallow networks can't approximate the class of non-linear functions represented by deep ReLU networks without blow-up.
 \cite{DBLP:journals/corr/Telgarsky16} shows that there is a ReLU network with $L$ layers and $U$ units such that any network approximating it with only $\mathcal{O}(L^{1/3})$ layers must have $\Omega(2^{L^{1/3}})$ units. 
 \cite{DBLP:journals/corr/MhaskarLP16} discuss the differences between composition versus additive models and show that it is possible to approximate higher polynomials much more efficiently with several hidden layers than a single hidden layer.

\subsection{PLS Estimation}

In the literature, there are multiple types of algorithms for finding the projections \citep{manne_analysis_1987}. Algorithms differ on whether they estimate the score matrix $V$ as orthonormal columns or not. The original one proposed by \cite{wold_collinearity_1984} uses the conjugate-gradient method \citep{golub2013matrix} to invert matrices. The first PLS projection $p$ and $q$  is found by maximizing the sample covariance between the $X$ and $Y$ scores \maxdisp{p,q}{\left(Xp\right)^\top\left(Yq\right)}{||p||=||q||=1.} Then the corresponding scores are
\[
v = Xp,\text{   and   } u = Yq.
\]
We can see from the definition that the directions (loadings) for $Y$ are the right singular vectors of $X^\top Y$ and loadings for $X$ are the left singular vectors. The next step is to perform regression of $V$ on $U$, namely $U = V\beta$.  The next column of the projection matrix $P$ is found by calculating the singular vectors of the residual matrices $(X - vp^\top)^\top(Y - u\beta q^\top)$.

The final regression problem is solved $Y = UQ = V\beta Q = XP^\top\beta Q $. Thus the PLS estimate is
\[
\beta_{\mathrm{PLS}} = P^\top\beta Q.
\]

\cite{helland_partial_1990} showed that PLS estimator can be calculated as 
\[
	\beta_{\mathrm{PLS}} = R(R^\top S_{xx}R)^{-1}R^\top S_{xy}
\]
where $R = (S_{xy},S_{xx}S_{xy},\ldots,S_{xx}^{q-1}S_{xy})$,
\[
	S_{xx} = \dfrac{X^\top(I-{\bf 11^\top}/n)X}{n-1},
\]
\[
	S_{xy}  = ave ( YX ),
\]
where the average is taken over training data, that is $ave (X)= \frac{1}{n} \sum_{i=1}^n \x_i$. Once the loadings $(Q,P)$ and scores $(U,V)$ have been estimated, we can then train a deep learner on $(U,V)$ and then use $U$ to predict $V$ and hence obtain an estimate of $Y$.

To understand why the deep-learner can be decoupled from the estimation of the projection matrices, we 
need to establish that the projection directions in PLS are invariant to the functional form of the regression between score matrices, up to a constant of proportionality. 

We first recall a key asymptotic result from \cite{Naik2000} which builds on an earlier result on OLS regression by \cite{brillinger_generalized_2012}, namely that PLS regression is a consistent estimator of the regression coefficient even when there exists an unknown non-linear relationship between the $\y$ and $\x$.


\begin{Assumption}
Assume $\x_1,\dots, \x_N$, where  $\x^\top = (x_1, \dots, x_p)$, are independent normals with mean $\mu_x$ and
non-singular covariance matrix $\Sigma_{p\times p}$, and that the xs are independent of $\myepsilon_1,\dots, \myepsilon_N$, where each $\myepsilon^\top =(\epsilon_1, \dots, \epsilon_q)$ has finite variance $\sigma^2$. 
\end{Assumption}

\begin{Assumption}
Let $\y = \myg + \myepsilon$, where $\y^\top= (y_1, \dots, y_q)$, $\myg^\top= (g_1, \dots, g_q),~ g_j := g_j(z_j), z_j = \beta_{0j} + \x^\top\mybeta_j, ~\beta_{0j}$ is a scalar and $\mybeta_j\in \mathbb{R}^p$. Assume that  $g(\cdot)$ is a measurable function of $z$ with $\mathbb{E}[ |g_j(z_j)| ] < \infty$ and
$\mathbb{E}[ |z_jg_j(z_j)|] < \infty$ for $j= 1, \dots, q$.
\end{Assumption}

\begin{theorem}[\cite{Naik2000}]
Assume the above two assumptions hold. Then the PLS estimator $\hat{B}_{PLS}$ is a strongly consistent estimator of $B = (\mybeta_1, \dots, \mybeta_q)$ up to constants of proportionality.
\end{theorem}
The theorem is closely related to Stein's Lemma \cite{stein2012}, which states that the covariance between the Gaussian regressand and a non-linear Lipschitz function of the regressors can be written as the covariance between the regressand and the regressors themselves up to a multiplicative term in the expectation of the derivative of the non-linear function.

The practical implication of Brillinger's theorem \cite{brillinger_generalized_2012} or Naik \& Tsai theorem is that OLS or PLS respectively generalize to the case $\mathbb{E}[\y~ |~ \x]$  is not linear when $\x$ is jointly Gaussian. 

For the purpose of modifying the PLS algorithm to use a deep-learner rather than the linear PLS estimator, we would like to know that the strong consistency result still holds if we introduce a non-linear regression  to estimate the relation between the scores. Of course, the mere exercise of reintroducing non-linear regression when Naik \& Tsai's  theorem shows that a linear PLS algorithm will generalize to a non-linear function seems like a tautology.

However, our immediate motivation, before building a regression model for non-Gaussian data, is to first show that we needn't restrict the PLS algorithm to a linear estimator for the result of Naik \& Tsai to hold. Rather, the power of PLS is in its composability --- we can use the x-scores and y-scores estimated by the PLS algorithm composed with a deep learning estimate of the map between the scores, rather than needing to substantially overhaul the PLS method to embed machine learning methods. More precisely, we first seek to generalize Naik \& Tsai theorem by showing that PLS can be composable with another machine learning method while preserving its consistency property.



In more detail, we show that the strongly consistent estimator, $\hat{B}_{PLS}$, is invariant under a parameterized non-linear function in the projected space, $G(\cdot)\equiv G_{\theta}(\cdot)$,~ s.t.  $u_j = G(\beta_{0j} + \myt^\top\mybeta_j)$, up to a constant of proportionality. The x-scores, $\myt$, are jointly Gaussian if the predictors, $\x$, are  jointly Gaussian, since the projection, $P\myt$, is a linear transformation.

\begin{definition}
Given the system of equations:
\begin{eqnarray*}
\y&=& Q \myu+\myepsilon_y,\\
\x&=&P\myt + \myepsilon_x,\\
\myu&=& G(\myt^\top B),
\end{eqnarray*}
where $P\in \mathbb{R}^{p\times q}$ and $Q \in\mathbb{R}^{q\times K}$ are projection matrices referred to as loadings. $\myu,\myt\in\mathbb{R}^K$ are the y-scores and orthonormal x-scores respectively. $\myepsilon_y\in\mathbb{R}^q$ and $\myepsilon_x\in\mathbb{R}^p$ are independent and have finite variance $\Sigma_{\epsilon_y}$ and $\Sigma_{\epsilon_x}$ respectively. Let $w_j=\myt^\top \mybeta_j$.
\end{definition}

\begin{Assumption}
Assume that $\myt$ is Gaussian and is a ``good'' predictor of $\myu$ through some non-linear transformation $\mathbb{E}[\myu ~|~ \myt] =G(\myt^\top B)$, where  $G(\cdot)$ is a measurable function with $\mathbb{E}[ |G_j(w_j)| ] < \infty$ and
$\mathbb{E}[ |w_jG_j(w_j)|] < \infty$ for $j= 1, \dots, q$. 
\end{Assumption}

\begin{theorem}[Composability of PLS]
Under this assumption on $G(\cdot)$, embedded in the PLS method, the PLS estimator, $\hat{B}_{PLS}$, is a consistent estimator of $B = (\mybeta_1, \dots, \mybeta_q)$.
\end{theorem}

\begin{proof}
\begin{eqnarray*}
\y&=&G(\myt^\top B)Q + \myepsilon,\\
\y&=&G((\x-\myepsilon_x)P^\top\boldsymbol{\beta})Q + \myepsilon,\\
\end{eqnarray*}
where $\myepsilon:=\myepsilon_t^\top Q + \myepsilon_y$. 

\begin{eqnarray*}
Cov(y_j, \x)&=&Cov(G_j(w_j)Q, \x)\\
&=& Cov(G_j(w_j), Q\x)\\
&=& Cov(w_j, Q\x)\frac{Cov(G_j(w_j),w_j)}{Var(w_j)}\\
&=&Cov((\x P^\top  B)_j,Q\x)\frac{Cov(G_j(w_j),w_j)}{Var(w_j)}\\
&=&Cov((\x P^\top B)_jQ,\x)\kappa_j\\
&=&Cov(\x,\x)(P^\top B)_jQ\kappa_j\\
\end{eqnarray*}
and the PLS estimator $(\hat{\mybeta}_{PLS})_j=\Sigma_x^{-1}Cov(y_j, \x)\kappa_j$ is a consistent estimator of $P^\top\mybeta_jQ, ~j=1,\dots,q$. 
\end{proof}
We discuss and speculate on this result. 

\begin{remark}
Since the PLS estimator is invariant under the choice of non-linearity in functional approximation between the x-scores and y-scores, we can use the x-scores and projection matrices estimated by the PLS method without needing to modify them. This provides the justification for composability of PLS with the deep leaner.
\end{remark}

\begin{remark}
While DPLS is intended to be a general methodology and not limited to Gaussian predictors, we make no such estimator consistency claim when $\x$ is not Gaussian. Rather, the stated goal of composability is simply to provide a better predictor of the Y-scores under the assumption that $\mathbb{E}[\y~ |~ \x]$ is non-linear as the linear estimator is no longer ``valid'' due to the non-Gaussian predictors.
\end{remark}

\begin{remark}
Under an idealized data generation process, we may use PLS as a baseline for measuring the performance of DPLS and thus if $\x$ is Gaussian, then we might expect to characterize any relative gain in performance when using DPLS in place of PLS under infinite sample sizes. After all, in such a case, a linear approximation should suffice up to a constant of proportionality.  However, we exercise some caution in over-interpreting the theorem  -- the Brillinger and Naik \& Tsai result is simply an asymptotic consistency result and says nothing about whether linear regression is ``optimal'' when there is a non-linear relationship between the response and the Gaussian predictors. In fact the asymptotic covariance matrix used for the least squares estimate of the regression coefficients depends on the unknown form of the non-linearity between $\y$ and $\x$. Needless to say, the limitations of an asymptotic result to sampling error from use of a finite sample size are part and parcel of any frequentist approach.
\end{remark}

\begin{remark}
One can estimate the weights in the first layer of the network using the PLS estimator. Put differently, we would expect the regression coefficients in the PLS method to match the weights in the first layer of the network when $\x$ is Gaussian, irrespective of the value of the weights and choice of architecture of the remaining network layers. In practice, the comparison will not be exact due to sampling noise, but also because of the constant of proportionality. However, the functional form of $G=G_{W,b}(\cdot)$ is known and the constant of proportionality can be estimated. When the data is non-Gaussian, we speculate that the weights in the first can be initialized more effectively than a random initialization under the assumption that they would be similar under mild departure from Gaussian assumptions.
\end{remark}



Our consistency result permits us to compose the loadings estimated by PLS, under the original assumption of a linear relationship between the input and output scores, with a non-linear function $G_{W,b}(\cdot)$ to improve the predictive performance of the output scores (See Figure \ref{fig:pls_reg} for an illustration).  The theorem shows that the composition of a deep-learner with the projection matrices estimated separately by the PLS algorithm does not ``break'' PLS -- the asymptotic consistent estimate result of Naik \& Tsai still holds. With this in mind, we propose the following two-step DPLS algorithm here:

\begin{enumerate}
    \item Assume that the predictors, $X$, are Gaussian and hence by the Naik \& Tsai Theorum, the PLS estimator is a consistent estimator of the coefficients in $g(\beta_0 + \beta x_j)$.
    \item Relax the assumption of Gaussian predictors: apply a deep-learner to the non-Gaussian x-scores and y-scores estimated by PLS, estimate a non-linear map to improve the prediction of the y-scores. Optionally use the PLS coefficients as initial guesses of the weights in the first layer of $G_{W,b}(\myt)$.
\end{enumerate}

In this respect, DPLS could be viewed, as an end-to-end deep learning architecture between $\x$ and $\y$: 

$$\hat{\y}=Q\circ(\act^{(L-1)}_{W^{(L-1)},b^{(L-1)}}\circ\dots\circ \act^{(2)}_{W^{(2)},b^{(2)}})(P^+\x),$$

where $P^+=$ in given by the PLS algorithm and is projection of $\x$ on to the first layer, the next layer is non-linear in the scores (with the weights $(W^{(2)},b^{(2)})$ possibly even estimated by PLS), the penultimate layer is a linear projection on to space of y-scores (i.e. the activation function $\act^{(L-1)}$ is linear and the final layer is the linear projection $Q$, estimated by PLS, from the $K$-vector space to the $q-$ vector space of $\y$.

\begin{figure}
\centering
\includegraphics[width=0.5\textwidth]{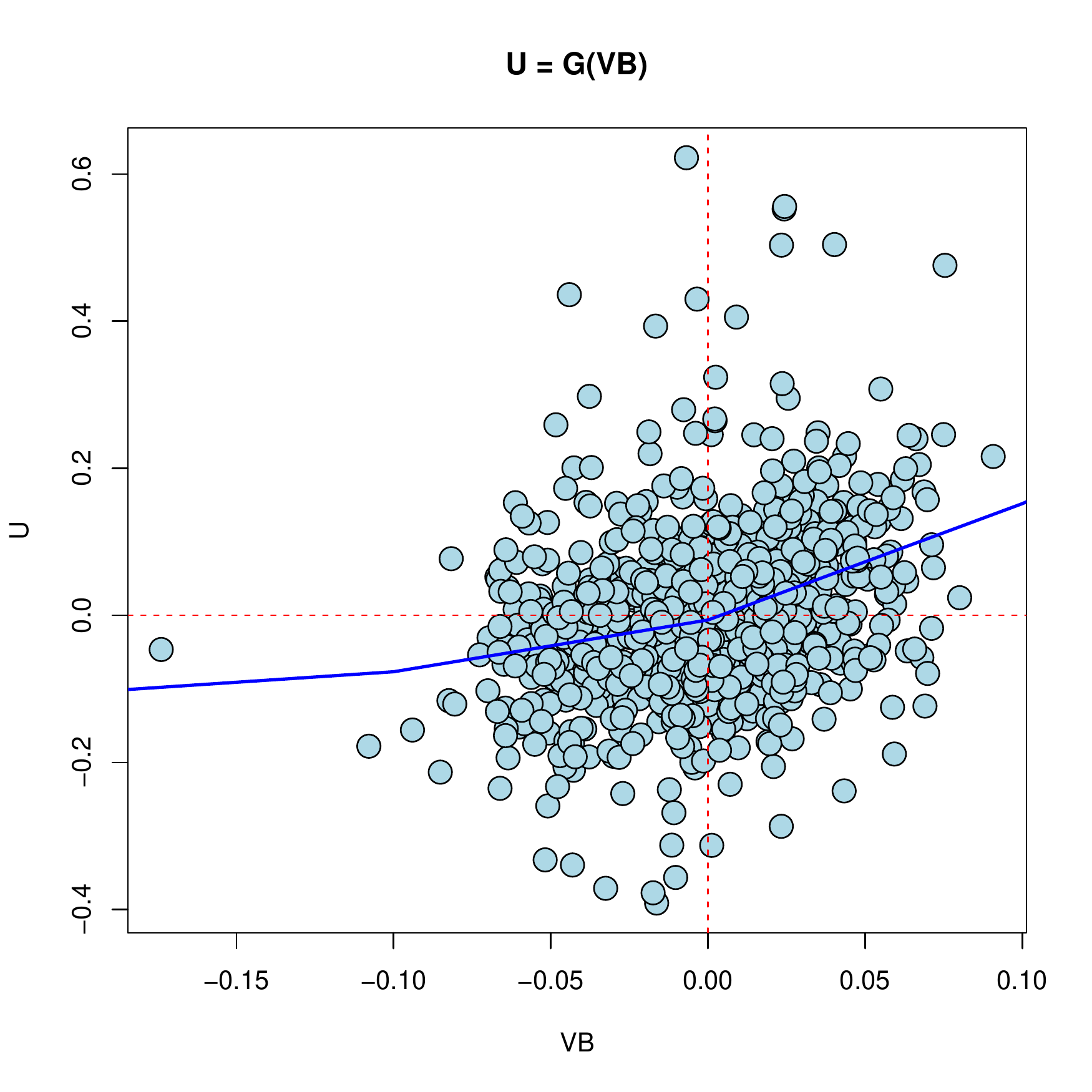} 
\caption{\textit{The y-scores ($K=1$) are shown against the PLS coefficient multiplied by the orthonormal x-scores to illustrate the non-linearity which must be captured by the deep-learner when the scores and returns are non-Gaussian. }}
\label{fig:pls_reg}
\end{figure}

More generally, PLS can be used as a layer at any stage of deep learning, not just the first. For example, in a convolutional neural network, 
$$
Z_1 = g\left(\sum_{i\in M} \x_i\star w_i + b_i\right),
$$
where $X\star w + b$ denotes the convolution over the region $M$, with input $X$, weights $w$ and bias $b$. Then we can add a PLS layer by regressing the output $Y$ on $Z_1$ and perform feature reduction, which results in a CNN-PLS model.

Our theorem and algorithm is not limited to use with deep learning either. PLS can be combined with many well-known machine learning models, such as kernel regression (\cite{rosipal2001kernel}, \cite{rosipal_kernel_2001}), Gaussian processes (\cite{gramacy2012gaussian}), tree models.  For example, \cite{higdon2008computer} uses a PCA decomposition of the $Y$-variable before applying a nonlinear regression method. Of course deep learning is attractive because of its universal representation theorem and composability with other algebraic transformations.

\subsection{PLS Algorithm}
One common problem with estimation of linear prediction rules, $ \hat{Y} = X\beta  $, on high dimensional datasets is the high variance due to the typically ill-conditioned nature of the design matrix, $ ( X^T X ) $.   Shrinkage methods are used to address this issues and work by biasing coefficient vector away from directions in which the predictors 
have low sampling variability--- or equivalently, sway from the "least important" principal components of $X$. 
 
The Bayesian shrinkage interpretation of \cite{frank_statistical_1993} provides a useful perspective for explaining the PLS algorithm. PLS, in fact, relies on a non-standard shrinkage mechanism which shrinks away from the origin for certain eigen-directions. The shrinkage factors (a.k.a. scale factors) along each of the eigen-directions are nonlinear in the response values. They also depend on the OLS solution. 

Polson and Scott \cite{polson2010,polson2012} provide a general theory of global-local shrinkage and, in particular, analyze g-prior and horseshoe shrinkage. \cite{RePEc:bla:istatr:v:88:y:2020:i:2:p:302-320} build on this work by demonstrating that horseshoe regularisation is useful for a wide class of machine learning methods.  
 
We first  center and standardise $ (Y, X )$. Then, we provide an SVD  decomposition of the sample covariance matrix:
$$ V = ave ( X X^\top ). $$  
This finds the eigenvalues $ \{e_j^2\}_{j=1}^p $ and corresponding eigenvectors $ \{\bm v_j\}_{j=1}^p $  arranged in descending order, 
so we can write the eigenvector decomposition
$$ V = \sum_{j=1}^p  e_j^2 {\bm v}_j {\bm v}_j^\top.
$$ 
This leads to a sequence of regression models $(\hat Y_0,\dots, \hat Y_R)$ with $ \hat Y_0 $ being the response mean and 
$$
\hat{Y}_K = \sum_{k=0}^K (ave ( w_k^\top X ) / e_k^2 ) \bm v_k^\top X,~~d\in \{0,\dots, R\}. 
$$
where $R$ is the rank of $V$ (number of nonzero $e_k^2$). Therefore, PLS finds  "features" $\{\bm f_k \}_{k=0}^K$= $\{\bm v_k^\top X\}_{k=0}^K $.  

\noindent The PLS estimator is of the form 
$$ 
 \hat Y^M = \sum_{j=1}^R  f_j^{PLS}  \hat{\alpha}_j \bm v_j^\top \bm x,
$$
where $f_j^{PLS}$ are PLS scale factors for the $K$ top eigenvectors  \citep{frank_statistical_1993} of the form
\begin{align*}
f_j^{PLS} &= \sum_{k=1}^{K} \theta_k e_j^{2k} \; ,  \; {\rm where} \;  \theta = w^{-1}\eta ,  \;  \eta_k = \sum_{j=1}^p \hat{\alpha}_j^2 e_j^{2(k+1)} . 
\end{align*}

\begin{figure}
\centering
\includegraphics[width=0.5\textwidth]{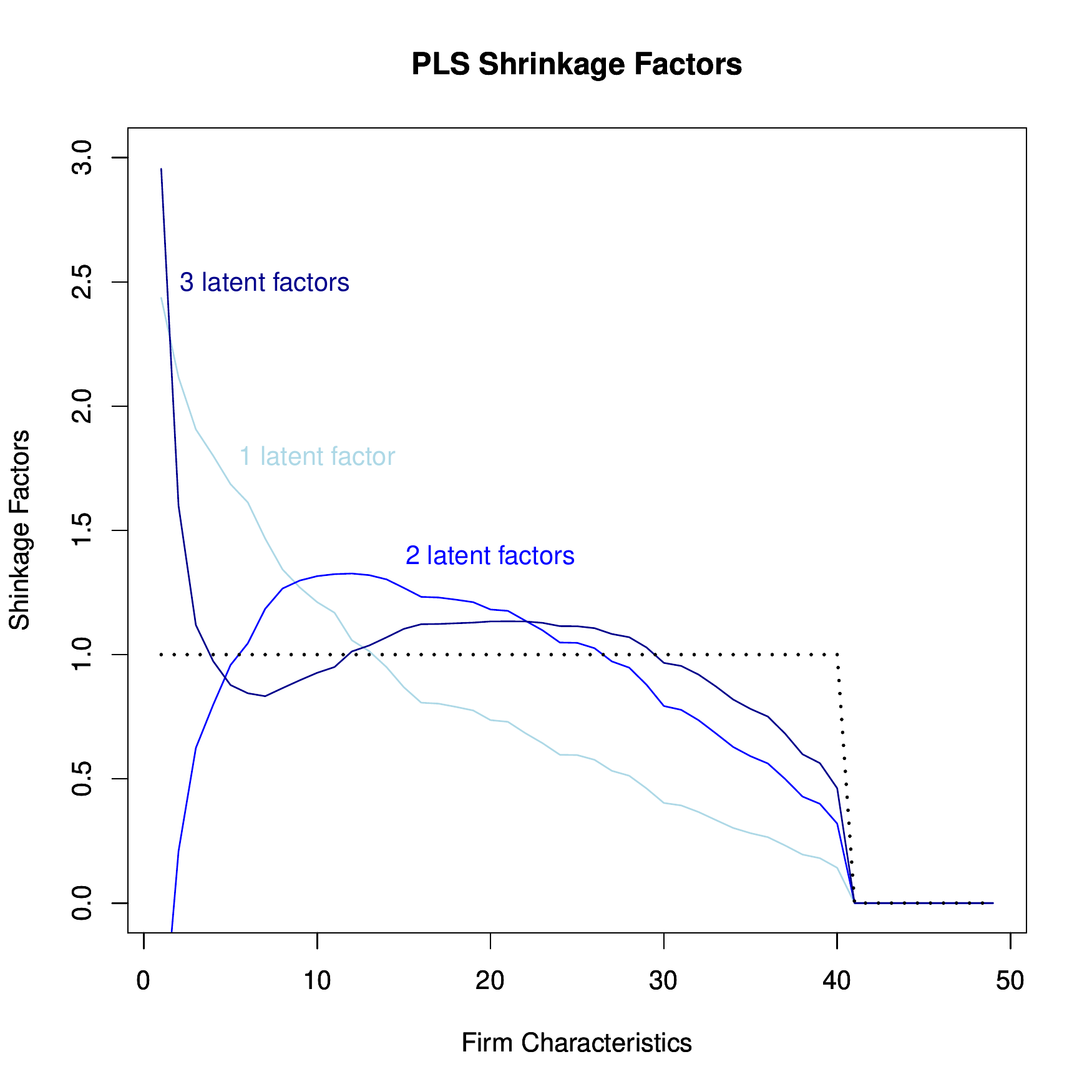} 
\caption{\textit{The scale factors are plotted for $j=1,2,3$. If any $ f_j > 1 $ then one can expect 
supervised learning  will lead to different predictions than   unsupervised learning.}}
\label{fig:shrinkage}
\end{figure}

As shown in Figure \ref{fig:shrinkage}, the scale factors provide a diagnostic plot: $f_j $.  If any $ f_j > 1 $ then one can expect 
supervised learning (a.k.a. PLS with $Y$'s influencing the scaling factors) will lead to different  predictions than   unsupervised learning.  In this  sense, PLS is an optimistic procedure in that the goal is  
to maximise the explained variability of the output in sample with the  hope of generalising well out-of-sample. For linear estimators, $f_j>1$ means that both the bias and the variance are increased. 
Once the projection matrices are estimated, a standard stochastic gradient descent algorithm, such as Adam \citep{https://doi.org/10.48550/arxiv.1412.6980}, can be used to fit the deep learning model to predict the y-scores.

\subsection{Covariate sensitivities}
We use a gradient-based technique for ranking the importance of the covariates in the fitted DPLS model. Note, for avoidance of doubt, that in contrast to PCA and autoencoders, the DPLS model sensitivities are given in terms of the observed covariates and not just the latent statistical factors. Our proposed method is consistent with how coefficients are interpreted in linear regression, i.e. as \emph{model} sensitivities and is thus appealing to practitioners accustomed to use OLS regression. In our approach, the model sensitivities are the partial derivatives of the fitted model output w.r.t. each input. This method is consistent with \cite{enguerr2019significance} who develop statistical tests for the significance of the input variables in a neural network, using their partial derivatives. The test statistic is based on a weighted distribution of the square of the neural network partial derivative, w.r.t. each input. However, the scope of their study is limited to a single-hidden layer network and treats the asymptotic distribution of the network as a Gaussian process.

Here we shall consider deep networks with an arbitrary number of layers and do not rely on asymptotic approximations, which may be limited when the network has a small number of neurons. Indeed we derive closed form expressions for the Jacobian and Hessian - the off-diagonals provide sensitivities to interaction terms. In contrast to \cite{enguerr2019significance}, our goal is to characterize the dispersion of the empirical sensitivity distribution - confidence intervals of the distributions can be found by Bootstrap sampling (see, for example, \cite{doi:10.1177/1094428105280059}). Such an approach is only useful in practice if the variance of these distributions are bounded. Our primary theoretical result then is to characterize the upper bound on the variance of the partial derivatives and show that it is bounded for finite weights.




To evaluate fitted model sensitivities analytically in the DPLS method, we require that the function $\hat{V}=G_{\hat{\theta}}(U)$ is continuous and differentiable everywhere. Furthermore, for stability of the interpretation, we shall require that $G_{\hat{\theta}}(\x)$ is a Lipschitz continuous\footnote{If Lipschitz continuity is not imposed, then a small change in one of the input values could result in an undesirable large variation in the derivative}. That is, there is a positive real constant $\kappa$ s.t. $\forall \x_1, \x_2 \in \mathbb{R}^p$, $|G_{\hat{\theta}}(\x_1) - G_{\hat{\theta}}(\x_2)| \leq \kappa|\x_1-\x_2|$. Such a constraint is necessary for the first derivative to be bounded and hence amenable to the derivatives, w.r.t. to the inputs, providing interpretability. 

Fortunately, provided that the weights and biases are finite, each semi-affine function is Lipschitz continuous everywhere. For example, the function $\tanh(x)$ is continuously differentiable with derivative, $1-\tanh^2(x)$, is globally bounded.  With finite weights, the composition of $\tanh(x)$ with an affine function is also Lipschitz. Clearly ReLU$(x):=\max(\cdot,0)$ is not continuously differentiable and one must instead use a softplus function, a smooth approximation to ReLU.  

From the definition of the DPLS model, we can write the fitted model in the compact form
$$\hat{Y}=G_{\hat{\theta}}(P^\top X)Q.$$

Taking first and second derivatives of the response w.r.t. the input variables gives
\begin{eqnarray}
\frac{\partial \hat{Y}}{\partial X}&=&\frac{\partial G_{\hat{\theta}}}{\partial V}P^\top Q\\
\frac{\partial^2 \hat{Y}}{\partial X^2}&=&P\frac{\partial^2 G_{\hat{\theta}}}{\partial V^2}P^\top Q
\end{eqnarray}
where $\frac{\partial G_{\hat{\theta}}}{\partial V}$ and $\frac{\partial^2 G_{\hat{\theta}}}{\partial V^2}$ are the Jacobian and Hessian matrices of the fitted deep learner w.r.t. to the X-scores, $V$. Thus, the covariate sensitivities in the DPLS model are linear transformations of the sensitivities of the fitted deep learner w.r.t. to the latent X-scores. See Section \ref{sect:sensitivities}, for further details of the functional form of the Jacobian and Hessian matrix.




\subsection{DPLS Factor Models} \label{sect:DPLS_factors}

To fix ideas, we begin by presenting a linear PLS factor model and then show this generalizes to our deep PLS factor model. A PLS factor model is of the form

\begin{equation}
\mathbf{r}^e_t=\hat{U}Q + \epsilon_t=\mathbf{z}_{t-1}P^+_tB_tQ_t + \epsilon_t=V_t(B_tQ_t) + \epsilon_t=V_t\mathbf{f}_t + \epsilon_t,
\end{equation}
$P^+$ is the Moore-Penrose Inverse of $P$, $B\in \mathbb{R}^{K \times K}$ is a diagonal matrix with the regression coefficients on the diagonal estimated by the PLS algorithm.

The latent risk factors, $\mathbf{f}_t=B_tQ_t$, are given by the projection of the regression coefficients and the factor loadings (the X-scores), $V_t$, are projections of the firm characteristics, $\mathbf{z}_{t-1}$ onto the K-factors.

In any time period, $t$, we can therefore attribute the expected excess returns of stock $i$ to the latent factors: 

\begin{equation}
\bar{r}^e_{i,t}:=\mathbb{E}_t[r^e_{i,t}]:=\mathbb{E}[r^e_{i,t}~|~ z_{i,t-1}] = \mathbf{v}_{i,t}^\top\mathbb{E}_t[\mathbf{f}_t],~i=1,\dots N_t. 
\end{equation}

Similarly, in any period, $t$,  we can attribute the conditional variance of the excess asset returns 
to the latent factors
\begin{equation}
\sigma^2_{i,t}:=\mathbb{V}_t[r^e_{i,t}]:=\mathbb{V}[r^e_{i,t}~|~ z_{i,t-1}] = \mathbf{v}_{i,t}^\top\Sigma_t^f\mathbf{v}_{i,t} + \tilde{\sigma}^2_{i,t},~i=1,\dots N_t,
\end{equation}
where $\Sigma^f_t$ is the covariance matrix of the latent factors at time $t$ and $\tilde{\sigma}^2_{i,t}$ is the variance of the error.

Note that PLS does not explicitly find the explained variance of the returns. Rather for a target explained covariance between the excess returns and the characteristics, PLS determines the number of factors $K$.

We now turn to a specific subset of deep fundamental factor models --- DPLS models. 
DPLS projects the response and predictors below as
\begin{align*}
	\mathbf{r}^e_{t}&  =  U_{t} Q_{t} + \boldsymbol{\epsilon}^y_t, \\
	Z_{t-1}& =  V_{t} P_{t}+ \boldsymbol{\epsilon}^x_t 
\end{align*}
and a feedforward neural network is used to predict the output scores from the input-scores:
\be
U_{t} =  G_{W_t,b_t}(V_{t}) + \mathbf{\epsilon}_t.
\ee
The primary benefit of using deep learning, is that it is able to automatically capture non-linear effects and, in particular, is robust to outliers. 

The DPLS model gives the risk premia $V_t$ on the latent risk factors by a Taylor expansion of the deep learner, $g:=G_{W,b}(V_t)$, about the origin $\mathbf{0}$:

\begin{eqnarray}
\mathbf{r}^e_t&=&\hat{U}_tQ_t + \epsilon_t\\
&=&g(V_t)Q_t +\epsilon_t=g(\mathbf{0})Q_t + V_t \nabla g(\mathbf{0})Q_t + \frac{1}{2}V_t \nabla^2 g(\mathbf{0}) V_tQ_t \\
&=& \alpha_t + V_t \mathbf{f}_t + \frac{1}{2}V_tH_0V_tQ_t,
\end{eqnarray}
where the latent risk factors are given by $\mathbf{f}_t=\nabla g(\mathbf{0})Q_t$ and the quadratic term, defining the non-linear factor structure, requires the Hessian, $H_0$, evaluated at the origin. This Hessian contains the interaction effects w.r.t. to the risk premia (x-scores).

\subsubsection{State-space DPLS fundamental Factor Models}
For completeness, we can explore other variants of the DPLS fundamental factor model above -- specifically it can be written in a more general form which includes higher order lag terms up to maximum lag $r$. Specifically we posit the non-linear state-space model:

\begin{eqnarray*}
\mathbf{r}^e_{t}&=& U_{t} Q_{t} + \boldsymbol{\epsilon}^y_{t},\\
Z_{t-1}&=& V_{t}P_t + \boldsymbol{\epsilon}^x_t
\end{eqnarray*}
with the score state-vectors defined as
\begin{equation*}
U_t=
\begin{bmatrix} 
U_{t} \\ \cdots\\ U_{t-r}
\end{bmatrix}, \qquad  V_t=
\begin{bmatrix} 
V_{t} \\ \cdots\\ V_{t-r}
\end{bmatrix},
\end{equation*}
where $Q_t\in \mathbb{R}^{q\times Kr}$ and $P_t\in \mathbb{R}^{p\times Kr}$.
Note, as a special case, that the vector autoregressive form (a.k.a. semi-supervised learning) of the DPLS model (i.e. $X=Y$) is
the system of equations

\begin{eqnarray*}
X_{t}&=& U_{t}Q_t + \boldsymbol{\epsilon}^x_{t},\\
U_t&=&G_{\boldsymbol\theta_t}(U_{t-1}) + \mathbf{e}_t
\end{eqnarray*}
where $\mathbf{e}_t^\top=(\mathbf{\epsilon}_t, 0,\dots, 0)$.

If the projection matrix and parameters, $\boldsymbol\theta$ are further fixed in time then we see that the DPLS factor model is in fact a non-linear Dynamic Factor Model (DFM) \citep{16320} with projections on to scores:

\begin{eqnarray*}
X_{t}&=& U_{t}Q + \boldsymbol{\epsilon}^x_{t},\\
U_t&=&G_{\boldsymbol\theta}(U_{t-1}) + \mathbf{e}_t,
\end{eqnarray*}
where $U_t$ is the latent state variable of factors (a.k.a. score state vector). In the \cite{16320} model, $Q$ is a linear projection onto lagged latent factors and $G_{\boldsymbol\theta}(\cdot)$ is replaced with a static state-space update matrix. Their model requires stationarity. One important advantage of our DPLS model, aside from the ability to capture non-linearity in the latent state update, is that $G_{\boldsymbol\theta}$ can be an advanced recurrent neural network architecture such as LSTMs or GRUs, thus relaxing the requirement for stationarity \citep{doi:10.1080/00401706.2021.1921035}.
\section{Experiments} \label{sect:results}

This section presents the application of our DPLS factor model to all stocks in the Russell 1000 index, which represents the top 1000 companies by market capitalization in the United States.   The historical data covers an approximately 30 year period of monthly updates between January 1989 to November 2018, with a coverage universe of 3290 stocks. In every month, each stock includes 49 BARRA model factors --- 19 of which are fundamental factors and the remainder are GICS sector dummy variables (see Section \ref{sect:appendix} for a description of the factors).

Note that the BARRA factor model includes many more explanatory variables than used in our experiments below, but the purpose, here, is to illustrate the application of our framework to a larger set of factors. Equally, the BARRA factors are designed to explain variance in asset returns rather than capture alpha, however, our results focus both on expected asset returns and explained variance of asset returns based on conditional exposures to latent risk factors. Note, for avoidance of doubt, that our goal in this section is not to shed further light on the ``factor zoo problem'', but simply to advocate for the use of DPLS as a methodology for factor modeling based on its ability to capture outliers and attribute prediction error to conditional exposures to latent risk factors.

Our experiments broadly focus on three aspects of the DPLS factor model:

\begin{enumerate}
    \item Ability to capture outliers in monthly excess returns;
    \item Ability to generate higher performing managed portfolios under a risk adjusted return metric; and
    \item Ability to explain asset returns and risk via risk premia. 
\end{enumerate}

\paragraph{Data cleaning:} All stocks with missing factor exposures are removed from the estimation universe for the date of the missing factors only. To avoid excessive turn-over in the estimation universe over each consecutive period, we include all dropped symbols in the index over the 12 consecutive monthly periods. Further details of the data preparation, including sanitization to avoid violating licensing agreements, are provided in the documentation for the source code repository referenced on the first page.

\paragraph{Rescaling:} Training and testing is alternated in each period. For example, in the first historical month of the data, the model is fitted to the factor exposures and monthly excess monthly returns over the next period. To avoid look-ahead bias, all factors in each training period are normalized using only the training data, i.e. the resulting mean and standard deviation of each factor across all stocks is 0 and 1 respectively. The same normalization parameters are then used to rescale the test set. This avoids introducing a bias into the test set while also avoiding incorporating knowledge of the test data when rescaling the training data. Of course, the test set is not perfectly normalized.

\paragraph{Hyper-parameter tuning:} Cross-validation is performed using this cross-sectional training data, with approximately 1000 symbols. Once the model is fitted and tuned, we then apply the model to the factor exposures in the next period, t+1, to predict the excess monthly returns over [t+1, t+2]. The process is repeated over all periods. Note, for computational reasons, we can avoid cross-validation over every period and instead stride the cross-validation, every other say 10 periods, relying on the optimal hyper-parameters from the last cross-validation periods for all subsequent intermediate periods. In practice we find that performing cross-validation in the first period only is adequate with the exception of the number of latent risk factors, $K$.

Three fold cross-validation over $\{50,100,200\}$ hidden units per layer, $\lambda_1\in\{0,10^{-3},10^{-2}, 10^{-1}\}$ and $\{1,2,3\}$ hidden layers is performed in the same time period. We find that the optimal architecture has two hidden layers, 100 units per layer, no $L_1$ regularization, and softplus (i.e. smooth ReLU) activation. For an interpretable model, $C^2(\mathbb{R})$ continuous activation functions are required to provide sufficient smoothness (with 50 hidden units per layer being optimal). 

All experiments are performed using \verb|Tensorflow| \citep{abadi2016tensorflow} to implement the deep feed-forward network and the \verb|pls| R package \citep{pls} for PLS. The LASSO model is implemented using the \verb|glmnet| R package \citep{glmnet}. The amount of $L_1$ regularization in the LASSO model is tuned every period\footnote{The LASSO model uses cross-validation to optimize the $L_1$ regularization parameter and iterative fitting, with 50 alphas, along the regularization path.}. The PCA results are generated using the asymptotic PCA method implemented in the \verb|covFactorModel| R package\footnote{\url{https://rdrr.io/github/dppalomar/covFactorModel}}.

\subsection{Russell 1000 stocks}

We begin by analyzing the predicted excess assets returns of the models over all constituents of the Russell 1000 index in each period\footnote{Note, for avoidance of doubt, that we do not construct index tracking portfolios -- all stocks in the index are equally weighted in our analysis.}, with the goal of identifying the ability of DPLS to capture outliers in the excess stock returns.
Figure \ref{fig:errors} compares the in-sample and out-of-sample performance of DPLS, PLS, deep feedforward network (NN), and LASSO regression using the $L_\infty$ norm of stock return errors over the coverage universe of stocks listed in the Russell 1000 index at each monthly period. This norm measures the largest absolute error across all assets, i.e. the error associated with the most mis-predicted excess asset return in each period.  The time averaged $L_\infty$ norms, over all periods, is shown in parentheses.

We observe the ability of DPLS to best capture outliers, with the $L_\infty$ norm of the error in the NN and DPLS regression respectively being on average between 8\% and 12\% smaller than the other methods out-of-sample. We also observe that the amount of variance between the in-sample and out-of-sample performance is most notable using NN regression. This is consistency with the high degree of parameters exhibited by the full neural network architecture relative to the training data size.  This, in fact, highlights a substantial defect of using deep learning for period-by-period cross-sectional regression --- the amount of training data required for deep learning required isn't compatible with factor modeling cross-sectional datasets. On the other hand, if one simply trains across all stocks and all periods, then the model loses its dynamic nature.  

We also observe periods in the out-of-sample test set where PLS exhibits the largest errors and DPLS appears to capture these. The combination of a lower time averaged $L_\infty$ error and fewer apparent outliers not captured by the PLS, provides evidence of the superiority of DPLS over PLS in capturing outliers.


\begin{figure}[H]
\centering
\includegraphics[width=\textwidth]{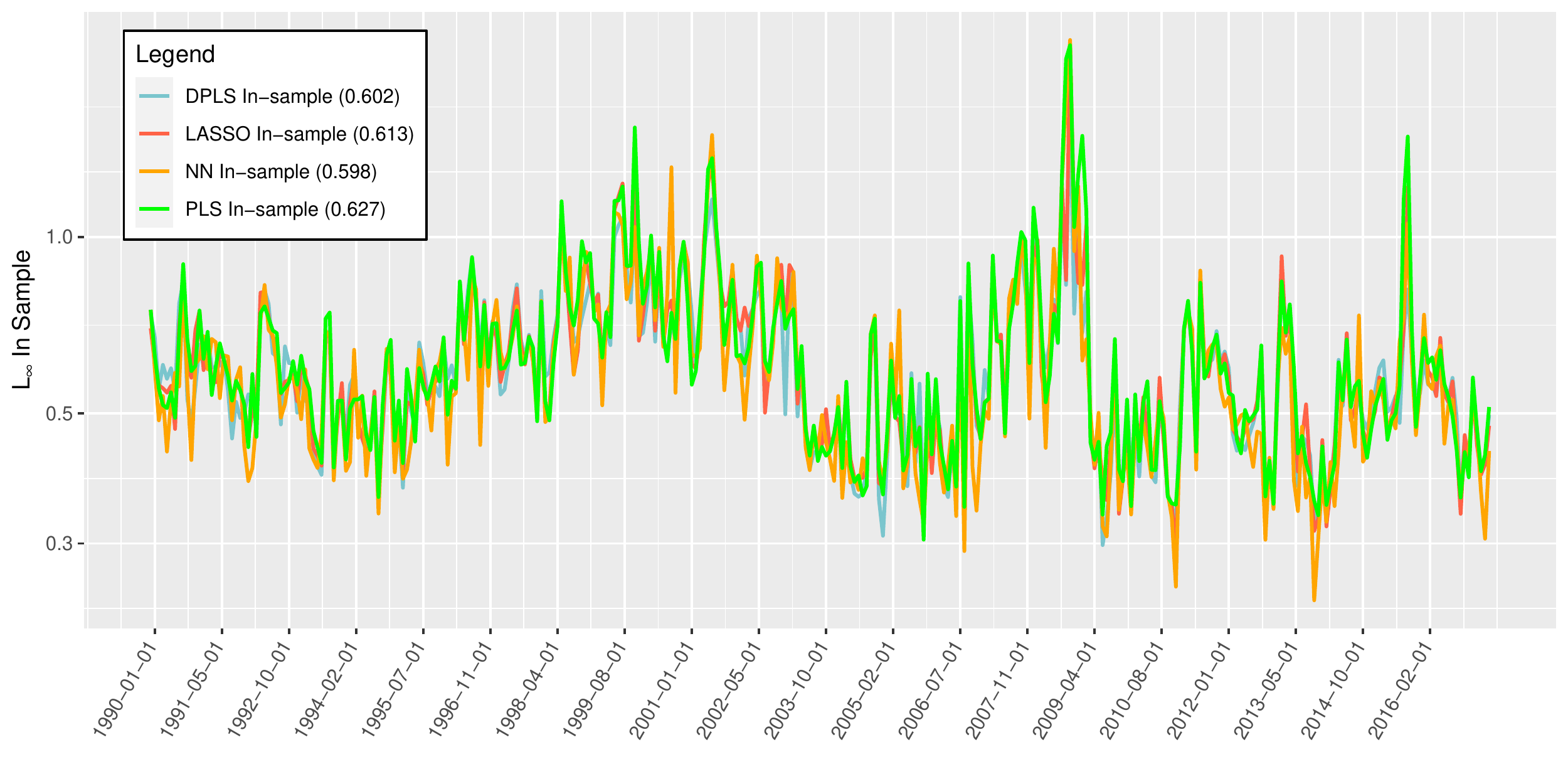}\\ 
(a) $L_\infty$ in-sample error\\
\includegraphics[width=\textwidth]{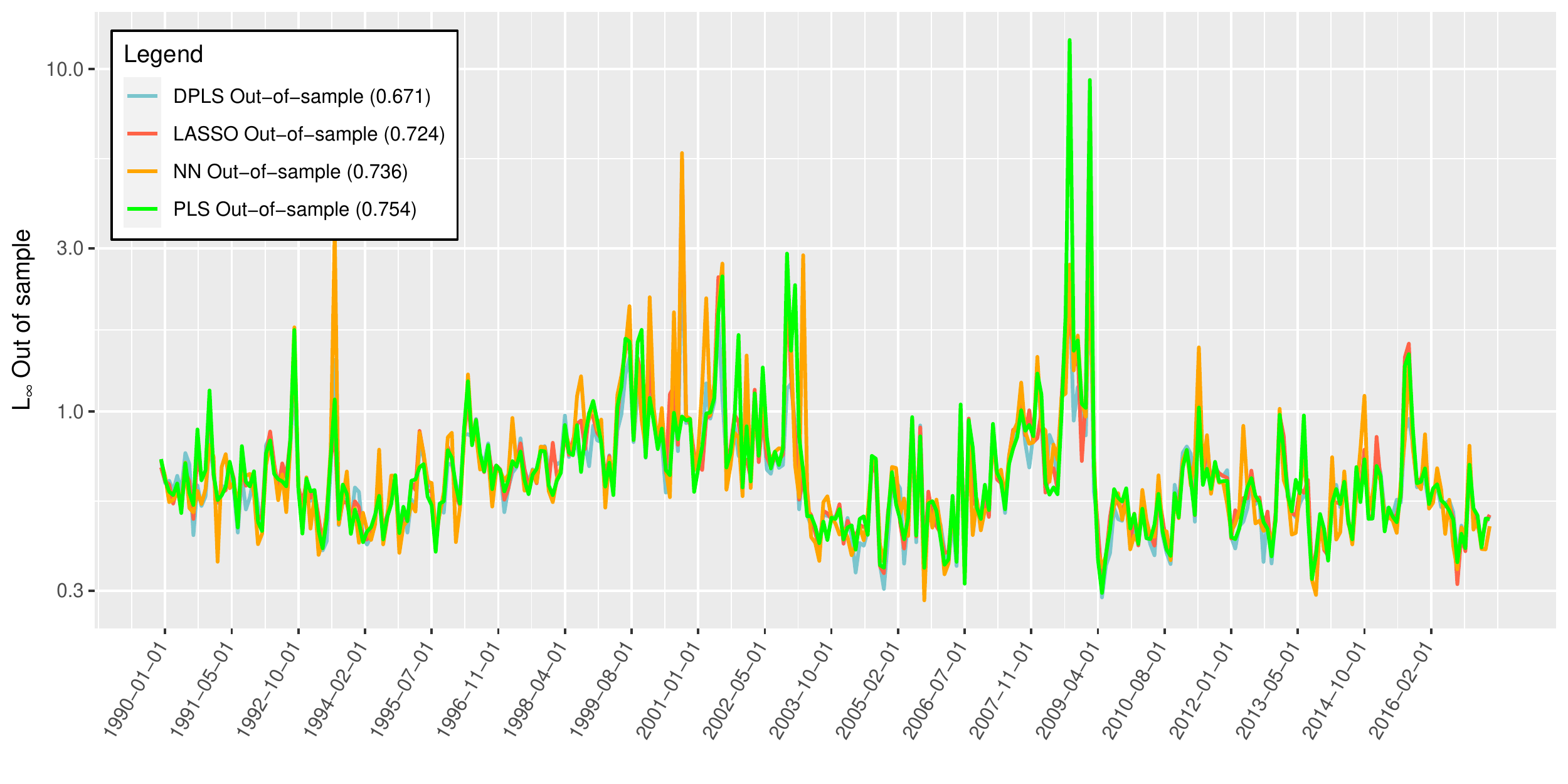}\\
 (b) $L_\infty$ out-of-sample error\\
\caption{\textit{The (a) in-sample and (b) out-of-sample stock returns error, under the $L_\infty$ norm, is compared between DPLS, PLS, NN, and LASSO regression applied to a coverage universe of 3290 stocks from the Russell 1000 index over the period from January 1990 to November 2018. The time averaged $L_\infty$ error is shown in parentheses.  }}
\label{fig:errors}
\end{figure}
Figure \ref{fig:errors_r2} compares the in-sample and out-of-sample performance of DPLS, NN, and LASSO regression using the $R^2$ of monthly asset returns. The time averaged $R^2$ over all periods is shown in parentheses. Overall, we observe that DPLS regression exhibits the best out-of-sample performance. We also find further evidence that the NN is over-fitting, on account of the low out-of-sample $R^2$.

\begin{figure}[H]
\centering
\includegraphics[width=\textwidth]{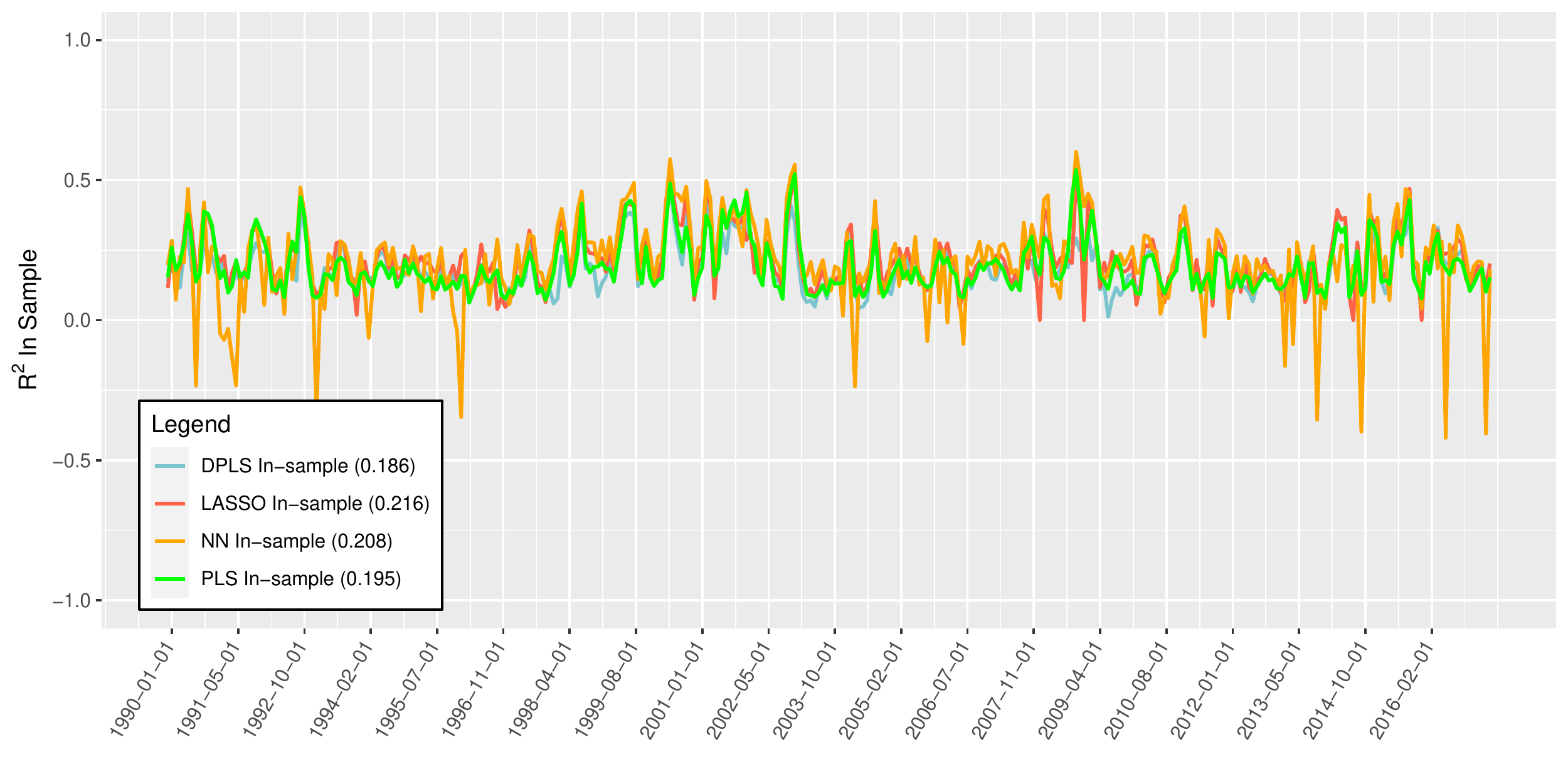} \\
(a) In sample $R^2$\\
\includegraphics[width=\textwidth]{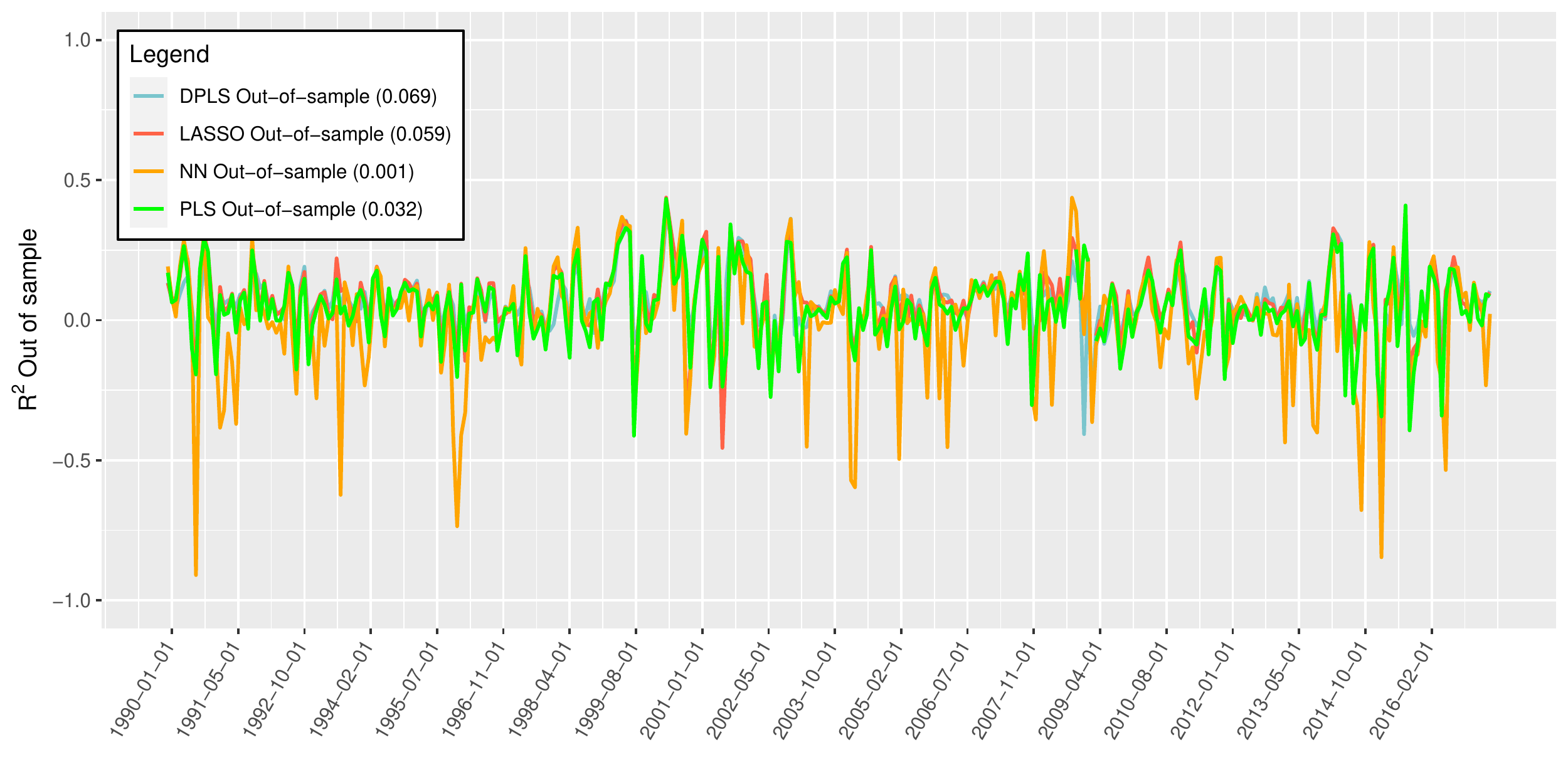}\\
(b) Out-of-sample $R^2$ \\
\caption{\textit{The (a) in-sample and (b) out-of-sample stock returns error, under the $R^2$ metric, is compared between DPLS, PLS, NN, and LASSO regression applied to a coverage universe of 3290 stocks from the Russell 1000 index over the period from January 1990 to November 2018. The time averaged $R^2$ error is shown in parentheses. }}
\label{fig:errors_r2}
\end{figure}

Figure \ref{fig:sensitivity} compares the time averaged factor model sensitivities, evaluated at $Z_{t-1}=\mathbf{0}$, over the thirty year period using (a) full NN and (b)  DPLS regression. For each method, the sensitivities are each sorted in descending order from top to bottom by their mean value over all periods. See Table \ref{tab:r3000} in Appendix \ref{sect:appendix} for a description of the factors. 
\begin{figure}[H]
\centering
\begin{tabular}{cc}
\includegraphics[width=0.48\textwidth]{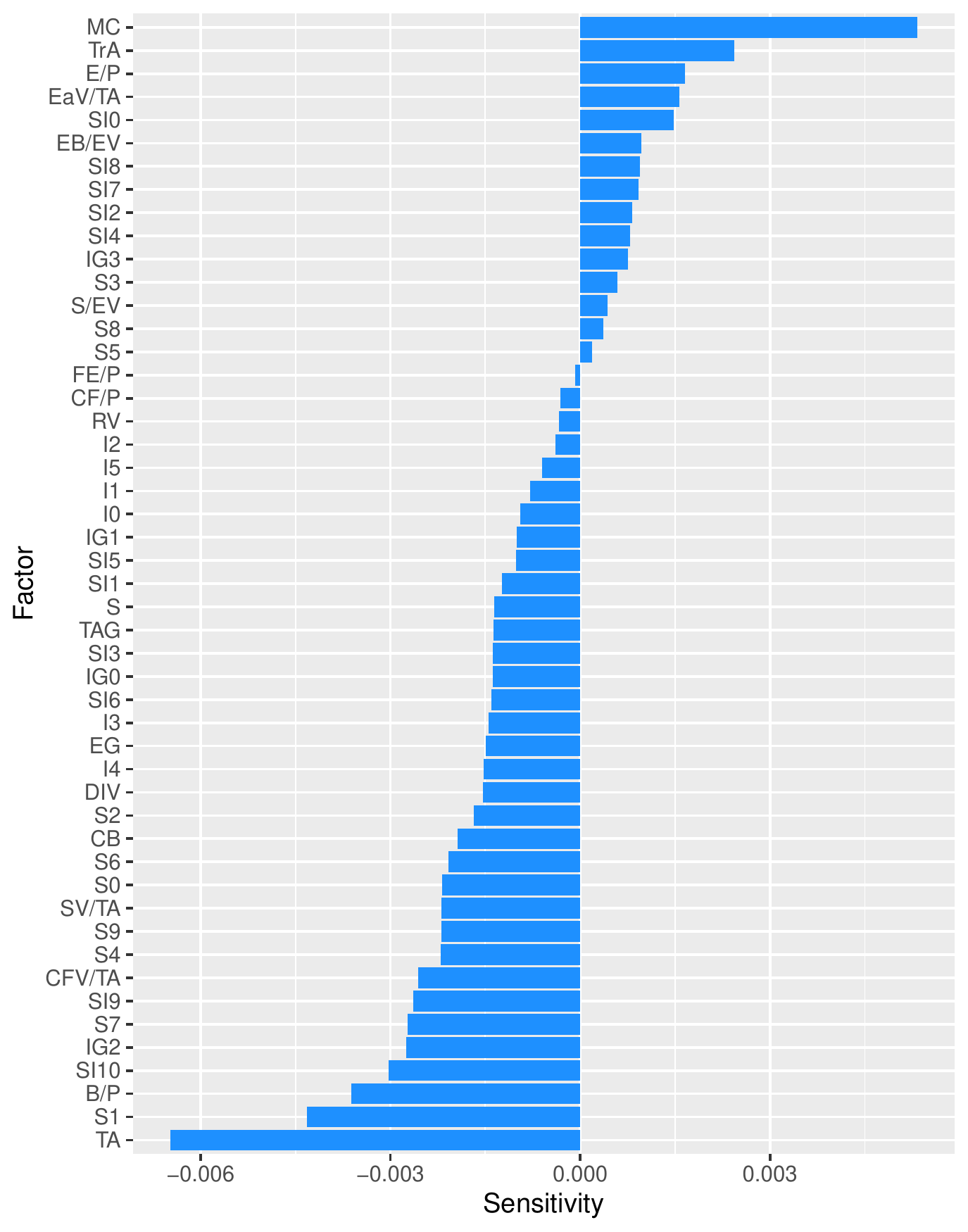}
&
\includegraphics[width=0.48\textwidth]{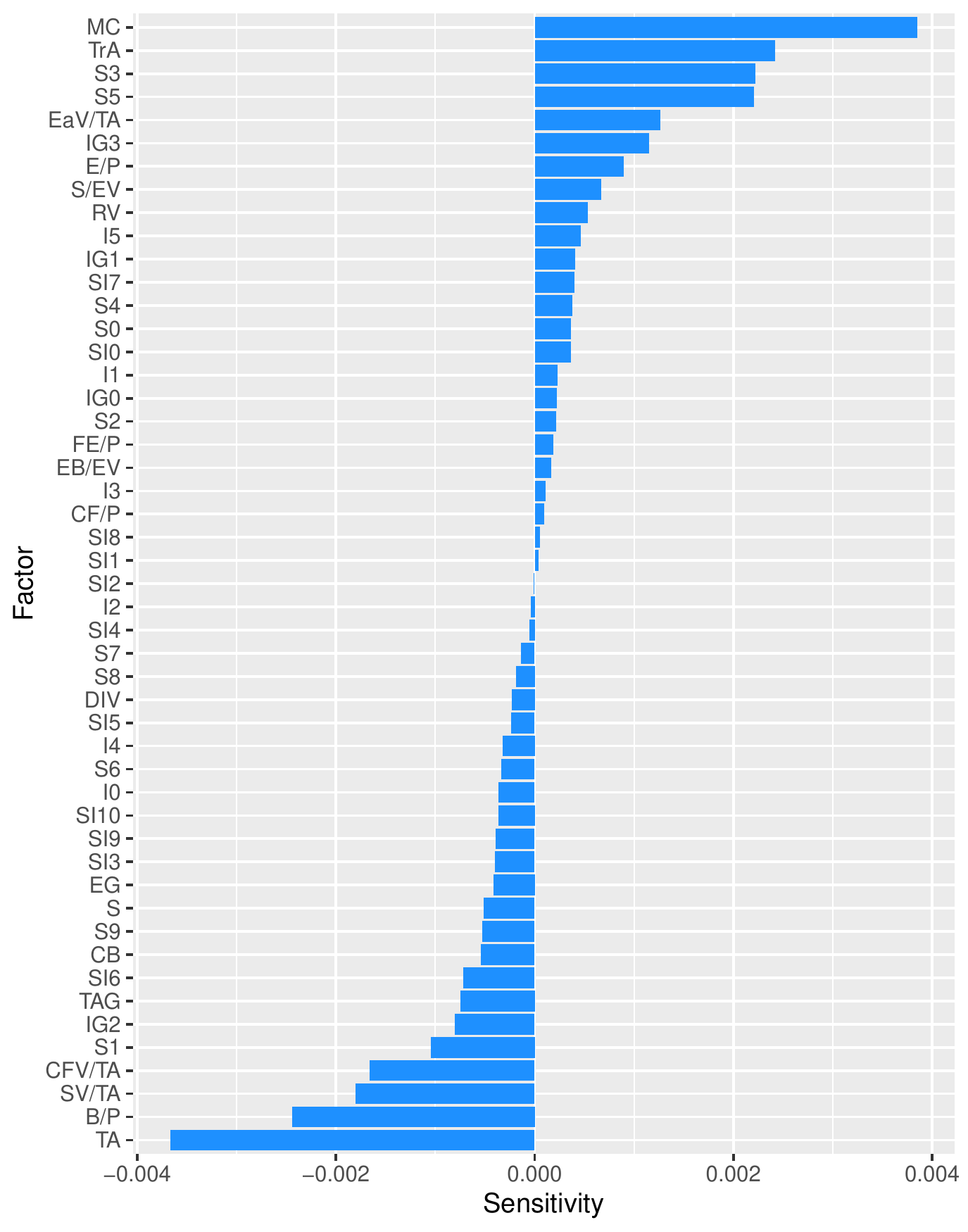}\\
(a) NN & (b) DPLS
\end{tabular}
\caption{\textit{The time averaged factor model sensitivities over the thirty year period using (a) NN and (b) DPLS regression applied to the constituents of the Russell 1000 index in each time period. Market Cap (MC), trading activity (TrA), and Earnings Volatility to Total Assets (EaV/TA) are ranked highly as positive factors by each model. Additionally, both models agree on Total Assets (TA), Book-to-Price (B/P), and Cash flow Volatility to Total Assets (CFV/TA) as the highest ranked negative factors contributing to asset returns. See Table \ref{tab:r3000} for a description of the factors.}}
\label{fig:sensitivity}
\end{figure}

All three methods share important consistencies. For example, Market Cap (MC), Trading Activity (TrA), and Earnings Volatility to Total Assets (EaV/TA) rank highly as positive factors. Conversely, Total Assets (TA),   Book-to-Price (B/P), and Cash flow Volatility to Total Assets (CFV/TA) rank highly as negative factors. Both models also agree that expected returns are, on average, largely independent of Cash Flow to Price (CF/P). 

However, there are also significant discrepancies between the methods. For example EBIDTA to Earnings Volatility (EB/EV) ranks highly as a positive factor in the NN model model, yet is weakly positive in the DPLS model. As another example, Total Assets Growth (TAG) is ranked as one of the most negative factors in the DPLS model, yet  the NN model places it as weakly negative. The NN model is observed to exhibit sensitivities which are up to approximately 50\% larger than in the DPLS model, suggesting that the higher degree of parameterization leads to factor over-sensitivity. The NN sensitivities are also, on average, skewed negative. Another apparent difference is the effect of dimension reduction in the DPLS model -- there are many more non-zero fundamental factors in the NN model, whereas the mid-section of the DPLS plot is comparatively thin, indicating fewer material factors on average.


\begin{figure}[H]
\centering
\includegraphics[width=\textwidth]{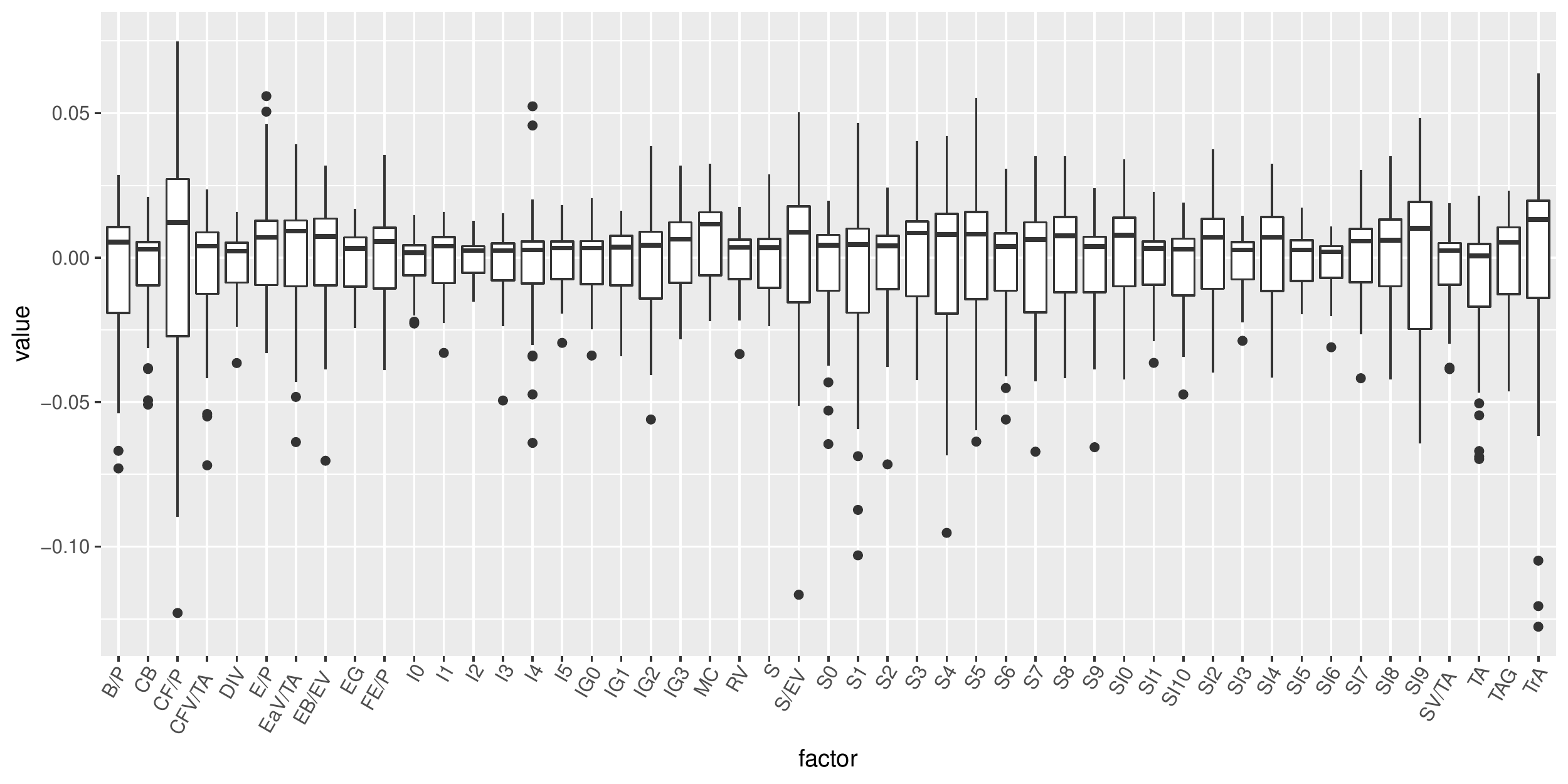}
(a) Temporal distribution of NN factor sensitivities
\\
\includegraphics[width=\textwidth]{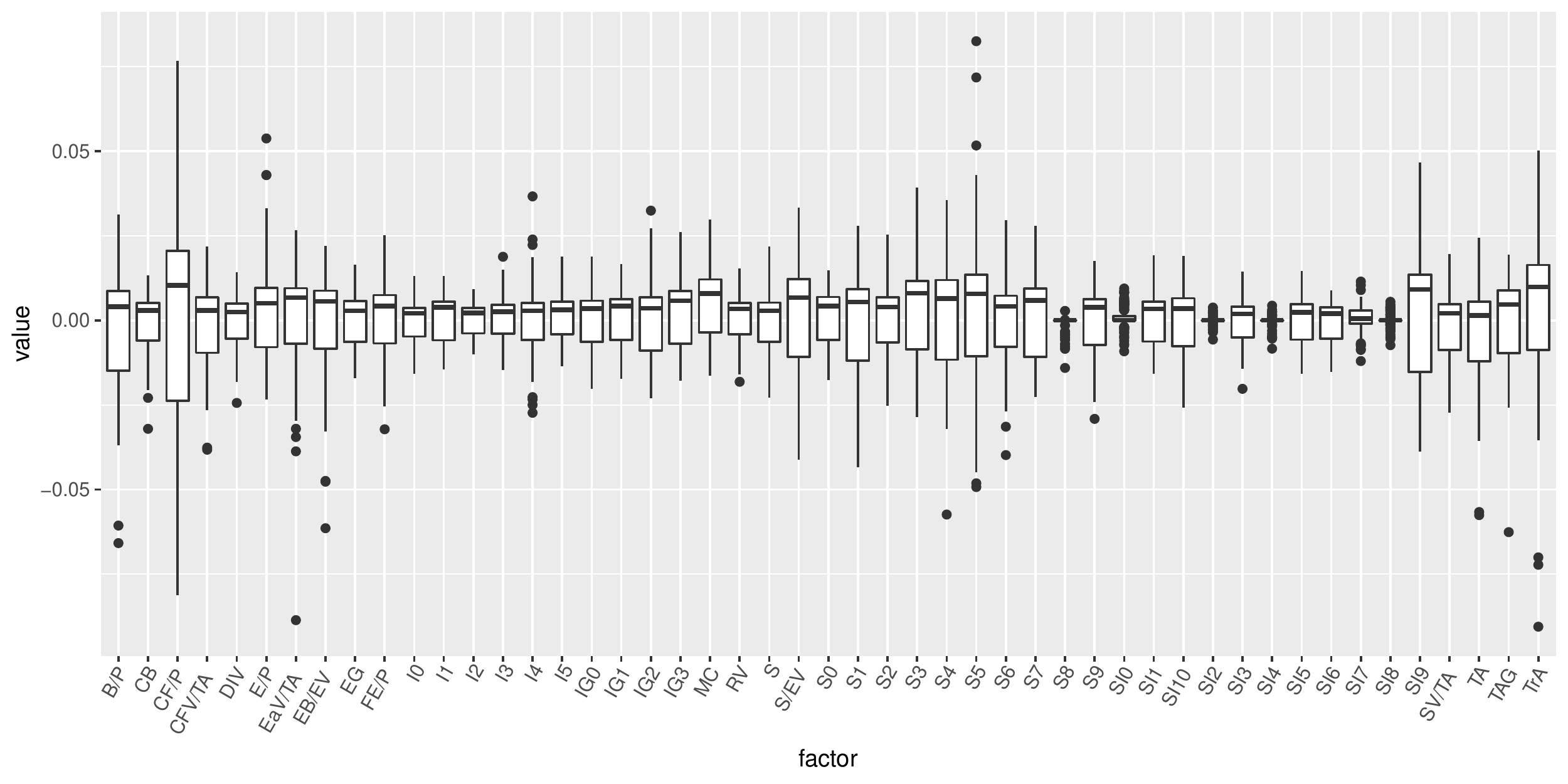}\\
(a) Temporal distribution of DPLS factor sensitivities

\caption{\textit{The temporal distribution of factor model sensitivities over the thirty year period using (a) NN and (b) DPLS regression applied to the constituents of the Russell 1000 index in each time period. The medians are shown by the horizontal solid lines. We observe that Cash flow to Price and sub-industry group 9 carry the most uncertainty - the inter-quartile range is the largest in both models. DPLS also collapses a number of the sector and sub-industry group variables closely around zero.}}
\label{fig:temporal_sensitivity}
\end{figure}

The time averages provide a limited view on the model sensitivities. We additionally plot the temporal distribution of the factor sensitivites.
We observe that Cash flow to Price and sub-industry group 9 carry the most uncertainty - the inter-quartile range is the largest in both models. DPLS also collapses a number of the sector and sub-industry group variables closely around zero.

We can also study the extent to which the most dominant factors persist over multiple periods. Figure \ref{fig:heatmap_sensitivities} compares the sensitivity of the predicted excess asset returns to the fundamental factors over the two year period between October 2007 and October 2009, covering the Lehman Brothers financial crash. We find consistency between the factor sensitivities and the September/October 2008 Lehman collapse period. In particular, DPLS is able to identify a number of prominent factors with large positive sensitivities in this period. In contrast, the NN model only identify one strongly positive sensitivity. We also see more variability from period to period in the NN model, whereas the DPLS exhibits more observations where the factor dominance persists over multiple periods.

\begin{figure}[H]
\centering
\begin{tabular}{cc}
\includegraphics[width=0.5\textwidth]{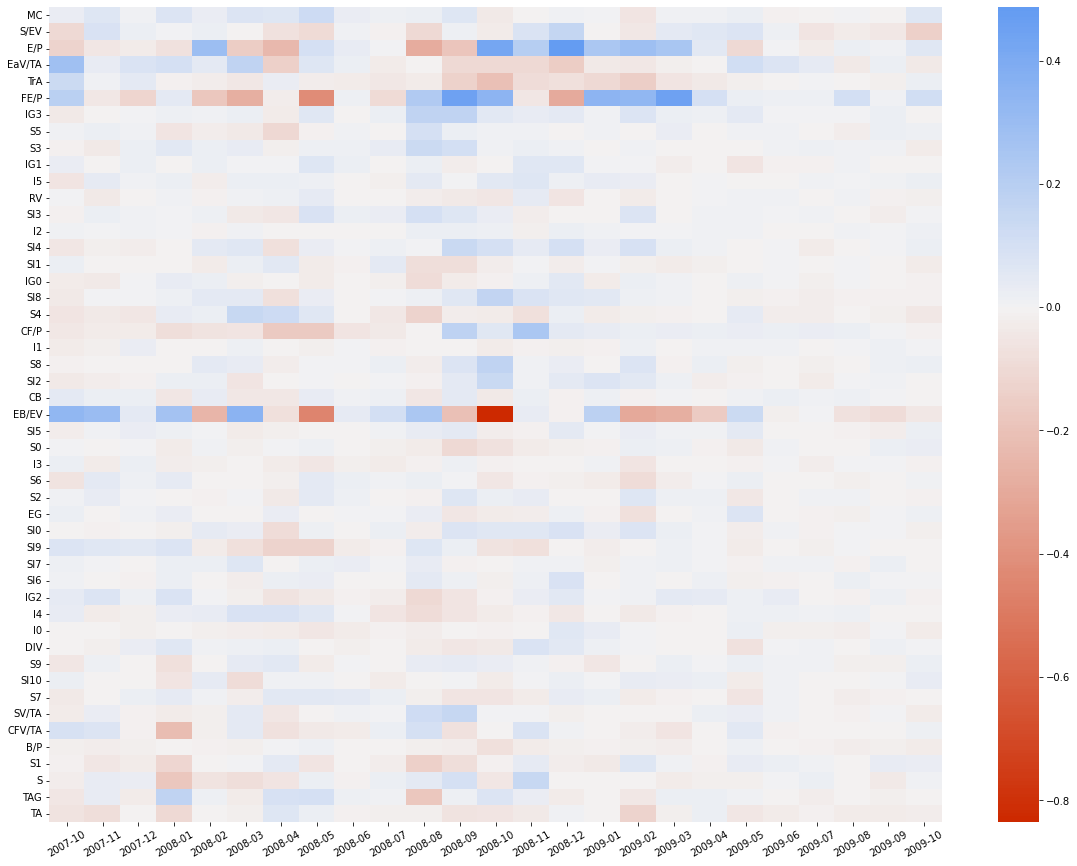} &
\includegraphics[width=0.5\textwidth]{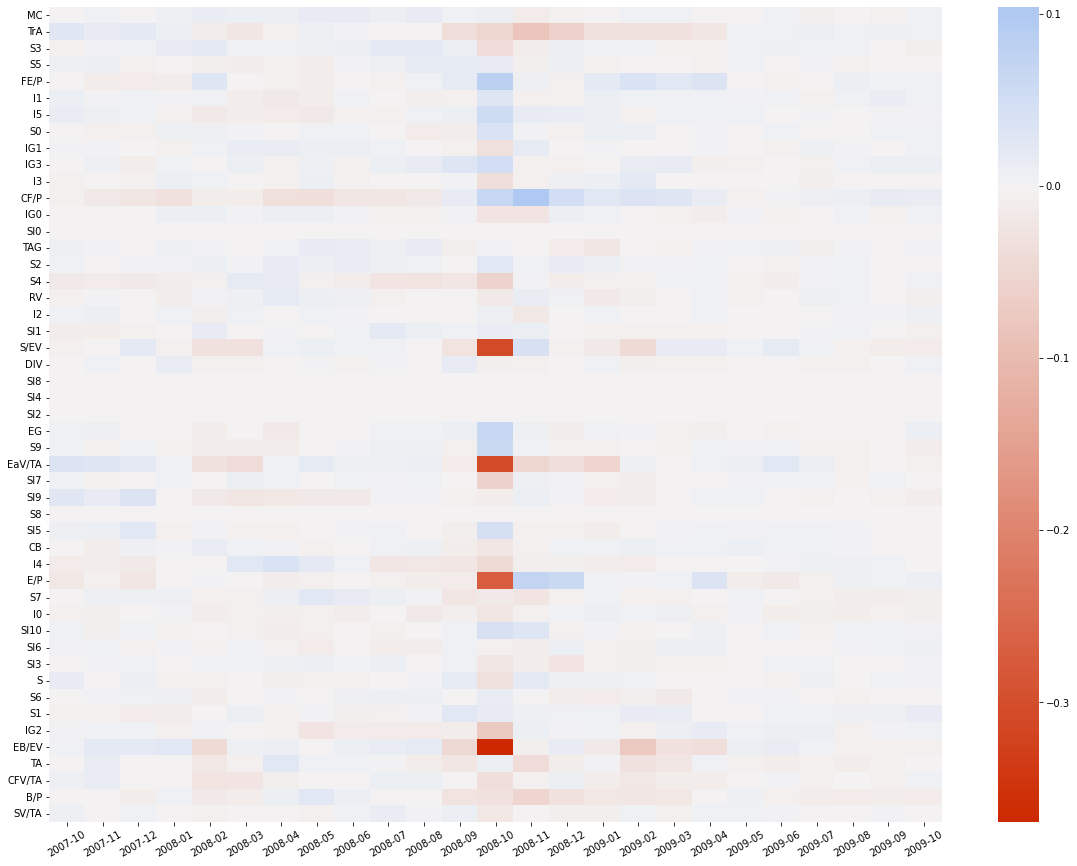}\\
(a) NN  & (b) DPLS\\
\end{tabular}
\caption{\textit{The dynamical model sensitivities over a two year period between October 2007 and October 2009 covering the Lehman Brothers financial crash. The factors are ranked by their time averaged sensitivities from top to bottom in descending order.}}
\label{fig:heatmap_sensitivities}
\end{figure}

Further insight can be gained by analyzing the interaction effects and quadratic terms in the NN and DPLS models, as shown in Figure \ref{fig:interaction}. These interaction terms are time averaged over the thirty year period using (a) NN and (b) DPLS regression applied to all stocks in the Russell 1000 index. $XX-YY$ denote pairwise interaction effects between factors XX and factors YY. 

In contrast to the sensitivities, the most important interaction terms are different across each model. A common pattern is for a model to find a fundamental ratio across multiple sectors and industry groups. For example, the DPLS model finds the quadratic Book-to-Price ratio positively dominant, but also finds Book-to-Price across SubIndustry 5 in addition to Industry Groups 0 and 3. The DPLS model also finds Cashflow/Price strongly negatively dominant across Industry 4, SubIndustry 7, and Sectors 7 and 9. The NN model finds Earnings Volatility/Total Assets (EaV/TA), Cash Flow Volatility/Total Assets (CFV/TA), Total Asset Growth (TAG), Total assets (TA), and Market Cap (MC) positively dominant in specific SubIndustries, Sectors, and Industry Groups. Whereas Dividends (DIV) are exclusively associated with negative effects in specific SubIndustries, Sectors, and Industry Groups. We also note that DPLS isolates a few particularly dominant interaction pairs - the interaction of Book-to-Price and Earnings-to-Price is a strong negative effect and the Industry 4-Sector 0 pair is strongly positive.

The top interaction effects in the NN model are also an order of magnitude larger than in the DPLS model, suggesting that the higher degree of parameterization renders the NN interaction effects as over-sensitive.

\begin{figure}[H]
\centering
\begin{tabular}{cc}
\includegraphics[width=0.5\textwidth]{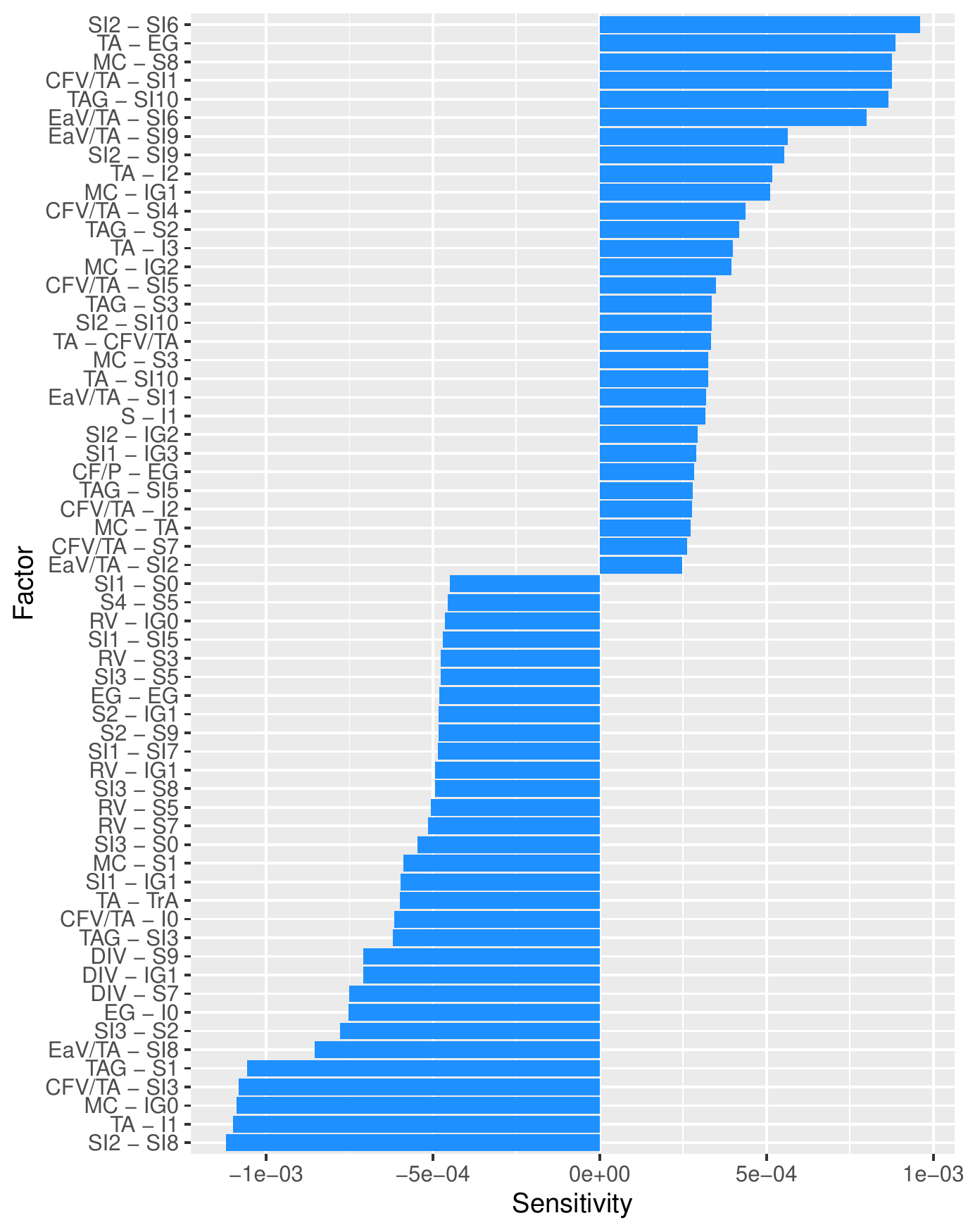}
&
\includegraphics[width=0.5\textwidth]{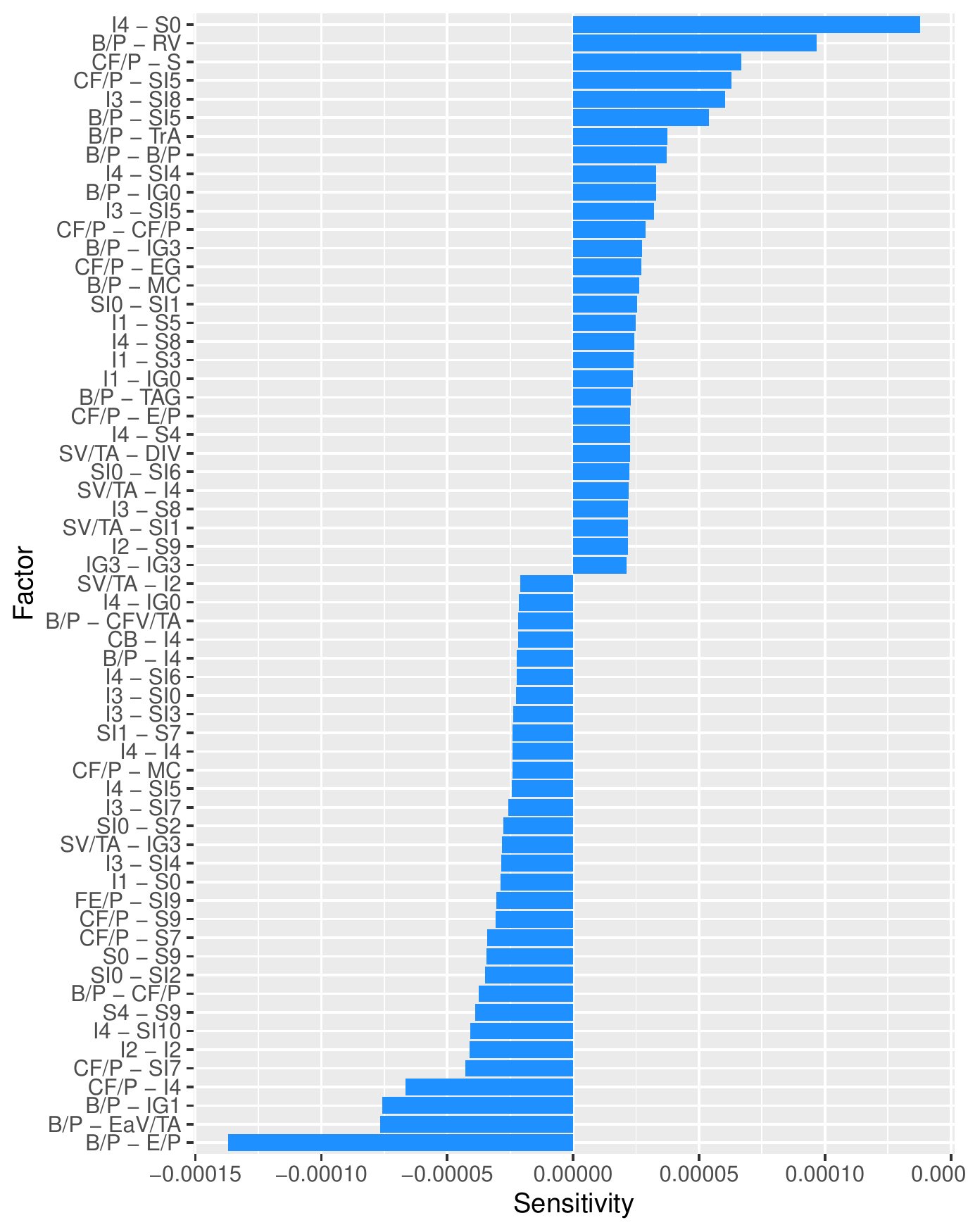}\\
(a)  NN &  (b) DPLS
\end{tabular}
\caption{\textit{The most important pairwise interaction effects, time averaged over a thirty year period using all stocks.}}
\label{fig:interaction}
\end{figure}

The temporal distribution of interaction effects can shed light on the degree of variability in interaction effects over time. Figure \ref{fig:temporal_interaction} compares the NN and DPLS interaction terms, listed in alphabetic order from left to right. In contrast to the NN model, most of the DPLS interaction effects are tightly distributed with a few notable exceptions - Book-to-Price$-$Earnings/Price (B/P-E/V), Cash flow to price$-$Industry 4 (CF/P-I4) and cash flow to price$-$SubIndustry 5 (CF/P-SI5). 

The NN interaction effects, on the other hand, are much more widely distributed and it more difficult, if not impossible, to draw meaningful conclusions from the time averaged ranking of the interaction effects, given there is so much variability. 
\begin{figure}[H]
\centering
\includegraphics[width=\textwidth]{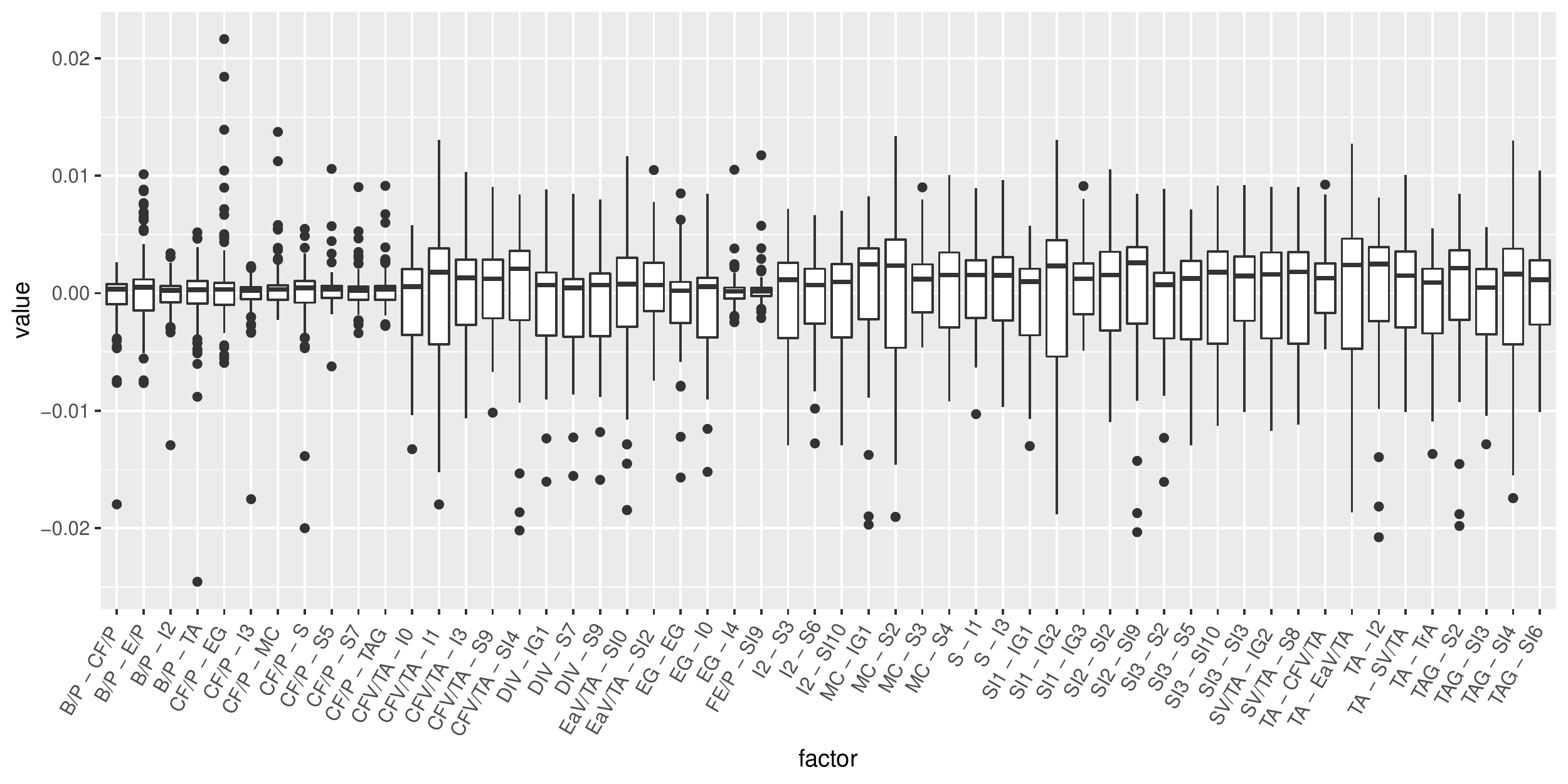}
(a) Temporal distribution of NN factor interactions 
\\
\includegraphics[width=\textwidth]{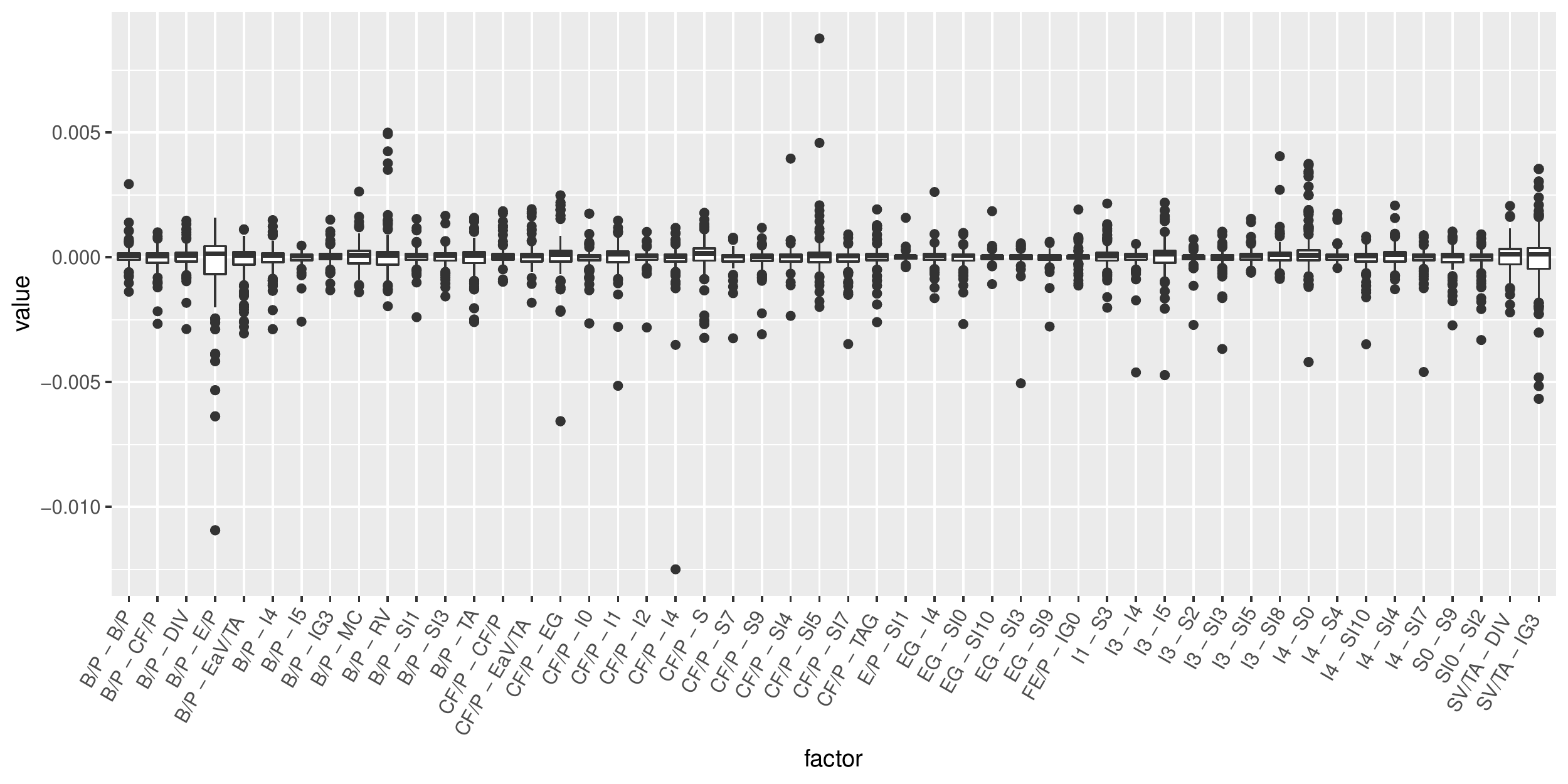}\\
(a) Temporal distribution of DPLS factor interactions

\caption{\textit{The temporal distribution of factor model interactions over the thirty year period using (a) NN and (b) DPLS regression applied to the constituents of the Russell 1000 index in each time period. The medians are shown by the horizontal solid lines.}}
\label{fig:temporal_interaction}
\end{figure}

 Here we observe much less consistency between the NN and the DPLS model. The top interaction effects in the NN model are also an order of magnitude larger than in the DPLS model, suggesting that the higher degree of parameterization renders the NN interaction effects as over-sensitive. 

\begin{figure}[H]
\centering
\begin{tabular}{cc}
\includegraphics[width=0.48\textwidth]{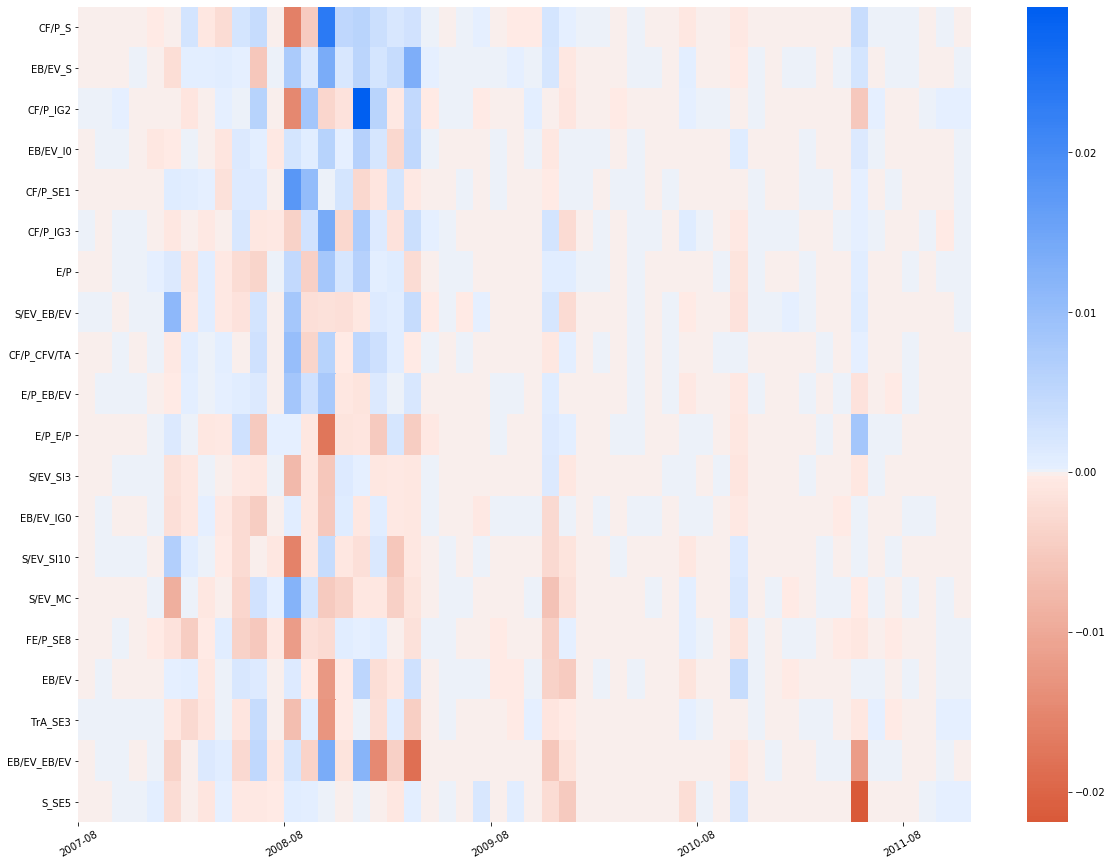}
&
\includegraphics[width=0.48\textwidth]{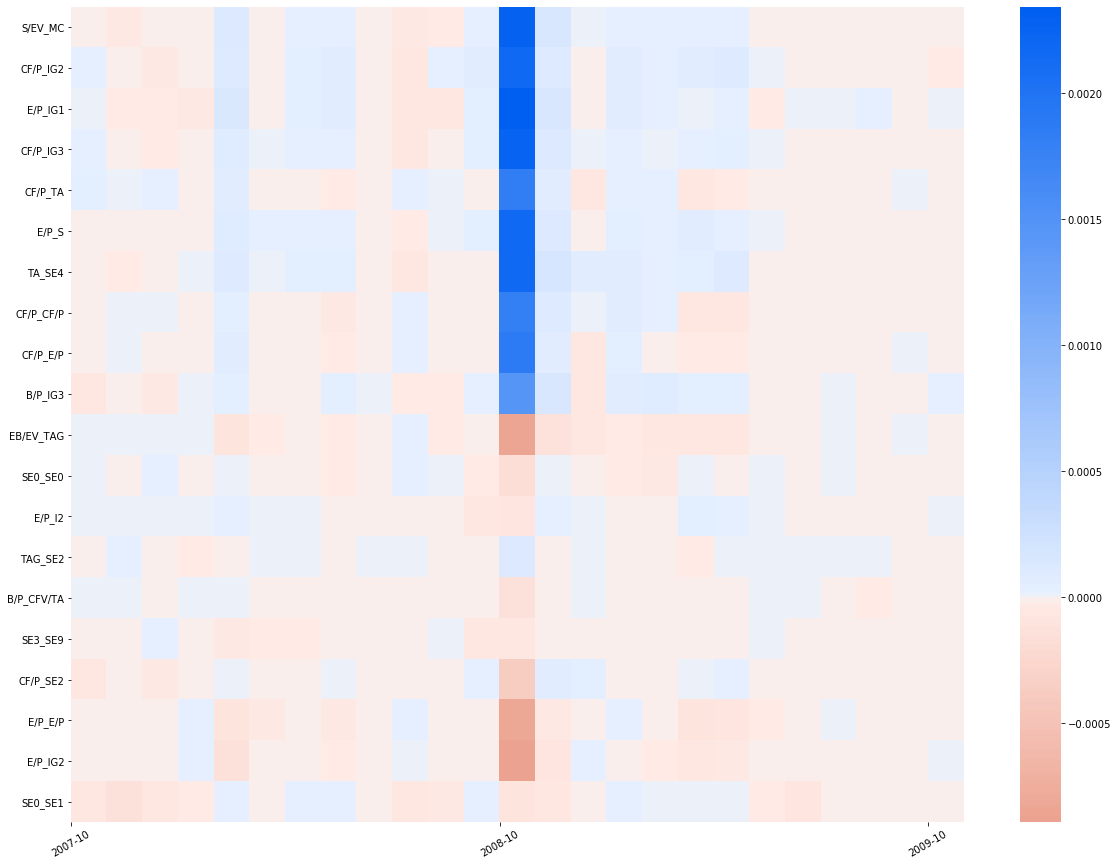}\\
(a) NN & (b) DPLS\\
\end{tabular}
\caption{\textit{The interaction effects (a) of the NN model over a four year period between October 2007 and October 2011 and (b) of DPLS over a two year period between October 2007 and 2009, both periods covering the Lehman Brothers financial crash. The factors are ranked by their time averaged interactions from top to bottom in descending order.}}
\label{fig:heatmap_interaction}
\end{figure}

Figure \ref{fig:heatmap_interaction} compares the most important interaction and quadratic effects over a multi-year period covering the Lehman Brothers financial crash. The DPLS interaction terms in October of 2008 dominant the other dates and largely govern the overall time averaged ranking. In contrast, the NN exhibits multiple outlier dates and more severe swings between an interaction effect being positive dominant and negative dominant from month to month. For example, in Figure \ref{fig:heatmap_interaction}(a), the interaction effects in the first and third row swing between dark red and dark blue in a relatively short period of time.

\clearpage
\subsection{Portfolio analysis}


We now evaluate the utility of the DPLS factor model as a predictive signal for constructing portfolios. 

To evaluate the relative portfolio performance under a DPLS factor model versus a PLS, LASSO or NN factor model, we construct an equally weighted, long-only, portfolio 
of $n$ stocks with the highest predicted monthly returns in each period. The portfolio is reconstructed each period, always remaining equal but varying in composition. This is repeated over the 30 year period in the data. We then estimate the information ratios from the mean and volatility of the excess monthly portfolio returns, using the Russell 1000 index as the benchmark.

This exercise is merely to assess the relative merits of DPLS factor models versus other methods as is not intended as an exhaustive study of the implications of using DPLS factor models fro portfolio construction. One important caveat is that the information ratio does not account for trading costs, which could erode performance due to high turnover. Additionally we do not consider more complex constraints such as sector balancing of the managed portfolios.

\begin{figure}[H]
\centering
 \includegraphics[width=\textwidth]{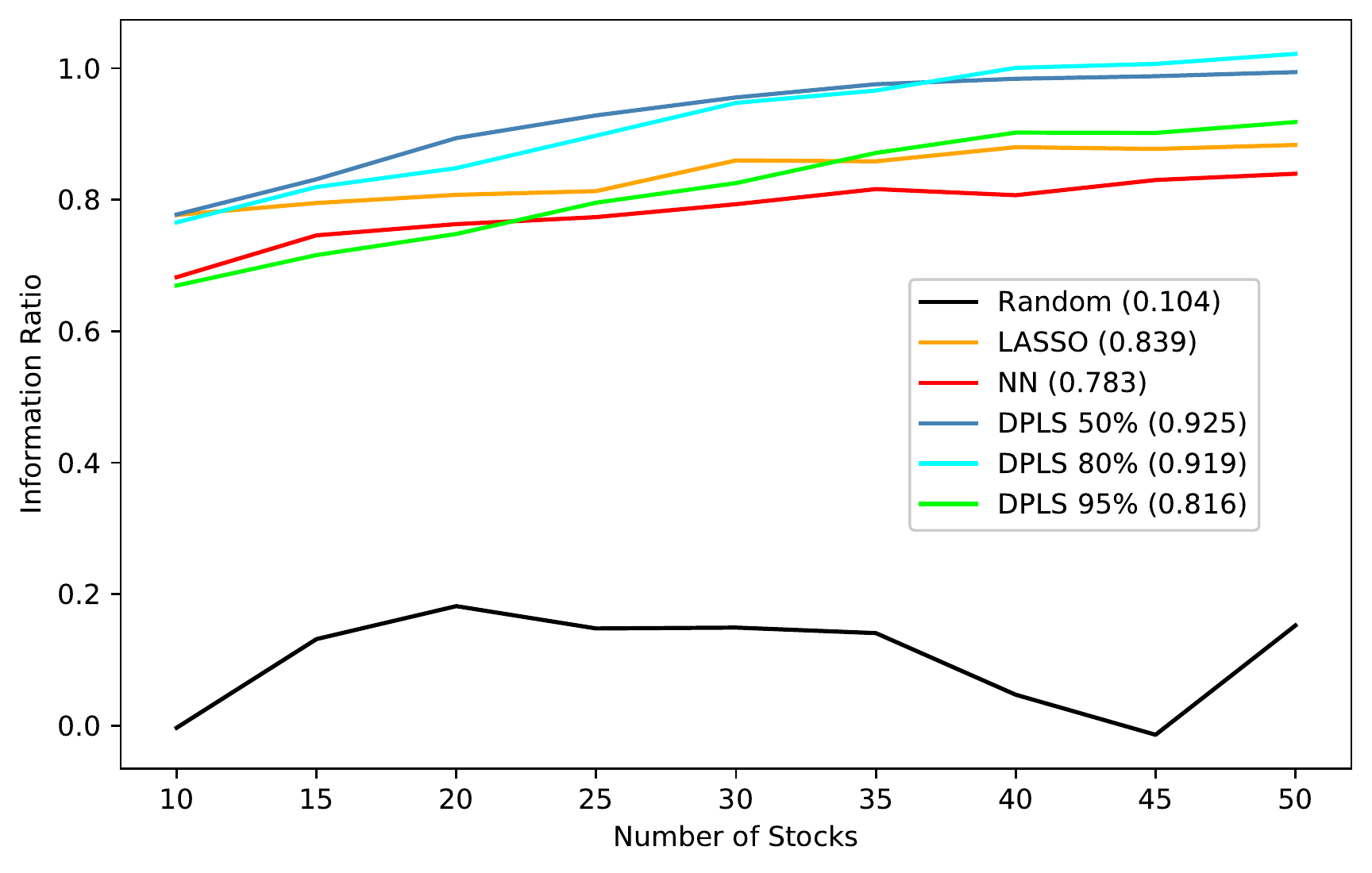}\\

\caption{\textit{The information ratios of a portfolio selection strategy which selects the $n$ stocks from the universe with the highest predicted monthly returns. The information ratios are evaluated for various equally weighted, long-only, portfolios whose size is shown by the x-axis. }}
\label{fig:IR}
\end{figure}

Figure \ref{fig:IR} compares the information ratios using DPLS, NN, and LASSO regression to identify the stocks with the highest predicted returns. The information ratios are evaluated for equally weighted portfolios of various sizes. Also shown, for control, are randomly selected portfolios, without the use of a predictive signal. The mean information ratio for each model, across all portfolios, is shown in parentheses. We observe that the information ratio of the portfolio returns, using DPLS regression, is approximately 1.2x greater than NN regression. NN regression, despite exhibiting substantially lower out-of-sample error than LASSO regression results in an inferior information ratio than LASSO, suggesting that stock selection based on more accurate predicted returns does not directly translate to higher information ratios on account of the stock return covariance. 

We also observe that the information ratio of the baseline random portfolio is small, but not negligible, suggesting sampling bias and estimation universe modification have a small effect.

Figure \ref{fig:factor_tilts_heatmap} shows the time averaged factor tilts of equally weighted portfolios constructed from the predicted top performing 50 stocks, in each monthly period, over all time periods. 

Most of the factor tilts are consistent across the models, however, some are not. For example, Market Cap (MC), Total Asset Growth (TAG), and Dividends (Div).  LASSO and NN regression overweight the TAG. LASSO also overweights Market Cap. DPLS (95\%) and NN regression overweight Dividends.  For completeness, the changes in factor tilts over each period are provided in Section \ref{sect:factor_heatmaps}.

\begin{figure}[H]
\centering
\includegraphics[width=0.7\textwidth]{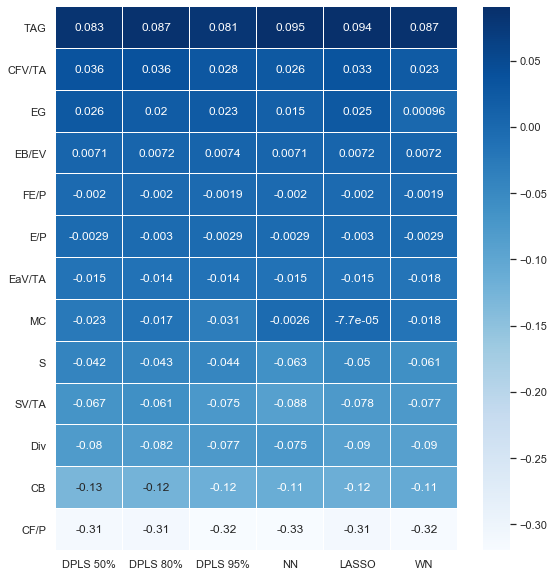}

\caption{\textit{The rescaled factors tilts of equally weighted portfolios constructed from the predicted top performing 50 stocks and then averaged over time.}}
\label{fig:factor_tilts_heatmap}
\end{figure}



Figure \ref{fig:sector_tilts_heatmap} shows the rescaled sector tilts of equally weighted portfolios constructed from the predicted top performing 50 stocks, and then averaged over all periods. The sectors are ranked by their time averaged ratios, but their tilts vary each month due to portfolios turn-over. WN denotes a random equally weighted portfolio. 

\begin{figure}[H]
\centering
\includegraphics[width=0.7\textwidth]{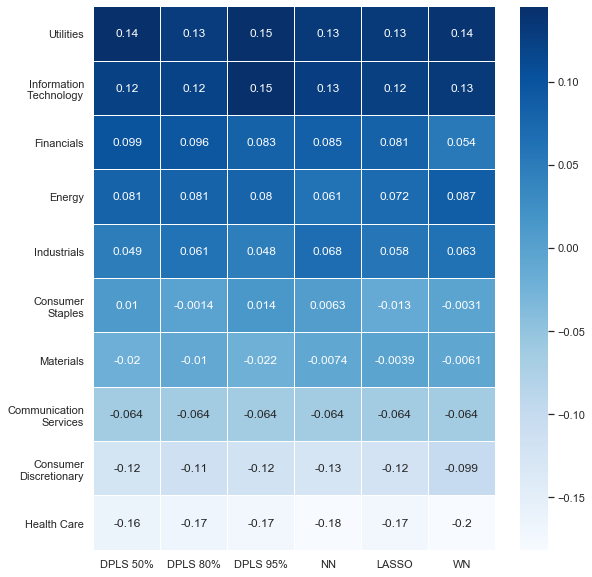}

\caption{\textit{The rescaled sector tilts of equally weighted portfolios constructed from the predicted top performing 50 stocks and then averaged over time.}}
\label{fig:sector_tilts_heatmap}
\end{figure}


Utilities and Information Technology are the most dominant sectors, i.e. with the highest time averaged representation. In contrast, Consumer Discretionary and Health Care are the least representative sectors shown. Note that the least representative sector, Real Estate, has been excluded. The sector tilts across DPLS, NN, and LASSO regression are found to generally concur. Notable differences include DPLS (95\%) favoring portfolios with a higher Information Technology and Consumer Staples tilt and lower Industrial tilt. For completeness, the changes in sector tilts over each period are provided in Section \ref{sect:sector_heatmaps}.

\subsection{Risk factor exposures}
We estimate the K-factor DPLS model for various $K$ using the in-sample total $R^2$:

$$R^{2,in}_{total}=1 -\frac{\sum^{N_t, T}_{i=1,t=1}(r_{i,t}- \hat{r}_{i,t}(Z_{t-1}; \hat{\theta}_t ))^2}{\sum^{N_t,T}_{i=1,t=1} r^2_{i,t}},$$ 
and the out-of-sample total $R^2$:

$$R^{2,out}_{total}=1 -\frac{\sum^{N_t, T-1}_{i=1,t=1}(r_{i,t+1}- \hat{r}_{i,t+1}(Z_t; \hat{\theta}_t ))^2}{\sum^{N_t,T-1}_{i=1,t=1} r^2_{i,t+1}},$$ 

over all periods. These total $R^2s$ measure the extent to which the latent risk factors capture the realized riskiness in the panel of individual stocks either in-sample or out-of-sample respectively. It can be interpreted as the amount of explained variance of the asset returns.  If the intercept is set to zero in the DPLS model, each of these $R^2s$ represent the model's ability to describe risk compensation solely through exposure to the latent risk factors. Otherwise, they represent a risk-free compensation and, for each $K$, incremental risk compensation adjustments. 

Table \ref{tab:factors} compares the performance of three different statistical factor models, both in-sample and out-of-sample, across all $N_t$ stocks in every period $t$.   The motivation for showing the in-sample and out-of-sample performance is to assess the bias-variance trade-off.  In particular, a drop in performance from in-sample to out-of-sample is indicative of over-fitting whereas poor performance in-sample suggests under fitting. The table also compares the performance of an equally weighted portfolio which is long the top ten performing stocks, as predicted by each model. In that case the sum of $i$ collapses as there is just one sum of $t$ for the portfolio returns.

The first model we consider is the asymptotic PCA model (APCA) which is fitted to panel data rather than period-by-period.  This model is solely based on excess returns and ignores the factor data. Hence, the model can not be used for prediction. When applied to the cross-section of stocks, we observe a monotonic increase in the in-sample total $R^2$ with increasing  $K$. 
The performance based on the total $R^2$ of the portfolio returns is very high and increasing factors brings the in-sample performance to an almost perfect level. However we note that only a few statistical factors in the PCA already substantially captures the portfolio risk.

The PLS and DPLS are different to APCA in two important ways: (i)  neither PLS nor DPLS adjust for heteroschedasticity as does APCA; and (ii) APCA fits to panel data whereas PLS ad DPLS fit to period-by-period cross-sectional data. Hence in PCA, the latent factors are the same across all periods, whereas in PLS and DPLS, the number of factors vary in each period. 

The methodology for selecting the number of latent risk factors is also more complex in PLS and DPLS as both are supervised learning methods.  In each period, $t$, the optimal number of factors, $K^*_t$, is chosen for each method using all stock data according to the hyper-parameter tuning details described earlier in this section. Each model is then trained with the optimal number of factors in each period. This is important, for two reasons: (i) the comparison between PLS and DPLS must show the attribution of the model error to factors given that the model fits the data well otherwise $K^*$ conflates quality of model fit with prediction attribution against the number of latent factors given the best fit; and (ii) it becomes computationally intractable to retrain the DPLS model in each every period over all values of $K$.  

Once $K^*_t$ is estimated, we predict the asset returns using $K_t\in\{1,\dots, K^*_t\}$ --- both PLS and DPLS quite naturally perform this since averaging over the first $K_t$ columns of $\hat{U}_{t}$ is a step in the prediction. However, it is not possible to predict with more than $K^*_t$ factors.  This procedure is repeated for all periods and the $R^{2}_{total}(K)$ is shown in the table for in-sample and out-of-sample forecasting.  Note that, because $K^*_t$ varies between each period, larger values of $K$ will exhibit fewer entries in the  $R^{2}_{total}$ formula prior to summing over time.  For example, suppose there are two periods only, and $K^*_1=1$ and  $K^*_{2}=2$, then there will be two error estimations at $K=1$ but only one at $K=2$ in the sum. We choose $K_{max}=10$ heuristically in the table based on a minimum sample size of 100 over all stocks and periods at $K=K_{max}$ in order to avoid statistical noise dominating the estimations. In other words, there must be at least 100 error terms to be ``double summed'' over $i$ and $t$ in the total $R^2$.

Across all stocks, the PLS and DPLS model's in-sample performance is substantially higher than PCA and monotonically increases with K, capturing as much as 25\% of the variance of the asset's excess returns. At first, this result seems counter intuitive given that PCA optimizes explained variance whereas PLS optimizes explained covariance between $X$ and $Y$. However, the tuning of $K_t$ is by out-of-sample MSE. Moreover, the neural network in DPLS is fitted by minimizing a MSE based loss function. Both PLS and DPLS perform comparably well in-sample across all stocks.

The out-of-sample performance of PLS starts to sharply deteriorate after 4 factors whereas DPLS peaks at 10 factors with as much as 15\% explained variance and almost monotonically grows with the number of factors. Note that monotonicity of MSE in the number of factors is only guaranteed in any one period --- over multiple periods, there is no reason why the monotonicity should be preserved. It just so happens, here, that $K^*_t$ tends to be close to 10.  

Applied to the portfolio, PLS and DPLS perform poorly in-sample compared to PCA, however, DPLS is shown to outperform PLS significantly out-of-sample.  Out-of-sample explained variances of portfolio returns are only available for $K=1$ using PLS due to too few samples for the total $R^2$ calculation.  In other words, there are too few periods where $K^*_t>1$ to generate sufficient terms to evaluate the $R^2$ over. 

A key observation is that the variance of asset returns is typically explained by a few latent risk factors in the PLS model whereas the DPLS model attributes the variance to a greater number of factors. One plausible explanation for this difference might be model reduction in a linear model--- for any stock in a given period, a linear model is characterized by its intercept and slope. Thus a small number of systemic factors are needed to capture the cross-sectional variability of the intercept and slope. In contrast, the DPLS model is not functionally equivalent nor can it be reduced to such a low dimensional set of risk factors due to the non-linearity. Thus one might arrive at different conclusions about the number of important systemic risk factors driving returns based on the inclusion of non-linearity or not.

\begin{table}[ht]
    \centering
    \resizebox{\columnwidth}{!}{%
    \begin{tabular}{|l|ccc|ccc|ccc|ccc|}
    \toprule    
    \hline
    & \multicolumn{3}{c|}{$K$=1} & \multicolumn{3}{c}{$K$=2} & \multicolumn{3}{|c|}{$K$=3} & \multicolumn{3}{c|}{$K$=4}\\
 \hline
    & PCA & PLS & DPLS & PCA & PLS & DPLS & PCA & PLS & DPLS & PCA & PLS & DPLS\\
\hline 

    Stocks (in-sample) & 0.0370 & 0.2182 & 0.1951 & 0.0589 & 0.2702 & 0.2280 & 0.0788 & 0.2887 & 0.2355 & 0.0976 & 0.2979 & 0.2382
\\
    Stocks (out-of-sample) & NA & 0.0858 & 0.0571 & NA & 0.0873 & 0.0735 & NA & 0.0790 & 0.0978 & NA & 0.0838 & 0.0971
 \\
    Portfolio (in-sample) & 0.6759 & 0.8993 & 0.7434 & 0.8301 & 0.7710 & 0.7630 & 0.8961 & 0.8303 & 0.7579 & 0.9132 & 0.8038 & 0.7351
 \\
    Portfolio (out-of-sample) & NA & 0.5070 & 0.5272 & NA & NA & 0.6097 & NA & NA & 0.5071 & NA & NA & 0.4656 
\\
\hline
     & \multicolumn{3}{c}{$K$=5} & \multicolumn{3}{|c|}{$K$=6}  & \multicolumn{3}{c|}{$K$=7} & \multicolumn{3}{c|}{$K$=8}\\
 \hline
    & PCA & PLS & DPLS & PCA & PLS & DPLS & PCA & PLS & DPLS & PCA & PLS & DPLS\\
\hline 

Stocks (in-sample) & 0.1132 &  0.3043 & 0.2320 & 0.1757 & 0.3108 & 0.2355 & 0.2697 & 0.3114 & 0.2453 & 0.3470 & 0.3114 & 0.2523\\

 Stocks (out-of-sample) & NA & 0.0489 & 0.1082 & NA &  0.0401 & 0.1240 & NA  & 0.0077 & 0.1395 & NA & -0.0039 & 0.1441\\
 
 Portfolio (in-sample) & 0.9359 & 0.8695 & 0.7357 & 0.9718 & 0.8878 & 0.7666 & 0.9653 & 0.8385 & 0.7757 & 0.9691 & 0.8301 & 0.7889\\
 
 Portfolio (out-of-sample) & NA & NA & 0.5496 & NA & NA & 0.5807 & NA & NA & 0.6281 & NA & NA & 0.6517\\

\hline
    \bottomrule    
    \end{tabular}
    }
    \caption{\textit{A comparison of the in-sample and out-of-sample total $R^2$ for individual stocks and managed equally weighted long-only portfolios for various numbers of latent risk factors across different statistical factor models. DPLS is shown to yield higher out-of-sample performance compared to PLS.}}
    \label{tab:factors}
\end{table}

In addition to understanding how the latent risk factors explain the asset returns and portfolio risk, we can use DPLS to explain the effect of non-linearity on the in-sample model prediction. After all, the non-linear factor structure is learned using the neural network composed with the PLS loadings. 

Figure \ref{fig:attribution} shows the excess monthly returns of an equally weighted, long-only, portfolio of top performing assets, as predicted by the DPLS factor model. The figures shows the attribution of the in-sample predicted excess monthly portfolio returns by the top three latent risk factors, quadratic term, higher order terms (H.O.T.), intercept and remaining risk factors. In any period, the quadratic term is found by a Taylor expansion about zero in the DPLS model. The higher order terms are found by subtracting the first two terms of the Taylor expansion from the DPLS prediction. The numbers in the parentheses denote the time averaged component values.

In periods of market dislocation, as exemplified by the collapse of Lehman Brothers in October 2008, we observe substantial quadratic effects (red) in the construction of the signal which highlight their importance in capturing the outliers. In other periods, where the volatility is lower, the effect of non-linearity is much more marginal.

 \begin{figure}
\centering
\includegraphics[width=\textwidth]{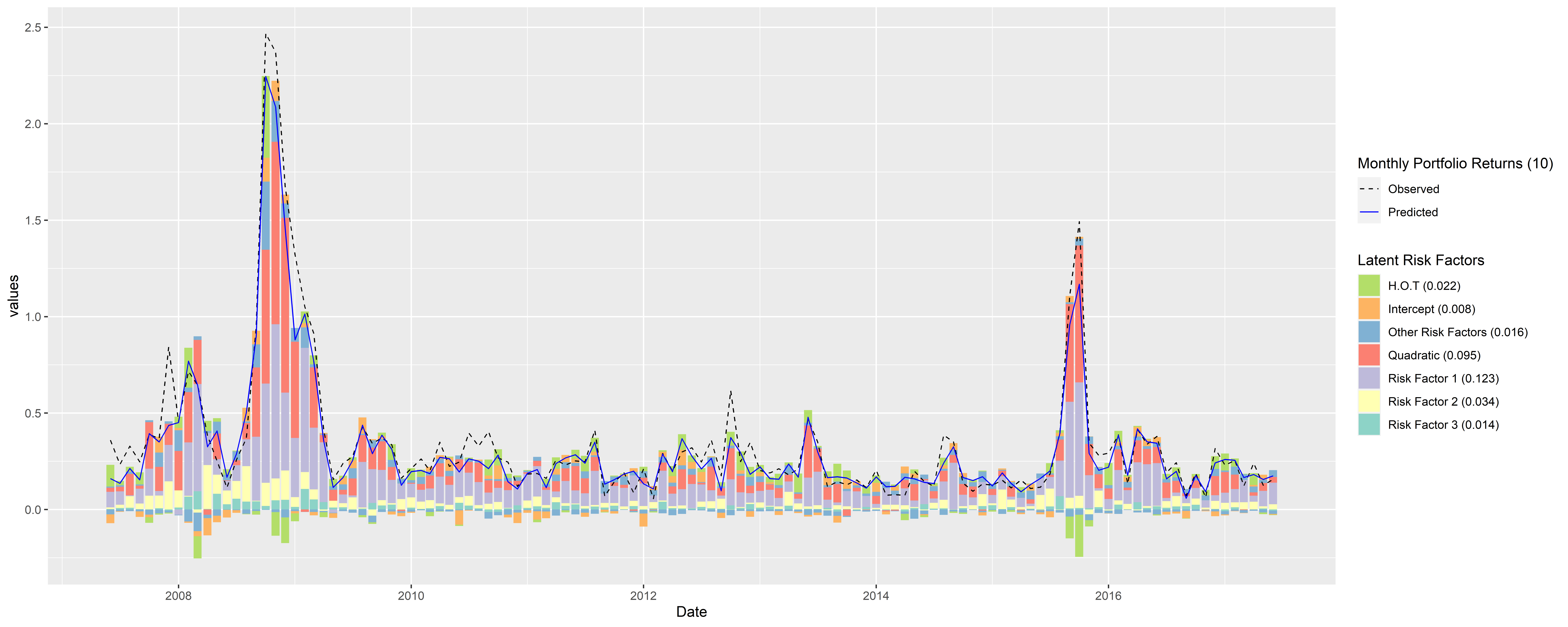} 
\caption{\textit{The attribution of the in-sample predicted excess monthly portfolio returns using the DPLS factor model. The predicted returns are attributed to the top three latent risk factors, quadratic term, higher order terms (H.O.T.), intercept and remaining risk factors. The numbers in the parentheses denote the time averaged component values.}}
\label{fig:attribution}
\end{figure}

\subsection{Computational considerations}
For completeness, we also comment on the computational advantages of using DPLS compared to NNs. Of course, PLS and NN require substantial computation to train in contrast with PLS or LASSO. 

Table \ref{tab:comp} below records the prediction performance and training times of the various methods, each cross-validated, for one period in the factor data. We see that the DPLS training times vary from approximately 2x up to 20x faster compared to a deep neural network and the number of parameters exceeds in one case more than 10x in reduction, providing substantial memory savings.
\begin{table}[!h]

\centering
\resizebox{\columnwidth}{!}{%
\begin{tabular}{|l|lll|ll|ll|l|}
\hline
Model Type& \#Predictors& \#Neurons per Layer& \#Parameters& MSE (in-sample)& MSE (out-of-sample)& $L_{\inf}$ Error (in-sample) & $L_{\inf}$ Error (out-of-sample)&Training Time (s)\\
\hline
DPLS (50\%)&14&100& 11701&0.022&0.033&0.908&0.963&3183.806\\
&&&&&&&&\\
DPLS (80\%)&28&50&4051&0.021&0.030&0.909&0.971&833.72\\
&&&&&&&&\\
DPLS (95\%)&37&100&14001&0.018&0.037&0.798&0.978&9870.077\\
&&&&&&&&\\
\hline
NN&49&200& 50401&0.019&0.028&0.798&1.065&19452.267\\
\hline
\end{tabular}
}
\caption{\textit{The prediction performance and training times of the various methods, each cross-validated, for one period in the factor data.}}
\label{tab:comp}
\end{table}

\section{Summary}\label{sect:summary}
This paper introduces a new approach to conditional latent factor modeling which combine PLS and deep learning to predict and explain stock returns. DPLS projects firm characteristics on to a smaller set of systemic dynamic risk factors by maximizing the covariance between characteristics and asset returns. Unlike other deep learning based factor models, DPLS factor models corresponds to a non-linear stochastic discount factor. Additionally the about of no-arbitrage can be estimated in the model over multiple periods and, if needed, the intercepts can be set to zero to enforce static no-arbitrage. 
DPLS is shown to identify the non-linear latent factor structure --- the convexity effects in the latent factor exposures can be quantified and used to estimate the extent to which linear risk premia fall short of explaining the risk reward trade-off.
Further more DPLS is shown to exhibit superior out-of-sample performance than OLS, LASSO, NNs, and PLS.

While DPLS is in its infancy, we comment on how it might gain adoption by portfolio managers. In our view, model risk is mitigated by running many models and establishing a consensus - this is a long established approach used in metereological forecasting for example.

The best use-case for the DPLS factor model is likely running it side-by-side a linear factor model and asking how the weights in the portfolios change under the inclusion of non-linearity. The non-linearity is likely to be a source of nuisance or a welcomed directional adjustment by portfolio managers as their exposure limits are adjusted up or down commensurately.

\bibliography{DPLS_paper} 

\clearpage
\appendix

\section{Factor Model Description} \label{sect:appendix}
\begin{table}[h!]
\caption{\textit{A short description of the factors used in the Russell $1000$ deep learning factor model.}}
\resizebox{\columnwidth}{!}{
\begin{tabular}{l|l|l}
  
\hline
ID& Symbol & \textbf{Value Factors}\\

\hline
1 & B/P & Book to Price\\

2 & CF/P & Cash Flow to Price\\

3& E/P & Earning to Price\\

4 & S/EV & Sales to Enterprise Value (EV). EV is given by \\
 && EV=Market Cap + LT Debt + max(ST Debt-Cash,0), \\
 & & where LT (ST) stands for long (short) term\\

5& EB/EV&   EBIDTA to EV \\

6& FE/P & Forecasted E/P. Forecast Earnings are calculated from Bloomberg earnings consensus estimates data. \\
& & For coverage reasons, Bloomberg uses the 1-year and 2-year forward earnings.\\

17& DIV & Dividend yield. The exposure to this factor is just the most recently announced annual net dividends\\ &&  divided by the market price. \\
&& Stocks with high dividend yields have high exposures to this factor.\\

\hline

& & \textbf{Size Factors}\\

\hline

8 & MC & Log (Market Capitalization)\\

9& S & Log (Sales)\\

10 & TA & Log (Total Assets)\\

\hline

& &  \textbf{Trading Activity Factors}\\

\hline

11& TrA & Trading Activity is a turnover based measure. \\
& & Bloomberg focuses on turnover which is trading volume normalized by shares outstanding. \\
&& This indirectly controls for the Size effect. \\
&&The exponential weighted average (EWMA) of the ratio of shares traded to shares outstanding: \\
&& In addition, to mitigate the impacts of those sharp shortlived spikes in trading volume, \\
&& Bloomberg winsorizes the data: \\
&& first daily trading volume data is compared to the long-term EWMA volume(180 day half-life), \\
&& then the data is capped at 3 standard deviations away from the EWMA average.
\\

\hline

&& \textbf{Earnings Variability Factors}  \\

\hline

12 &EaV/TA & Earnings Volatility to Total Assets. \\
&& Earnings Volatility is measured \\

 && over the last 5 years/Median Total Assets over the last 5 years\\

13 & CFV/TA & Cash Flow Volatility to Total Assets. \\
&& Cash Flow Volatility is measured over the last 5 years/Median Total Assets over the last 5 years\\

14 & SV/TA &  Sales Volatility to Total Assets. \\
&& Sales Volatility over the last 5 years/Median Total Assets over the last 5 year\\

\hline 

& & \textbf{Volatility Factors}\\

\hline

15 & RV& Rolling Volatility which is the return volatility over the latest 252 trading days\\

16 & CB & Rolling CAPM Beta which is the regression coefficient\\
& & from the rolling window regression of stock returns on local index returns\\

\hline

&& \textbf{Growth Factors}\\
\hline
7& TAG & Total Asset Growth is the 5-year average growth in Total Assets\\ && divided by the Average Total Assets over the last 5 years\\
18 & EG & Earnings Growth is the 5-year average growth in Earnings\\ &&  divided by the Average Total Assets over the last 5 years\\
\hline
& & \textbf{GICS sectorial codes}\\
\hline
19-24 & (I)ndustry & $\{10, 20, 30, 40, 50, 60, 70\}$ \\
25-35 & (S)ub-(I)ndustry & $\{10, 15, 20, 25, 30, 35, 40, 45, 50, 60, 70, 80\}$\\
36-45 & (SE)ctor & $\{10, 15, 20, 25, 30, 35, 40, 45, 50, 55, 60\}$\\
46-49 & (I)ndustry (G)roup & $\{10, 20, 30, 40, 50\}$\\
\hline
\end{tabular}
}
\label{tab:r3000}
\end{table}
\section{Feedforward Network Sensitivities} \label{sect:sensitivities}
In a linear regression model
\be
\hat{Y} =F_{\mathbf{\beta}}(X):= \beta_0 + \beta_1X_1 + \dots +\beta_K X_K,
\ee
the model sensitivities are $\partial_{X_i} \hat{Y}=\beta_i$. In a feedforward neural network, we can use the chain rule to obtain the model sensitivities 
\be
J:=\partial_{X} \hat{Y} =W^{(L)}J(I^{(L-1)})=W^{(L)}D(I^{(L-1)})W^{(L-1)}\dots D(I^{(1)})W^{(1)},
\ee
where $D_{ii}(I):=\act'(I_i),~D_{ij}=0,~i\neq j$ is a diagonal matrix.

\subsection{Feedforward network interaction effects}\label{sect:interaction}  The pairwise interaction effects are readily available by evaluating the elements of the Hessian matrix. For a L layer network, we define the $(i,j)^{(th)}$ element of the Hessian of the response, $\hat{Y}=G_{\hat{\theta}}(X)$, as 
\be
\partial^2_{X_iX_j} \hat{Y}=\sum_{\ell=1}^{L-1} H_{i,j,\ell},~ H_{i,j,\ell}:=W^{(L)}D^{(L-1)}W^{(L-1)}\dots \partial_{X_j} D^{(\ell)}W^{(\ell)}\dots w_i^{(1)}.
\ee
where it is assumed that the activation function is at least twice differentiable everywhere, e.g. $tanh(x)$, softplus etc.

\section{Proof of Non-linear Stochastic Discount Factors}\label{sect:appendix:proof}

\begin{proof}
Denote the regression coefficients $a=g(\mathbf{0})\in\mathbb{R}, \mathbf{b}=g'(\mathbf{0}) \in\mathbb{R}^K, G=g''(\mathbf{0}) \in \mathbb{R}^{K \times K}$ are given by the Taylor expansion (to second order):
$$g(f)= a + \mathbf{b}^\top(\mathbf{f}-\mathbb{E}[\mathbf{f}]) + \frac{1}{2}(\mathbf{f}-\mathbb{E}[\mathbf{f}])^\top G (\mathbf{f}-\mathbb{E}[\mathbf{f}])+\dots$$

$$1=\mathbb{E}[mR]=a \mathbb{E}[R] + cov(\mathbf{f}^\top R)\Sigma^{-1}\Sigma \mathbf{b} + \frac{1}{2}\mathbb{E}[(\mathbf{f}-\mathbb{E}[\mathbf{f}])^\top G(\mathbf{f}-\mathbb{E}[\mathbf{f}])R] + \dots$$

Applying the trace trick, then using the invariance of the trace operator under cyclic permutations, and then the linearity of the trace operator

\begin{eqnarray*}
\mathbb{E}[(\mathbf{f}-\mathbb{E}[\mathbf{f}])^\top G(\mathbf{f}-\mathbb{E}[\mathbf{f}])R]&=& \mathbb{E}\left[tr[(\mathbf{f}-\mathbb{E}[\mathbf{f}])^\top G(\mathbf{f}-\mathbb{E}[\mathbf{f}])]R\right]\\
&=&  \mathbb{E}\left[tr[G(\mathbf{f}]-\mathbb{E}[\mathbf{f}])(\mathbf{f}]-\mathbb{E}[\mathbf{f}]^\top)]R\right]\\
&=& tr\left[\mathbb{E}[G(\mathbf{f}-\mathbb{E}[\mathbf{f}])(\mathbf{f}-\mathbb{E}[\mathbf{f}]^\top)R]\right]\\
&=& tr\left[G \mathbb{E}[(\mathbf{f}-\mathbb{E}[\mathbf{f}])(\mathbf{f}-\mathbb{E}[\mathbf{f}]^\top)R]\right]
\end{eqnarray*}

We can then write:
$$1=\mathbb{E}[mR]=a \mathbb{E}[R] + cov(\mathbf{f}^\top R)\Sigma^{-1}\Sigma \mathbf{b} + \frac{1}{2} \mathbf{b}\Sigma\Sigma^{-1}\mathbf{b}^{-1} tr\left[G \mathbb{E}[(\mathbf{f}-\mathbb{E}[\mathbf{f}])(\mathbf{f}-\mathbb{E}[\mathbf{f}])^\top R]\right]\mathbf{b}^{-\top}\Sigma^{-1}\Sigma \mathbf{b} + \dots, $$
where $x^{-\top}\equiv[x^{-1}]^\top$.
This expression can be simplified to 
$$\mathbb{E}[R]=\alpha  + \mathbf{\mybeta}^\top\mathbf{\mylambda} + \frac{1}{2} \mathbf{\mylambda}^\top\Gamma\mathbf{\mylambda} + \dots,$$

where $\alpha=\frac{1}{a}, \mathbf{\mybeta}=cov(\mathbf{f}R)\Sigma^{-1}, \mathbf{\mylambda}=-\Sigma \frac{\mathbf{b}}{a},$ and $\Gamma=\Sigma^{-1}\mathbf{b}^{-1} tr\left[aG \mathbb{E}[(\mathbf{f}-\mathbb{E}[\mathbf{f}])(\mathbf{f}-\mathbb{E}[\mathbf{f}])^\top R]\right]\mathbf{b}^{-\top}\Sigma^{-1}$ and the regression coefficients $\alpha=f(\mathbf{0})\in\mathbb{R}, \mathbf{\mybeta}=h'(\mathbf{0}) \in\mathbb{R}^K, \Gamma=h''(\mathbf{0}) \in \mathbb{R}^{K \times K}$ are given by the Taylor expansion (to second order):
$$g(\mathbf{\mylambda})= a + \mathbf{\mybeta}^\top\mathbb{\mylambda} + \frac{1}{2}\mathbf{\mylambda}^\top\Gamma\mathbf{\mylambda} + \dots $$

\end{proof}

\section{Additional Results} 
\subsection{Factor tilt heatmaps}\label{sect:factor_heatmaps}

\begin{figure}[H]
\centering
\begin{tabular}{cc}
\includegraphics[width=0.48\textwidth]{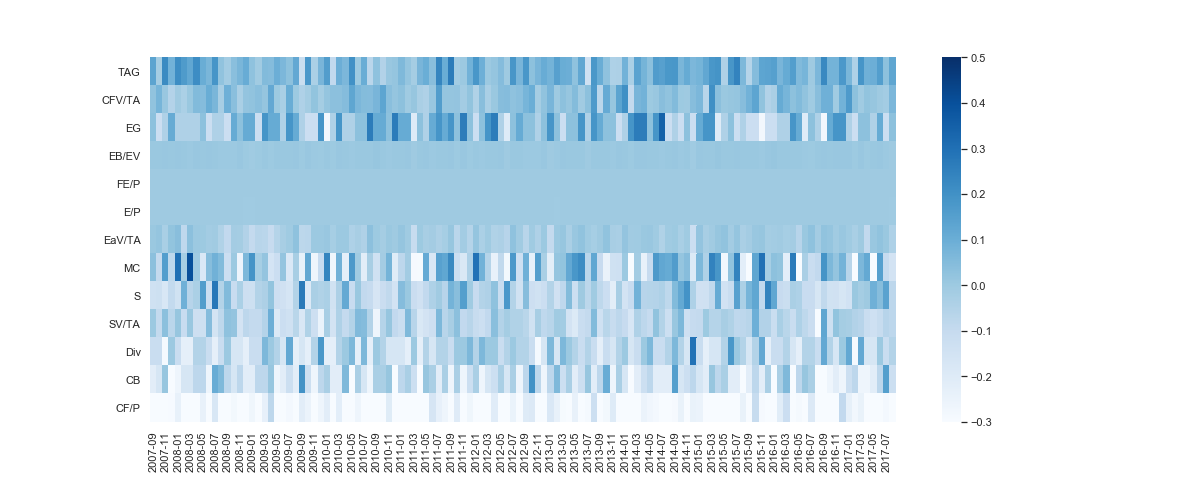} &
\includegraphics[width=0.48\textwidth]{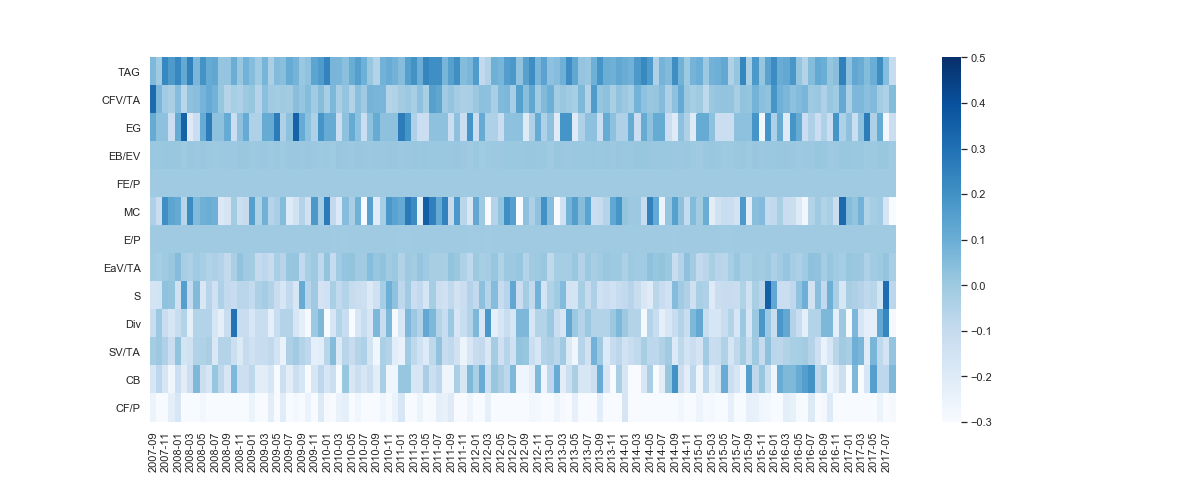}\\
(a)  DPLS(50\%) & (b)  NN\\
\includegraphics[width=0.48\textwidth]{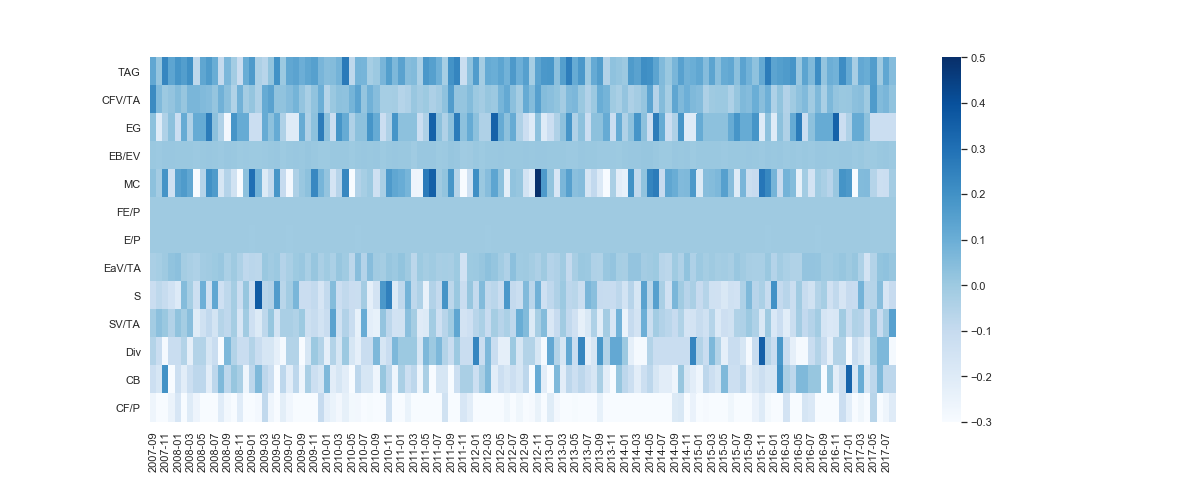} & 
\includegraphics[width=0.48\textwidth]{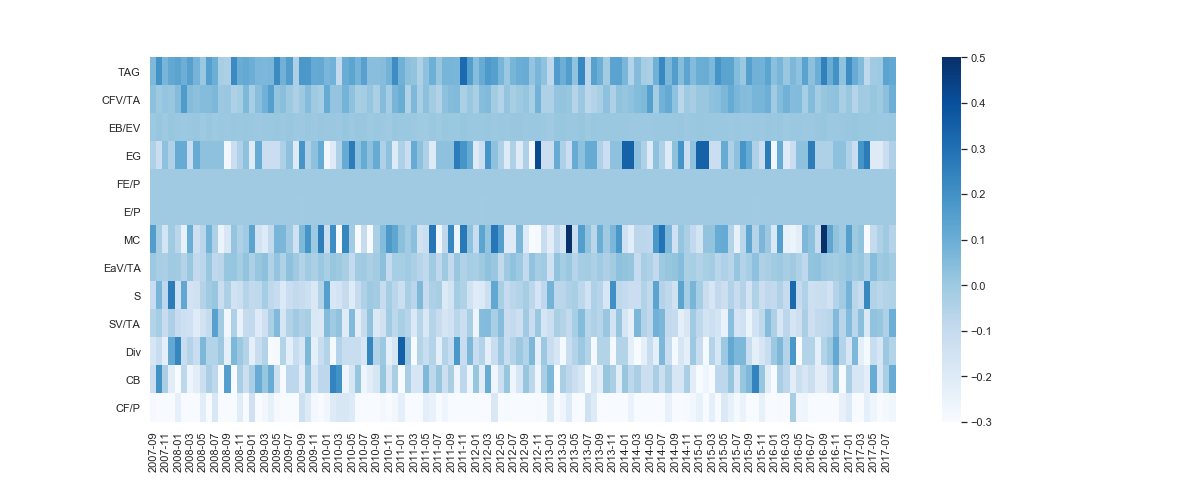}\\
(c)  LASSO & (d)  Random (WN)
\end{tabular}
\caption{\textit{The rescaled factor tilts of equally weighted portfolios constructed from the predicted top performing 50 stocks against time.}}
\label{fig:sector_tilts_heatmap}
\end{figure}

\subsection{Sector tilt heatmaps}\label{sect:sector_heatmaps}

\begin{figure}[H]
\centering
\begin{tabular}{cc}
\includegraphics[width=0.48\textwidth]{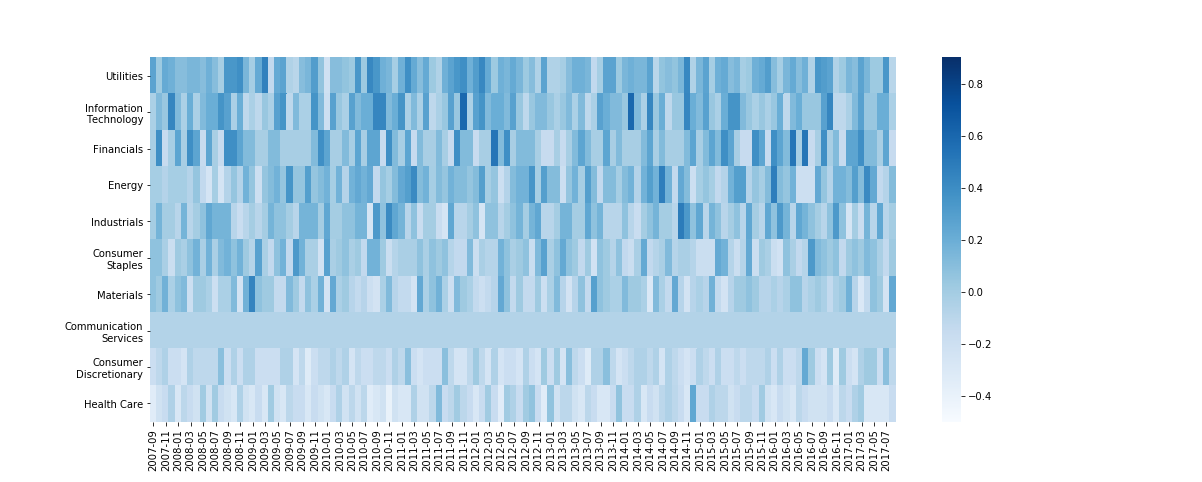} &
\includegraphics[width=0.48\textwidth]{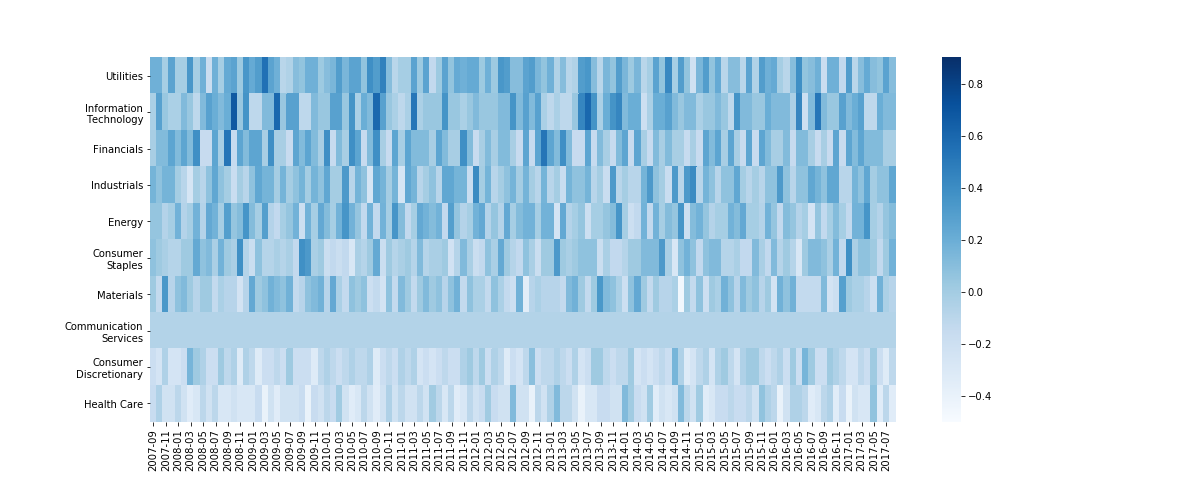}\\
(a)  DPLS(50\%) & (b)  NN\\
\includegraphics[width=0.48\textwidth]{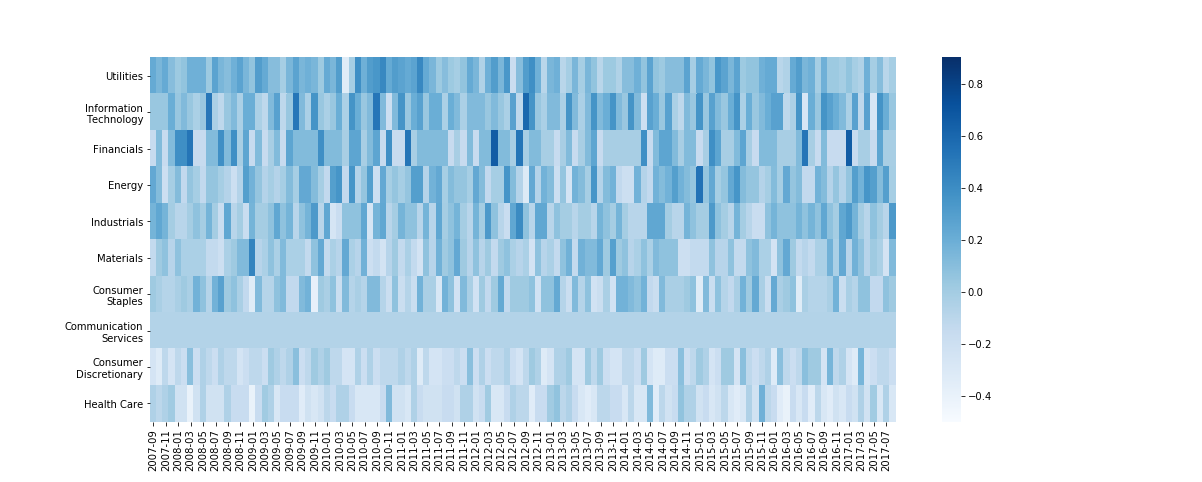} & 
\includegraphics[width=0.48\textwidth]{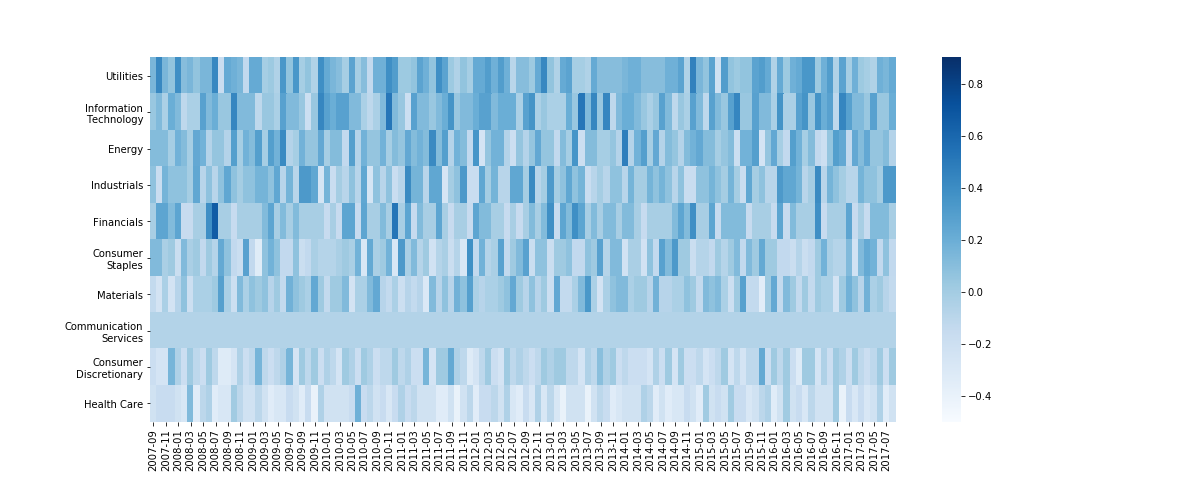}\\
(c)  LASSO & (d)  Random (WN)
\end{tabular}
\caption{\textit{The rescaled sector tilts of equally weighted portfolios constructed from the predicted top performing 50 stocks against time.}}
\label{fig:sector_tilts_heatmap}
\end{figure}

\end{document}